\documentclass[12pt]{article}
\usepackage[utf8]{inputenc}

\usepackage[affil-it]{authblk}
\usepackage{multirow}
\usepackage{comment}
\usepackage{lineno}
\usepackage{lipsum}
\usepackage{setspace}
\usepackage{titling}
\usepackage{subcaption}
\usepackage{algorithm}
\usepackage[noend]{algpseudocode}
\usepackage{epigraph}
\usepackage{graphicx}
\usepackage{color}
\usepackage{enumitem}
\usepackage{algorithm}
\usepackage[noend]{algpseudocode}
\usepackage{mdframed}
\usepackage{xr-hyper}
\usepackage[colorlinks,citecolor=northwestern-purple,urlcolor=chicago-maroon,linkcolor=chicago-maroon,backref=page]{hyperref}
\hypersetup{
    colorlinks,
    linkcolor={chicago-maroon},
    filecolor={cornell-red},
    urlcolor={northwestern-purple},
    citecolor={northwestern-purple}
}
\usepackage[round]{natbib}
\usepackage{amsmath}
\usepackage{amssymb}
\usepackage{amsthm}
\usepackage{subcaption}
\usepackage[lmargin=1.25in,rmargin=1.25in,tmargin=1in,bmargin=1in]{geometry}
\usepackage{etoolbox}
\usepackage{lipsum}
\usepackage{fullpage}
\usepackage[nameinlink]{cleveref}

\makeatletter
\newcommand*{\addFileDependency}[1]{
  \typeout{(#1)}
  \@addtofilelist{#1}
  \IfFileExists{#1}{}{\typeout{No file #1.}}
}
\makeatother

\setlist[enumerate]{leftmargin=1.5cm,rightmargin=0.5cm,noitemsep, topsep=2pt}

\definecolor{clemson-orange}{RGB}{234,106,32}
\definecolor{chicago-maroon}{RGB}{128,0,0}
\definecolor{northwestern-purple}{RGB}{82,0,99}
\definecolor{cornell-red}{RGB}{179,27,27}
\definecolor{sauder-green}{RGB}{171,180,0}
\definecolor{gray}{RGB}{192,192,192}
\definecolor{lawngreen}{RGB}{0,250,154}

\makeatletter
\def\BState{\State\hskip-\ALG@thistlm}
\makeatother
\allowdisplaybreaks

\newcommand{\Q}{{\bb Q}}

\newcommand{\E}{\bb E}
\newcommand{\I}{\bb I}

\DeclareMathOperator\supp{supp}

\DeclareMathOperator{\argmax}{arg\,max}

\theoremstyle{definition}
\newtheorem{theorem}{Theorem}
\newtheorem{lemma}{Lemma}
\newtheorem{fact}{Fact}
\newtheorem{corollary}{Corollary}
\newtheorem{proposition}{Proposition}

\newtheorem{definition}{Definition}
\newtheorem{remark}{Remark}

\newtheorem{step}{Step}
\theoremstyle{definition}

\makeatletter
\patchcmd{\@addmarginpar}{\ifodd\c@page}{\ifodd\c@page\@tempcnta\m@ne}{}{}
\makeatother
\crefname{assumption}{Assumption}{Assumptions}
\crefname{lemma}{Lemma}{Lemmas}
\crefname{theorem}{Theorem}{Theorems}
\crefname{corollary}{Corollary}{Corollaries}
\crefname{proposition}{Proposition}{Propositions}
\crefname{claim}{Claim}{Claims}
\crefname{procedure}{Procedure}{Procedures}
\crefname{algorithm}{Algorithm}{Algorithms}
\crefname{figure}{Figure}{Figures}
\crefname{remark}{Remark}{Remarks}
\crefname{section}{Section}{Sections}
\crefname{procedure}{Procedure}{Procedures}
\crefname{example}{Example}{Examples}
\crefname{definition}{Definition}{Definitions}
\crefname{table}{Table}{Tables}
\crefname{equation}{}{}
\crefname{enumi}{}{}
\crefname{conjecture}{Conjecture}{Conjectures}
\crefname{step}{Step}{Steps}

\def \m{\mu}
\def \l{\lambda}
\def \m{p}

\def \Q{\mathcal Q}
\def \I{\mathcal I}

\def \E{\mathbb{E}}

\def \I{\mathcal {I}}

\def \l{\lambda}

\def \Q{\mathcal Q}

\def \ll{\lower1.6truept\hbox{${{\scriptstyle =\atop \scriptstyle <}}$}}
\def \gl{\lower1.6truept\hbox{${{\scriptstyle >\atop \scriptstyle
=}\atop{\scriptstyle <}}$}}
\def \lg{\lower1.6truept\hbox{${{\scriptstyle <\atop \scriptstyle =}\atop
{\scriptstyle >}}$}}

\def\l{\lambda}
\def\tsum{\text{$\sum$}}

\input{tcilatex}
\thanksmarkseries{arabic}

\begin{document}

\title{\textbf{Optimal Queue Design\thanks{%
We are grateful to Ethan Che, Laurens Debo, Laura Doval, Drew Fudenberg,  Refael Hassin, Moshe Haviv, Yash Kanoria, Krishnamurthy
Iyer, Ioannis Karatzas, Jinwoo Kim, Jacob Leshno, Shengwu Li, Vahideh Manshadi,  Chiara Margaria, Afshin Nikzad, Chris
Ryan, Robert Shumsky, Eduardo Teixeira, and seminar participants at Chicago-Booth, Columbia DRO-IEOR, Harvard/MIT, NYU-Stern, Dartmouth-Tuck, NUS, North American Summer Meeting of ES (2023 Semi-plenary), KER Conference, ACM-EC, INFORMS, Hong-Kong virtual,  VSET, AMET, and Toulouse,  for their helpful comments. We
acknowledge research assistance from Will Grimme, Dong Woo Hahm, and Sara Shahanaghi. Yeon-Koo Che is
supported by the Ministry of Education of the Republic of Korea and the
National Research Foundation of Korea (NRF-2020S1A5A2A03043516) }}}
\author{Yeon-Koo Che\thanks{
Department of Economics, Columbia University, USA. Email: \href{mailto: yeonkooche@gmail.com}%
{\texttt{yeonkooche@gmail.com}}.} \, \and \, Olivier Tercieux\thanks{
Department of Economics, Paris School of Economics, France. \ Email: \href{mailto: tercieux@pse.ens.fr}%
{\texttt{tercieux@pse.ens.fr}}.}}
\date{\today \endgraf }
\maketitle

\begin{abstract}
We study the optimal method for rationing scarce resources  through a queue system. The designer controls agents'
\textit{entry} into a queue  and their \textit{exit}, their
\textit{service priority}---or \textit{queueing discipline}---as well as their \textit{information} about queue priorities, while providing them with the incentive    to \textit{join} the queue
and, importantly, to \textit{stay} in the queue, when recommended by the designer. Under a mild condition, the optimal mechanism induces agents to enter up to a certain queue length and never removes any agents from the queue; serves them  according to a
first-come-first-served (FCFS) rule; and provides them with no information
throughout the process beyond the recommendations they receive.  FCFS is also necessary for optimality in a rich domain. We
identify a novel role for queueing disciplines in regulating agents'
beliefs and their dynamic incentives, and uncover a hitherto unrecognized virtue of FCFS in this regard. \vskip0.2cm \noindent \textbf{%
JEL Classification Numbers}: C78, C61, D47, D83, D61 \newline
\end{abstract}

\setstretch{1.22} 

\section{Introduction}

As a method for allocating scarce resources, queueing, or ``waiting in line,'' remains as old and ubiquitous as its equally-celebrated brethren---market-clearing prices.  Unlike the price mechanism, however, queueing is time-consuming and imposes deadweight losses for the agents in the queue.  To this date, providing and managing the incentives to queue remains the fundamental challenge for businesses that must rely on queueing for providing goods and services.

In managing the queueing incentives,  real-world queues often deploy  several instruments.   First, they often control agents'  entry into the queue, and sometimes their exit. For instance, service call
centers sometimes encourage customers to wait in line (i.e., to be put on
hold); other times, presumably in the face of high call volume, they tell
customers to try another time. Some call centers ask customers  to leave the queue and return later.

Second, they decide how to prioritize service  among agents in the queue.  In this regard, \textit{first-come-first-served} (\textsf{FCFS}) is the oldest and
by far the most common queue discipline, but \textit{service-in-random-order}
(\textsf{SIRO}) which assigns priority at random, has been also used. Some authors
have proposed other rules such as \textit{last-come-first-served} (\textsf{LCFS})
(e.g., \cite{hassin1985}, \cite{su-zenios2004}, and \cite{platz2017}).

Finally, they can often control the information available to an agent, both when he arrives at
the queue and while he is in the queue.
Many call centers keep the customers completely in the dark about the queue
length, their relative positions, or their estimated waiting times. Similarly, many offices for social housing do not disclose any information on
positions on waiting lists.\footnote{%
This is the case, for instance, for several housing choice voucher programs
in California, e.g., \href{http://www.plumascdc.org/housingFAQ.html}{PCCDS
Housing Service} or \href{http://www.haca.net/applicants/current-wait-list-applicants/}%
{HACA} among others.}  Meanwhile, other  systems provide customers with some information. For instance, popular
ride-hailing apps provide a customer with not only the estimated arrival
time of a vehicle but also its current location on a map.

We ask: {\it how should the queue system be chosen along these three dimensions?} To ask this question, we consider a queueing model in which agents' arrival and  servicing follow general Markov processes.  As in the standard model (e.g., \cite{naor1969regulation}),  agents have homogeneous preferences; they realize some positive lump-sum surplus from service and incur linear costs from waiting until  the service concludes. Given these primitive processes, the designer chooses a queue system that is incentive compatible.  While our designer can keep an agent from joining the queue or remove one
from the queue, she cannot coerce an agent to enter the queue or to stay in the queue against his will. In other words, when recommended to either join or stay in a queue, an agent
must have the incentive to obey this recommendation given the information that he
has.  Subject to this incentive constraint, the designer maximizes a  weighted sum of the agents' welfare and the service provider's profit. Since the
weight is arbitrary, the designer could be a service provider who maximizes the profit, a consumer advocate who maximizes agents' welfare, or
a regulator who values both.


 The queue system, together with the primitive arrival and service process, induces a Markov chain on the length of the queue. Our analysis focuses on the steady state, or the invariant distribution, of this Markov chain.
Under a very mild \textit{regularity condition} on the
process, our answer is strikingly simple and consistent with many observed
practices of queue design.
(i) The optimal queue design has a \emph{cutoff} policy: namely,
there exists a maximal queue length $K\ge 0$ such that agents are
recommended to enter the queue if and only if its length is less than $K$.%
\footnote{\label{fn:eq class} {When the queue length is $K-1$, an agent is
recommended to enter with a positive probability possibly equal to one. If
this probability is less than one, the entry is ``rationed'' at $K-1$.
}}
(ii) Those who join the queue are then prioritized to receive a
service according to FCFS. (iii) No information is provided to
agents beyond the recommendations they receive to join or to stay in the
queue.\footnote{%
Since recommendations contain information about the state, this policy
should not be confused with ``no information'' authors often use, which
refers to ``no communication'' whatsoever. Agents can make Bayesian
inferences on their expected waiting times, based on the recommendation they
receive, the queue design that the designer commits to, and the elapsed time
after joining the queue.}

 Result (i) (shown in \cref{sec:cutoff}) means that one can achieve
an optimal queue design, without removing agents or incentivizing them to
leave the queue once they join the queue.
Reneging---or abandonment of the queue---is then never part of
our optimal queue behavior.\footnote{Removal of agents can only be consistent with optimality if it occurs when the queue is
full or near full.} Results (ii) and (iii) (both shown in \cref{sec:FCFS}) mean that, at least in the canonical model we
consider, the most tried-and-true queueing norm is (at least weakly) better
than any others, provided that agents receive no information beyond
the recommendations from the designer.

 The optimal design we identify is consistent with  many commonly-observed queue practices. The cutoff policy conforms to the standard
practice of capping the queue length at some level (e.g., offices for social housing often cap waiting lists when they are too long). The optimality of FCFS accords well with its prevalent use in practice.   The \textit{no
information beyond recommendation} policy also conforms to standard practice in call centers which often put customers on hold with little or no information. Similarly, as we already pointed out, offices for social housing often provide applicants with very limited information on their position on the list. Offering a rough estimate on the waiting time, another common practice, is also consistent with our policy, which can be implemented via two estimates, a short estimate that encourages entry and a long estimate that discourages entry.

The simplicity of our optimal design, particularly the optimality of FCFS, contrasts with the existing literature which finds it suboptimal (see our literature review).  As we explain below, these earlier findings can be traced to some aspects of queue design, particularly the information policy, being exogenously fixed in a suboptimal manner.  Allowing for {\it all} aspects of queue design to be chosen optimally leads us to find FCFS optimal.  This finding is   reassuring in light of the perceived fairness of FCFS (see \cite{larson1987}).  According to the common perception,  ``...the universally acknowledged standard is first-come-first-served: any deviation is, to most, a mark of iniquity and can lead to violent queue rage''  (``Why Waiting is Torture,'' Alex Grey, {\it New York Times}, Aug 18, 2012).


The intuition behind the information policy---no information beyond
recommendation---is explained as follows.
It is well known and intuitive that incentive constraints are relaxed most
when agents are given as little information as possible. If an agent has the
incentive to join or to stay in a queue for a set of signals, he must also
have the same incentives when all these signals are pooled into one,
regardless of the queueing discipline. Since this ``pooled'' signal is
precisely what the agent will have given ``no information'' beyond the
recommendation, the no information policy is optimal.

To explain why FCFS is optimal, fix an optimal entry and exit policy---i.e.,
a cutoff policy with some maximal length $K$. Assuming agents obey the
recommendation, this induces a distribution of queue length in the steady
state. Since our agents are homogeneous, the expected waiting time \textit{%
when averaged across possible initial queue lengths} is the same for each
agent, and does not depend on the queueing discipline in use. Then, given no
information, the incentive for joining the queue will be the same across all
queueing disciplines, 
and on this account, FCFS is not particularly necessary or desirable.

However, the dynamic incentives that agents face---their incentive 
to ``continue'' queueing once they join the queue---differ across
queueing disciplines, assuming the no information policy.
The reason is that the distribution of waiting times differs across queueing
disciplines, so one updates beliefs about the remaining waiting times
differently as time passes under different queueing disciplines.
Our main insight is that, under the \emph{regularity} condition on the
primitive process, the evolution of these beliefs become progressively more
favorable under FCFS. Consequently, under the condition, agents are willing
to stay in the queue under FCFS with no information, thus implementing the
optimal queueing outcome.


The progressively improving beliefs under FCFS stem from its fundamental
property: namely, that one's service priority can only improve over time
under FCFS. 
Hence, starting with any initial queue length, the elapse of time is indeed
\textit{good news} about the remaining waiting time. 
But there is also a countervailing force. Since an agent is not told about
the queue length $k$ when he joins the queue (recall that agents get no information beyond the designer's recommendations), his belief about this will be also updated as time progresses. On this
account, the elapse of time is actually \emph{bad news}, since it indicates
that the agent likely underestimated the initial length of the queue when he
joined it. We show that the good news dominates the bad news under the
regularity condition. As noted above, this means that incentive
compatibility is maintained once an agent is willing to join the queue under FCFS.

The belief evolution is not as favorable for other queueing disciplines,
however. Consider SIRO. Since priority is assigned randomly, one's queue
position does not matter; instead, his belief about the current queue length
is what matters for his incentives: the more agents there are in the queue, the less
likely it is for an agent to receive service.
Hence, the passage of time
(without being served) is a signal that there are more agents in the queue than
he initially thought. Further, unlike FCFS, his priority does not improve over time. So, the agent becomes more pessimistic as
time passes. Indeed, we can find simple examples
in which an agent's belief worsens over time to such a degree that he  leaves the queue in the midstream, thus undermining the implementation of the optimal cutoff policy.

While the optimality of FCFS does not preclude the possibility that another queueing rule may  be also optimal, we establish the sense in which the FCFS is uniquely best in dealing with
the dynamic incentives problem. In \cref{sec:necessity}, we show that for \textit{any} queueing discipline differing from FCFS, there exists a (regular) environment under which it
is strictly suboptimal no matter the information policy adopted. That is, FCFS does not just attain the optimal outcome under the no-information policy, but its use is also \emph{necessary} to achieve optimality in a rich domain.

The reason for this can be traced to the fairness property of FCFS:  among all queueing rules, the distribution of  wait times is least {\it dispersed} under FCFS, meaning both unusually short waits and unusually long waits are rare under FCFS (\cite{Shanthikumar1987}).  By contrast, other rules, such as  LCFS, induce more dispersed wait times, making more probable both lucky early breaks and unlucky long delays. Such a dispersion is bad for {\it conditional} belief about one's residual waiting time and his dynamic incentives.  As time passes, the fact that one {\it still remains in the queue} indicates that he has ``missed the early breaks'' and therefore  the residual wait   will be longer.  The fairness property of FCFS alleviates this problem.  To the best of our knowledge, we are the first to connect the distributional fairness of the queueing rules  with the agents' dynamic incentives and identify the crucial role it plays in the optimal queue design.


%
\smallskip

\indent\textbf{Related Literature.} 
The current paper follows the long line of queueing theory research, in
particular, the \emph{rational queueing} literature---which
has developed into a significant body of work since the seminal work by \cite%
{naor1969regulation}---studies the strategic behavior of rational Bayesian
agents in a variety of queueing scenarios.\footnote{%
See \cite{hassin-haviv2003} and \cite{hassin2016}, for an excellent survey
of the literature.} While sharing their focus and approach, the current
paper is distinguished from standard works by the generality and comprehensiveness of the queue designs, designer objectives, primitive processes, as well as agents' queue incentives we consider.

The existing literature typically studies one aspect of design such as the queueing discipline, while taking other aspects such as entry/exit or information policies as exogenously given.  In particular, exising papers show FCFS to be suboptimal in a variety of environments.

For instance,   \cite%
{naor1969regulation} finds that FCFS produces excessive incentives for agents to queue, due to the ``congestion'' externality they face under FCFS.\footnote{Plainly, under FCFS agents ignore the delay their joining the queue causes for the agents who will arrive later.}  \cite{hassin1985} and \cite{su-zenios2004} argue that LCFS can ``cure'' this externality and is thus optimal for a designer who maximizes consumer welfare.\footnote{\cite%
{platz2017} find a similar result when there are a continuum of agents who
enter at their endogenously chosen times. See also \cite{haviv-oz2016} for
alternative schemes in the observable environment and \cite{haviv-oz2018}
for extensions to the unobservable queue environment.} But, this literature assumes that agents fully observe the queue length upon arrival and the designer can't control their entry into the queue.  Indeed, the negative externality problem  can be easily fixed, and optimality achieved, under FCFS  if  entry is controlled, as in our optimal cutoff policy.

Meanwhile,  FCFS may give too few
incentives if the designer maximizes (or is close to maximizing) the service
provider's profit or his service utilization, or there is an excessive supply
of agents as in the case of \cite{leshno2019dynamic}.  Then, other mechanisms such as SIRO were shown to outperform FCFS by providing greater incentives for queueing. But this conclusion rests crucially on agents having full information about the queue
length.  The result does not hold  if  the designer can control the agents' information; in fact, it can be drastically overturned if agents can freely leave the queue,  an issue that the existing literature largely ignores.\footnote{A few papers
consider incentives by agents to abandon a queue, or to ``renege''; see \cite%
{hassin-haviv1995impatient}, \cite{haviv-ritov2001renege}, \cite%
{mandelbaum2000}, \cite{sherzer-kerner2018}, and \cite{cripps-thomas2019}.
However, their approach is  positive rather than normative; they seek to explain reneging as an equilibrium phenomenon arising from nonlinear waiting costs or aggregate uncertainty, rather than as an incentive constraint to be controlled in an optimal mechanism. }  Of course, there are important settings in which  not all design instruments, particularly information, can be controlled by the designer; our results do not apply to them.\footnote{In many ``physical'' queue settings (such as grocery check-out lanes), the length of  the queue is visible, so  the scope for information design is limited. Even in this case, our theory offers some useful insight: organizing a single serpentine line (as is done by {\it Trader's Joe}) is better than organizing multiple parallel lines.  The former admits less variance in wait times; this is not only fairer to the customers but more importantly reduces their incentives to leave the queue in the midstream. } The reader should therefore view the alternative works as  complementing one another.

 Indeed, we show that FCFS is \textit{always} optimal  regardless of the designer's objective, {\it provided that she can also control the entry of the agents and their information optimally.}\footnote{Several papers study alternative queueing disciplines in environments that
are less related or comparable to ours. FCFS is shown to be optimal in \cite%
{bloch-cantala2016} and a part of the optimal design in \cite{margaria2020}
in models where, unlike the standard queueing model, the lengths of queues
are non-stochastic, either because arrival occurs only when an agent exits
(the former) or because there is a continuum of agents (the latter).
Further, they do not consider information design, so the reason for the
optimality of FCFS is completely different in these models than in our model.
\cite{Kittsteiner2005} consider the allocation of priority in queues via
bidding mechanisms where processing time is private information. The crucial
difference is the use of transfers implicit in bidding mechanisms, which is
not allowed in our model.}  Further, FCFS is uniquely optimal and strictly dominates the other rules, if agents cannot be prevented from leaving the queue. In particular, any rule departing from FCFS such as LCFS and SIRO is likely to run afoul of this issue, as the fear of losing priority grows large with the elapse of time on the queue and convinces them to abandon the queue.

Finally, our paper is related to the burgeoning literature in queueing that considers  information design; see \cite{simhon2016optimal}, \cite%
{hassin-koshman2017}, \cite{lingenbrink2019optimal}, \cite{Kamenica2017}, and \cite%
{Anunrojwong2020}.\footnote{%
In a less related model, \cite{Ashlagi2020} study optimal dynamic matching
with information design, showing that FCFS, together with an information
disclosure scheme, can be used to implement the optimal outcome. Although
similar at first glance, their model is quite different from, and not easily
comparable to, ours. There is a continuum of agents in their model, and
their information policy pertains to the quality of goods rather than to
agents' queue position. In particular, the virtue of FCFS in regulating
agents' beliefs on where they stand in the queue is orthogonal to \cite%
{Ashlagi2020}'s insights.}
While the latter two papers identify the same optimal information design as the current paper, they do not study the optimal queueing discipline but
they instead take FCFS as given.  All of them also ignore the dynamic incentives issue, a crucial necessary condition for FCFS to be uniquely optimal.
\section{Model and Preliminaries}

\label{sec:model}

We consider a generalization of a canonical queueing model (e.g., \cite{naor1969regulation}) in which agents
arrive sequentially at a queue to receive a service. Time indexed by $t\in
\mathbb{R}_+$ is continuous.

\paragraph{Agents' payoffs.} 
There are three parties: a \textit{designer}, who organizes resource
allocation including the queueing policy, a \textit{service provider} who
services agents, and \textit{agents} who receive service. As will be seen,
the designer may be the service provider, a representative of the agents, or
a planner who reflects the welfare of both parties.

The agents are homogeneous in their preferences. Each agent enjoys a payoff
of
\begin{equation*}
U(t)\triangleq V- C\cdot t,
\end{equation*}
if she receives service after waiting $t\ge 0$ time period, where $V>0 $ is
the net surplus from service (possibly after paying a service fee to
the provider) and $C>0$ is a per-period cost of waiting. The service
provider earns profit $R>0$ for each agent she serves. In a customer
service context, the profit may not take the form of monetary fees collected from customers but rather the shadow value of fulfilling a warranty service or more generally addressing any customer needs
(See \Cref{sec:conclusion} for a discussion of an endogenously set monetary fee collected from customers). The designer's objective (to be specified below) is a weighted sum of the service
provider's and agents' payoffs. An agent's outside option, which she
collects when not joining the queue or exiting one, yields zero payoffs.

\paragraph{Primitive process.} At each instant, given the number of
agents in the queue, or \textsf{queue length}, $k\in \mathbb{Z}_+$, an agent
arrives at a Poisson rate of $\lambda_k\ge 0$. The technology allows for an agent to be served at each instant at the Poisson rate of $\mu _k>0$.\footnote{Different interpretations apply to different settings. In the service scenario (imagine a call center or  in a Apple repair center), multiple servers are serving customers simultaneously, but each takes a stochastic amount of time for completion; the service time for the first to be completed is then distributed exponentially with mean $1/\mu_k$.  In the housing assignment context, a housing becomes available at the Poisson rate $\mu_k$. }
Hence, a pair $(\lambda,\mu )$, where $\lambda \triangleq \{\lambda_k\}$ and $\mu
\triangleq\{\mu _k\}$, $\mu_0=0$, and $\l_0>0$, specifies a \textsf{primitive process}.
We view $%
(\lambda,\mu )$ as arrival and service rates that arise in many queueing
environments of interest, including $M/M/c$ queue models and
dynamic matching models, as illustrated in \Cref{sec:applications}; for
instance, the possibility of arrival and service rates depending on the
current queue length $k$ emerges naturally from a dynamic matching context.

We interpret $\mu_j$ as the maximal service rate
that \textit{any} set of $j$ or fewer agents may receive in any queue of length $k\ge j$.  It
is then natural to assume that $\mu_k$ is nondecreasing in $k$.\footnote{See \Cref{online_sec:axiomatic} in the online appendix for further details.}
 We also assume that $\mu_k$ is bounded uniformly in $k$.  
In addition, our
results invoke one of the following conditions:

\begin{definition} \label{def: regular}
(i) The \textbf{service process} $\mu =\{\mu _k\}$ is \textsf{regular} if $%
\mu _k-\mu _{k-1}$ is nonincreasing in $k$. (ii) The \textbf{primitive
process} $(\lambda,\mu )$ is \textbf{regular} if the service process $\mu $
is regular and $\lambda_k-\lambda_{k-1}\le \mu _k-\mu _{k-1}$ for each $k
\geq 2$.
\end{definition}

These two regularity conditions are extremely mild.
\Cref{sec:applications} shows that  all the canonical queueing models, as well as
dynamic matching models,   satisfy these two conditions.\footnote{%
In particular, as shown in the online appendix \Cref{online_sec:axiomatic}, the regularity of the service process,
namely, (i), has a desirable axiomatic foundation.}

\paragraph{Designer's policy.} The designer has a
number of instruments at her disposal. We focus on an anonymous stationary Markovian policy
that treats all agents identically based on two \textbf{state}
variables: the queue length $k$ and the queue position $\ell$, namely the
arrival order of an agent among those in a queue. The stationarity
restriction means that the policy does not depend on the calendar time.
The designer chooses the following set of policies. \smallskip

\textbf{$\bullet$ Entry and exit rule:}
The entry and exit rules specify how the designer regulates the entry of agents
who arrive at a queue and exit from those who are already in the queue.
Formally, an \textbf{entry rule} is given by $x= (x_k) $, where $x_k\in
[0,1] $ denotes the probability that an arriving agent is asked to join a
queue of length $k$. An \textbf{exit rule} is given by $(y,z) =(y_{k,\ell}
,z_{k,\ell})_{k,\ell}$. The designer removes the agent with queue position $\ell$ from the queue of length $k\ge \ell$ at a Poisson rate $y_{k,\ell}\ge 0$. In addition, upon a new arrival in the queue, the designer can keep the queue length constant by
removing an agent currently in the queue: $z_{k,\ell}\in [0,1]$ denotes the probability that
an agent with queue position $\ell$ is removed from a queue of length $k$
when another agent is joining the queue (where $k$ is the length of the queue before the new arrival).\footnote{By definition, if an agent $\ell$ is removed, no other agent $\ell' \neq \ell$ is removed.} The entry rule could
accommodate the possibility of non-entry that is either involuntary or voluntary.
Similarly, the exit rules $y$ and $z$ capture both the explicit policy of removing some agent away from a service pool (e.g., \cite{mandelbaum2000})
as well as the abandonment induced by a queueing policy (to be described
below). The main difference between $y$ and $z$ pertains to whether the
removal is conditional on the entry of another agent. In particular, $z$
captures the possibility of an agent being ``preempted" by a
new arrival, e.g., under an LCFS rule (see \cite{hassin1985}). We let $(\mathcal{X}%
, \mathcal{Y}, \mathcal{Z})$ denote the set of all feasible $(x,y, z)$'s.
\smallskip

 \textbf{$\bullet$ Queueing rule:}
A queueing rule specifies the allocation of service priority  among
agents in the queue.  Although we can accommodate any arbitrary  policy in this regard, for expositional ease, here we restrict attention to a ``Markovian" policy that depends on the queue length $k$ and the agent's queue position $\ell\le k$, or her arrival order, at any point.\footnote{There are two reasons for this restriction.  First, the current restriction makes the queueing rule more easily interpretable with respect to the standard queueing disciplines than the general class described in \Cref{app-sec: general}. Second, even the restricted class of queueing rules is quite broad and encompasses any standard service allocation rule.}
A queueing rule specifies the allocation of an available service rate
based on the queue length and agents' queue positions.%
\footnote{%
In fact, we can allow queueing rules to be fully general, i.e., without
limiting ourselves to those that depend only on $(k,\ell )$; examples
include rules that allow service probabilities to vary with time and to
depend on the history leading up to the current queue length and positions.
However, our class entails no loss since the optimal rule in this fully
general class belongs to the current class that we focus on.} Formally, a \textsf{%
queueing rule} is given by $q= (q_{k,\ell})$, where $q_{k,\ell }\geq 0$ is
the Poisson rate at which an agent receives service when the queue length is
$k$ and her position in the queue is $\ell$. Feasibility requires that
$\sum_{ \ell\in S} q_{k,\ell}\le \mu_{|S|}$, for all $k$ and all $S\subset \{1,..., k\}$: that is, the total service rate received by  a subset of agents in the queue cannot exceed the service rate available for the number of those agents.\footnote{ Recall that we interpret $\mu_{j}$ as the maximal rate at which a set of $j$ (or fewer) agents in the queue can be served collectively. Hence, the feasibility condition simply requires that any subset of agents of size $j$ must be collectively served at a rate no greater than this maximal service rate $\mu_{j}$. For instance, in the $M/M/c$ queue model, there are $c$ servers each able to serve an agent at rate, say $\mu$. Then, any  $j $ agents can be served at most at rate $\mu_j= \min\{j, c\} \mu$ in total.}   As is standard, we also require a feasible queueing rule to
be \textsf{work conserving}: $\sum_{\ell =1}^{k}q_{k,\ell }=\mu _{k}$, for
all queue length $k$. This means that the allocation of service is ``non-wasteful,'' or
exhausts the available service capacity. We let $\mathcal{Q}$ denote the set
of all work-conserving queueing rules.  The set $\mathcal{Q}$ encompasses all
standard queueing disciplines. For instance, assuming the service process is
regular, \textbf{first-come-first-served (FCFS)} satisfies $q_{k,\ell
}\triangleq \mu_{\ell }-\mu _{\ell -1}$. Namely, the agent in position 1
enjoys the highest possible service rate $\mu_1$ for any single agent; given
this, the agent in position 2 receives the highest possible service rate, $%
\mu_2-\mu_1 \geq 0$, and so on. The regularity condition guarantees the
service rate can only fall as one's position gets worse. (We will see in %
\Cref{sec:applications} how this corresponds to more familiar expressions in
the canonical queuing models such as $M/M/1$, $M/M/c$, or dynamic matching
models.) Similarly,  \textbf{last-come-first-served (LCFS)} satisfies $%
q_{k,\ell }\triangleq \mu _{k-\ell +1}-\mu _{k-\ell }$, and \textbf{%
service-in-random-order (SIRO)} satisfies $q_{k,\ell }\triangleq \mu_{k}/k$,
for all $k\in \mathbb{N}, \ell \le k$.\footnote{In online appendix \Cref{online_sec:axiomatic}, we provide a definition of FCFS based on the concept that the
priority must be assigned greedily to maximize the service rates for earlier arriving agents.  If the class of allocation rules  satisfies feasibility and the service process is regular, it is shown that FCFS indeed corresponds to our formula. In addition, under regularity,  we show that these standard queueing
disciplines  (FCFS, LCFS, and SIRO) are work-conserving. Conversely, the regularity property is
necessary if one requires FCFS and LCFS to be work-conserving.} Our results remain valid beyond the class $\Q$, in fact, for any arbitrary work-conserving rules; see \Cref{app-sec: general}.

\smallskip

 \textbf{$\bullet$ Information rule:}
An information rule specifies the payoff-relevant information given to an agent in the queue after each time $t\ge 0$
he has spent in the queue, including $t=0$ when he has just arrived at the
queue.
Since an agent has a linear waiting cost, the only payoff-relevant
information at each time  $t\ge 0$ spent on the queue is the probability $\sigma^t\in[0,1]$ that he
will be eventually served and the expected remaining waiting time $\tau^t\in
[0,\infty]$.\footnote{%
The waiting time refers to the duration of time an agent spends in the
queue, including the service time. In the queueing literature, this is
sometimes referred to as \textit{sojourn time}. Since the waiting cost is linear, the waiting time distribution matters only through its expectation.} Given the memoryless nature
of the process $(\lambda, \mu, x,y,z, q)$, these two variables depend only on the
current queue length $k$ and one's queue position $\ell\le k$ and are independent
of the time $t$ one has spent in the queue, so we write $(\sigma_{k,\ell},
\tau_{k,\ell})\in[0,1]\times [0,\infty]$ for each $(k,\ell)$. An agent's
(payoff-relevant) information then boils down to his information regarding $%
(k,\ell)$ at each time $t\ge 0$. As is well-known, say from \cite{KG2011},
this information can be represented as a distribution of \textquotedblleft
posterior beliefs\textquotedblright\ about $(k,\ell)$, which does, in general,
depend on time in the queue $t\ge 0$.

Formally, an \textbf{information rule}
is given by $I= (I^{t})_{t\in \mathbb{R}_{+}}$, where $I^{t}\in
\Delta( \Delta(\mathbb{Z}_+\times \mathbb{N}))$ specifies the distribution of
posterior beliefs on $(k^t, \ell^t)$ conditional on the time-on-the queue $t$.\footnote{Note that the process $(I^{t})_{t\in \mathbb{R}_{+}}$ does not form a martingale since the belief distributions are conditional on staying in the queue.} Feasibility requires that
posterior beliefs at each $t$ must be adapted to the filtration generated by
the process $(\lambda, \mu, x,y,z,q)$ and must satisfy Bayes rule given his
prior belief and knowledge of the process $(\lambda, \mu, x,y,z,q)$.  The agents' prior belief is given by the steady state distribution of the stochastic process induced by the entry and exit rule (see next paragraphs).\footnote{This is formally justified by the PASTA property (\cite{wolff1982}). One can think of an agent's {\it unconditional} (i.e., before conditioning on her arrival or on recommendations) belief about the state  as  given by the invariant distribution over states.}
Let $\mathcal{I}$
denote the set of all feasible information rules. (We suppress the
dependence both of $(\sigma_{k,\ell}, \tau_{k,\ell})$ and $\mathcal{I}$ on $%
(\lambda, \mu, x,y,z, q)$ for notational ease.)

The set $\mathcal{I}$ is large enough to include all realistic information rules, particularly given the Markovian queueing rule $q$. Special cases include \textbf{full information}, in which case
$I^{t}$ coincides with the true distribution of $(k^t, \ell^t)$, and \textbf{%
no information}, in which case the posterior $I^{t}$ is degenerate on the
belief obtained by Bayes updating via $(\lambda, \mu, x,y,z, q)$ from the prior beliefs $I_{0}$.   We allow for many other rules between the two.
For instance, the designer may simply reveal whether, upon joining the queue,
the agent's expected waiting time
is below or above some predetermined threshold.\footnote{ For many queueing rules (e.g., FCFS), this will mean specifying whether the agent's position is above or below a certain predetermined
integer $L$. Formally, $I^{0}$ will put weight only on two possible posterior beliefs, one with support in
$\{1,\dots,L\}$, the other one with support in $\{L+1,\dots,K\}$. 
}
As we show in \Cref{app-sec: general}, our main results hold beyond $\I$ under the fully unrestricted class of information rules.\footnote{ The information rules considered there allow for information to be any garbling of all events observable by the designer, including  a possible change of information in a non-stationary fashion.}

\paragraph{Steady State.}
Given the primitive process $(\lambda,\mu )$, a Markov policy $(x,y,z)$
generates a Markov chain---more specifically, a birth-and-death process---on
the queue length $k$. Given $(\lambda,\mu)$, we only consider a Markov
policy that induces an invariant distribution $p\triangleq (p_0,
p_1,\dots)$ on the queue length.
Specifically, this means that the distribution $p$ must satisfy the
following balance equation:
\begin{equation*}
\lambda_k x_k (1-\tsum_{\ell}z_{k,\ell}) p _k = (\mu
_{k+1}+\tsum_{\ell}y_{k+1,\ell})p _{k+1}, \,\forall k   \tag{$B$}
\label{B}
\end{equation*}
The LHS of the equation is the rate at which the queue length transits
from $k$ to $k+1$: with probability $p _k$ the queue length is $k$, in which
case an agent arrives at rate $\lambda_k$, is recommended to join the queue
with probability $x_k$, and no agent is removed from the queue with
probability $1-\sum_{\ell}z_{k,\ell}$. The balance equation \cref{B}
requires this rate to equal the rate at which the queue length transits
from $k+1$ to $k$, namely its RHS: with probability $p _{k+1}$ the queue
length is $k+1$, in which case an agent is served at rate $\mu _{k+1}$ or is
removed at rate $\sum_{\ell}y_{k+1,\ell}$ from the queue.
We say that an
entry/exit policy $(x,y,z)\in \mathcal{X}\times \mathcal{Y} \times \mathcal{Z%
}$ \textbf{generates} an invariant distribution $p$ if $(x,y,z,p)$ satisfies %
\cref{B}, and call the associated tuple $(x,y,z,p)$ an \textbf{outcome}.
From now on, we evaluate the policy at the associated outcome, assuming that the dynamic system is at a steady state.   This treatment is largely for expositional ease; \Cref{app:beyond-stationary} in the online appendix shows how our analysis carries through even when we focus on  a long-run time average of the Markov process that starts at an empty queue with $k=0$.\footnote{We prove that the Markov chain satisfying the incentive constraint must converge to a unique invariant distribution.  Further, our optimal queue design is optimal for this long-run time average formulation of the problem as long as the optimal queue length is finite, which, for instance, holds true when the designer puts a nonzero weight on the agent welfare in his objective.}

\smallskip

\paragraph{Incentives.}  The designer may keep an agent from joining the queue or remove an agent from the queue,\footnote{This assumption can be dispensed with under a broad set of circumstances, see the discussion at the end of \Cref{sec:FCFS}.} but
 the designer cannot coerce an agent to join or stay in the queue
against his preference.  Consequently, when recommended to enter the queue or
to stay in the queue, an agent must be provided with the incentive to obey that
recommendation, given the information available to him.

Formally, this obedience constraint is specified in terms of an agent's
beliefs about the queue length and position $(k_t, \ell_t)$ at each time,
which in turn determines the conditional service probability and expected
residual waiting times $(\sigma_{k,\ell }, \tau_{k,\ell})$. We evaluate
these variables when the system is at its invariant distribution $p$.
Obedience then requires:
\begin{equation*}
\sum_{k,\ell }\gamma_{k,\ell }^{t}\left[ V \cdot \sigma_{k,\ell }-C\cdot
\tau _{k,\ell }\right] \geq 0,\forall \gamma^{t}\in \supp(I^{t}),\forall
t\geq 0,  \tag{$IC$}  \label{IC}
\end{equation*}%
where $(\sigma_{k,\ell }, \tau_{k,\ell})$ is induced by the policy $(x,y,z,
q)$.    In words, \cref{IC} states that each agent, when recommended to join or
stay in the queue, must find the prospect of being
served to be high enough to justify the remaining waiting cost, given each possible
belief $(\gamma_{k,\ell}^t)$ at each $t\ge 0$.

In the sequel, we refer to the incentive constraint for $t$ by $(IC_{t})$.
We say that a queueing/information policy $(q,I)\in \mathcal{Q}\times
\mathcal{I}$ \textsf{implements} an outcome $(x,y,z,p)$ if \cref{IC} holds.
Even though we interpret an implemented outcome as resulting from the
designer's policy choice, this is without loss, due to the revelation
principle. Our model can capture any equilibrium outcome, both regulated and
unregulated.\footnote{\label{fn:mixing}For instance, consider the textbook unregulated and
unobservable $M/M/1$ queue (where agents arrive at rate $\lambda$ and where there is a single server serving an agent at rate $\mu$) governed by FCFS, in which agents make their
entry decisions without any recommendation or any information about the
queue length (see \cite{hassin-haviv2003} for
instance). If $%
\lambda $ is sufficiently large so that  $(\mu -\lambda )V<C$,  then there
exists a random entry probability $e\in (0,1)$ such that if all agents adopt
this mixing strategy, each agent becomes indifferent to entry, making it an
equilibrium behavior. In our model, this corresponds to our entry policy of
$x_{k,\ell }=e$ and $y_{k,\ell }=z_{k,\ell }=0$, for all $k,\ell $ (along
with FCFS and no information).} \medskip

\noindent\textbf{Problem statement.} 
The designer's objective is evaluated at the invariant distribution $p=(p_k)
$ of the Markov chain. It can be written as follows:
\begin{equation*}
W(p)\triangleq(1-\alpha) R \sum_{k=1}^{\infty}p _k \mu _k + \alpha
\sum_{k=1}^{\infty}p_k (\mu _k V - k C),
\end{equation*}
where $\alpha\in [0,1]$. The first term is the flow expected profit for the
service provider: with probability $p_k$, the queue has $k$ agents, and an
agent is served at rate $\mu_k$, generating a profit (or shadow
value) of $R$ for each agent served. The second term is the flow expected
utility for agents: again with probability $p_k$, the queue has $k$ agents,
each of whom pays a holding/waiting cost of $C$ per unit time (the second
term), and an agent is served and realizes a surplus of $V$, at rate $\mu_k$%
. The objective is a weighted sum of these two terms, with weight $\alpha\in
[0,1]$.  One can show that this objective corresponds to the expectation of the long-run time average of the designer's payoff
(see online appendix \Cref{app:beyond-stationary}).

The designer's problem is to choose $(p,x,y,z,q,I)\in \Delta (\mathbb{Z}%
_{+})\times \mathcal{X}\times \mathcal{Y}\times \mathcal{Z}\times \mathcal{Q}%
\times \mathcal{I}$ to 
\begin{equation*}
\sup \, W(p) \mbox{   subject to  }\mbox{\cref{B} and \cref{IC}}, %
\leqno{[P]}
\end{equation*}%
%
%
%
where the conditional service probabilities and residual waiting times $%
(\sigma_{k,\ell }, \tau_{k,\ell})$ in \cref{IC} are induced by $(p,x,y,z,q)$.%
\footnote{%
While the entry/exit policy $(x,y,z)$ uniquely pins down the invariant
distribution, we include $p $ as part of the designer's choice.} In words, the designer
picks the outcome that maximizes her objective among those that are
implementable by some queueing/information policy. Let $\mathcal{W}$ denote
the supremum of the value of program $[P]$.

\smallskip

\section{Scope of Applications}

\label{sec:applications}

Our model encompasses a variety of queueing and dynamic matching models
considered by the existing literature. \smallskip


 \textsf{$\bullet$ $M/M/c$ queue model:}
In this model, agents arrive at some constant rate $\lambda$. There are $c\ge 1$ servers each serving at a constant rate $\mu$.  A special case with $c=1$, known as $M/M/1$, is particularly common in the literature.\footnote{This model is adopted by \cite{naor1969regulation}, \cite{hassin1985}, \cite{simhon2016optimal}, \cite{hassin-koshman2017}, \cite%
{lingenbrink2019optimal}, among others.}
 $M/M/c$ model is a special case of our model in which $\lambda_k\equiv\lambda$, and
the service rate is linear up to
the number of available servers, so $\mu _k=\min\{k, c\}\mu $.
Clearly, this model satisfies regularity. In this model, our queueing formula simplifies to $q_{k,\ell}=\mathbf{1}%
_{\{\ell\le c\}} \cdot \mu$ under FCFS, $q_{k,\ell}=\mathbf{1}_{\{k-\ell+1\le
c\}} \cdot \mu$ under LCFS, and $q_{k,\ell}=\min\{k, c\}\mu/k$ under SIRO.
In fact, our model may capture a more general, and arguably more realistic, version of the $M/M/c$ model in which servers differ in their service rates.\footnote{That is, server $j$ serves at rate $\tilde \mu_j:=\mu_j-\mu_{j-1}$, with $\tilde{\mu}_0:=0$.
}

\smallskip
 \textsf{$\bullet$ Team servicing model:}
Suppose there are $m$ customers (or machines) each having a service need
arising at an independent Poisson rate while operating (see \cite%
{gnedenko1989}, p. 42). There are $c$ servers each of whom can serve a customer at  rate $\mu$. When there are $k$ agents in the queue, the arrival rate is then $%
\lambda_k=(m-k)\lambda$ and the service rate is $\mu _k=\min\{k,
c\}\mu $. Again, our regularity condition holds.
 \smallskip

 \textsf{$\bullet$ Dynamic one-sided matching with
stochastic compatibility:}
Suppose each agent is compatible with another agent with probability $%
\theta\in (0,1]$. In this model, an agent joins a queue only when he arrives
at some rate $\eta$ \emph{and} is incompatible with the agents already in
the queue, which occurs with
probability $(1- \theta)^k$, or else, he matches with a compatible partner and does not join
the queue,
which occurs with probability $(1-(1-\theta)^k)$. This is a special
case of our model in which $\lambda_k= \eta (1-\theta)^k$ and $\mu _k=\eta
(1-(1-\theta)^k)$. Observe that $\mu_k$ is increasing at a decreasing rate,
and $\lambda_k$ is decreasing, in $k$, so the process is regular. Our
queueing formula for FCFS, for instance, yields the service rate for $\ell$%
-th positioned agent to be $q_{\ell}= \mu_{\ell}-\mu_{\ell-1}= \eta
(1-\theta)^{\ell-1} \theta$, the probability that all agents ahead of him
are incompatible, and he is compatible, with an incoming agent. Likewise,
LCFS and SIRO formula have intuitive interpretations. \cite%
{doval2018efficiency} consider such a model with $\theta=1 $ and study
agents' incentive to join a queue under FCFS. \cite{akbarpour2017thickness}
study the limit as $\theta\in (0,1)$ tends to 0 but the arrival rate
increases.\footnote{Their focus differs from ours; for instance, they do not consider
the incentive to join or stay in a queue, the queueing rule, or information
design. Instead, they study the benefit from thickening the market, which we
do not consider.}

 \textsf{$\bullet$ Dynamic two-sided matching with
stochastic compatibility:} Heterogeneous agents on one side match with
heterogeneous agents or objects (e.g., housing) on the other side. If the
types of the matched pair are compatible, then high surplus is realized; if
not, a low surplus is realized. The designer operates buffer queues for
different types of agents or objects to keep the agents waiting until a
compatible match is found. \cite{leshno2019dynamic} and \cite%
{baccara2020optimal} consider such models. In these models, if one buffer
queue is active, the other is empty. Hence, the system can be analyzed as a
one-dimensional Markov chain. Some
of our results below rely on the system induced by a given policy to exhibit
birth and death processes. Indeed, this feature is satisfied under the
optimal policy under \cite{baccara2020optimal} but not under \cite%
{leshno2019dynamic}. Nevertheless, our central results apply to the latter
setup, as we show in \Cref{app: formal arg
	discussion sec} of the online appendix.\footnote{\cite{baccara2020optimal} consider optimal matching policy
under both FCFS and LCFS, whereas \cite{leshno2019dynamic} considers a
general class of queueing rules, and finds FCFS to be suboptimal.
Again, the current paper is differentiated by its consideration of
broad incentive issues (i.e., the incentive to stay in, not just to join, a
queue) and a general class of queueing rules as well as information design.
The fact that we draw a different conclusion on the optimal queueing
rule---namely, FCFS---relative to \cite{leshno2019dynamic} is attributed to
the combination of information design and choice of queueing rule together with our consideration of agents' dynamic incentives
(see \cref{sec:conclusion} for further discussion).}

\section{Main Result}

Below we state the main result of the paper: Under regularity of the primitive process,  FCFS and no information beyond recommendations together with the following particularly intuitive form of entry/exit policy solves the designer's program $[P]$:
\begin{definition}
An entry/exit policy $(x,y,z)$ is a \textsf{cutoff policy} if there exists $%
K\in \mathbb{Z}_{+}\cup \{+\infty \}$ such that $x_{k}=1$ for all $%
k=0,1,...,K-2$, $x_{K-1}\in (0,1]$, and $x_{k}=0$ for all $k\geq K$ and that
$y_{k,\ell }=z_{k,\ell }=0$ for all $k,\ell $.
\end{definition}
In words,  a cutoff policy  sets a maximum queue
length $K$ and recommends that an arriving agent joins a queue as long as $%
k\leq K-1$ and that those who join the queue stay in the queue
until they are served. Thus, no agent is removed
or induced to abandon the queue once he has joined it. It is possible that $%
x_{K-1}\in (0,1)$, in which case the $K$-th entrant may be randomly rationed.%
\footnote{%
While we assume $y_{k,\ell }=z_{k,\ell }=0$ for all $k,\ell $, this is just
a convenient normalization. If $x_{K-1}\in (0,1)$ in a cutoff policy, the
same $p^{\ast }$ can be implemented by any $(x^{\prime },y^{\prime
},z^{\prime })$ such that ${x_{K-1}^{\prime }}=\frac{\mu _{K}+\sum_{\ell
}y_{K,\ell }^{\prime }}{\mu _{K}(1-\sum_{\ell }z_{K-1,\ell }^{\prime })}%
x_{K-1}$; see \cref{B}. In this sense, the reader should interpret the
cutoff policy as an equivalence class involving a set of such pairs. This
means that while it is unnecessary to induce an agent to exit from a queue
after he joins it, doing so when the queue length is $K-1$ (and $x_{K-1}\in
(0,1)$) or $K$ is consistent with a cutoff policy. In other words,
encouraging a customer to come back later is not at odds with a cutoff
policy.}

We are now in a position to state our main theorem.
\begin{theorem}
\label{thm: main} Assume that the primitive process is regular. There is an optimal solution $(x^*,y^*,z^*,q^*,I^*)$ of $[P]$ s.t.
(i) $(x^*,y^*,z^*)$ is a cutoff policy; (ii) $q^*$ is FCFS; and (iii) $I^*$ is the no information rule.
\end{theorem}
In order to prove this statement we study a relaxed problem for the designer where, in essence, the designer only chooses the entry/exit policy $(x,y,z)$ (or, equivalently, the invariant distribution). We define this relaxed problem in the next section (\Cref{sec:cutoff}) and prove that the optimal solution is a cutoff policy when the service process is regular (\Cref{thm:cutoff}). In \Cref{sec:FCFS}, we show that this cutoff policy together with FCFS and the no information rule satisfy all constraints of problem $[P]$ proving that this forms an optimal solution of $[P]$ (\Cref{thm:dyn-fcfs}). These two results together yield \Cref{thm: main}.

Intuitions for \Cref{thm: main} will be provided in the next sections when we establish the intermediary theorems (\Cref{thm:cutoff} and \Cref{thm:dyn-fcfs}).

\subsection{ Optimality of the Cutoff Policy}

\label{sec:cutoff}

The designer's problem $[P]$ is, in general, difficult to solve. Instead, we
consider the following  relaxed problem: \setstretch{0.5}
\begin{equation*}
\max_{p \in \Delta(\mathbb{Z}_+)} W(p) \leqno{[P']}
\end{equation*}
subject to
\begin{align*}
\sum_{k=1}^{\infty}p_k (\mu _k V - k C) &\ge 0;  \tag{${IR}$}  \label{IR} \\
\lambda_k p _k - \mu _{k+1} p _{k+1}&\ge 0, \forall k.  \tag{${B}'$}  \label{B'} \\
\end{align*}
\setstretch{1.22} Here, the planner maximizes the designer's objective
subject only to individual rationality \cref{IR} and a weakening %
\cref{B'} of the balance equation \cref{B}. The problem constitutes a linear
program (LP) involving an infinite-dimensional measure $p$.

Clearly, $[P^{\prime }]$ is a relaxation of $[P]$. First, \cref{IR}
must be implied by \cref{IC}. If the former condition fails,
the agents do not ex ante break even. Then, there must exist \textit{%
some} agent and \textit{some} belief induced by that mechanism such that the
agent with that belief would not wish to join a queue when called upon to do
so. Hence, \cref{IC} would fail.  (A rigorous proof is provided in \Cref{lem:general} of \Cref{app-sec: general} of the online appendix.\footnote{\label{fn:IR and IC_0}  The proof of that lemma can be sketched here.
 Fix any $(x,y,z,p,q,I)$ that satisfies $%
(IC^0)$. Aggregating $(IC^{0})$ across all beliefs $\gamma^{0}\in \supp%
(I^{0})$, we get
\begin{equation*}
\int_{\gamma ^{0}}\sum_{k,\ell }\gamma _{k,\ell
}^{0}[V \sigma _{k,\ell }-C\tau _{k,\ell }]I^{0}(d\gamma ^{0})\geq 0.
\end{equation*}%
Since the queueing rule is work-conserving, the ex-ante  probability of eventually receiving service, $%
\int_{\gamma ^{0}}\sum_{k,\ell }\gamma _{k,\ell
}^{0} \sigma _{k,\ell } I^{0}(d\gamma ^{0})$, must equal $\sum_{k}p_{k}\mu
_{k}/[\sum_{k}p_{k}\lambda _{k}x_{k}]$---the average rate of receiving
service divided by the average rate of entering the queue at $p$. Next, by
Little's law, the  {ex-ante} expected waiting time, $\int_{\gamma
^{0}}\sum_{k,\ell }\gamma _{k,\ell }^{0}\tau _{k,\ell
}I^{0}(d\gamma ^{0})$, equals $\sum_{k}p_{k}k/[\sum_{k}p_{k}\lambda
_{k}x_{k}]$---the average queue length divided by the average entry rate.
Substituting these two expressions and simplifying the terms, the above
inequality implies \cref{IR}.}) Next, since the $y_{k,\ell }$ are nonnegative
and $z_{k,\ell },x_{k,\ell }$ are all in $[0,1]$, \cref{B} implies \cref{B'}.  Let $\mathcal{W}^{\ast }$ denote the supremum of the value of program $%
[P^{\prime }]$. Then, whenever $\mathcal{W}^{\ast }<\infty $, we must have $%
\mathcal{W}^{\ast }\geq \mathcal{W}$.

The program $[P^{\prime }]$ is interesting in its own right: it can be
interpreted as the problem facing a planner who chooses the invariant
distribution $p$ directly to maximize her objective, simply facing the
primitive process $(\lambda,\mu )$, but disregarding agents' incentives
altogether, except for guaranteeing some minimal payoff for them.
Ultimately, however, we are interested in $[P^{\prime }]$ as an analytical
tool for characterizing an optimal queue design that solves $[P]$, since a
solution to this relaxed program $[P^{\prime }]$ may be attained by a mix of
policy tools $(x,y,z, q, I)$.

Indeed, our ultimate goal is to prove such a policy mix exists, which will
then imply that it optimally solves $[P]$, the real object of interest. The analysis proceeds in three claims: (i) an optimal solution $%
p^{\ast }$ to $[P^{\prime }]$ exists, (ii)
under regular service processes, the optimal solution to the relaxed problem is implemented by a simple entry/exit rule, called a cutoff policy;
(iii)  FCFS, together with  no information  rule, satisfies
\cref{IC} under the optimal cutoff policy.  Since $\mathcal{W}^{\ast }\geq \mathcal{W}$, it would then follow that the latter policy mix solves $[P]$, our original problem of interest.
The remainder of this section
will address (i) and (ii), while claim (iii) will be taken up in the next
section.






 Our next result establishes that under regular service processes, an optimal solution of $[P']$ can be implemented by a cutoff policy. All proofs of the paper are relegated to the Appendix.

\begin{theorem}
\label{thm:cutoff} An optimal solution of $[P^{\prime }]$ exists. If $\mu $
is regular, there is an optimal solution to $[P^{\prime }]$ implemented by a
cutoff policy with maximal queue length $K^*\ge \argmax_k \mu _k V- k C$.
\end{theorem}






 The intuition behind the result can be traced to the fundamental trade-off associated with queueing. Although queueing agents
may at first glance appear wasteful, it serves  as an ``insurance'' against the risk of the service capacity going idle and wasted when too few agents show up for the queue. While this insurance benefit is positive for any queue length, it falls as more agents enter the queue due to the concavity of $\mu_k$ in $k$. Moreover, the waiting costs of agents increase as more of them enter the queue. These two observations explain that a cutoff policy would be optimal.

\subsection{Optimality of FCFS with No Information}

\label{sec:FCFS}

In this section, we establish the general optimality of FCFS with no
information. From now on, we assume that the service process is regular (i.e., part (i) of
\Cref{def: regular}). Then, by \cref{thm:cutoff}, the optimal solution $p^*$ to $%
[P^{\prime }]$ is implemented by a cutoff policy $(x^*,y^*,z^*)$ with a
maximal queue length $K^*\in \mathbb{Z}_+ \cup \{+\infty\}$. To
avoid the trivial case, we assume that $K^*>1$. Further, recall that the
optimal cutoff policy has $y^*_{k,\ell}=z^*_{k,\ell}=0$. For notational ease, we sometimes simply write this optimal cutoff policy as $%
x^*$, and similarly, write the optimal policy $(x^*, y^*, z^*, p^*)$ as $%
(x^*, p^*)$.

In what follows, we fix the optimal outcome $x^*$ and the maximal queue length $K^*>1$. We will then show that FCFS, together with an optimal information
design, implements $(x^*, p^*)$; namely, \cref{IC} holds under that policy.
Since $[P^{\prime }]$ is a relaxation of $[P]$, this will prove that the
identified policy mix solves $[P]$.

We denote the first-come-first-served (FCFS) rule by $q^*$, where, as
defined before, the service rate is given by $q_{k,\ell}^*=\mu _{\ell}-\mu
_{\ell-1} \triangleq q^*_{\ell}$ for each $(k,\ell)$ with $k\ge \ell$. Not surprisingly, under
FCFS the expected waiting time depends only on one's queue position $\ell$,
so we use $\tau^*_{\ell}$ to denote the expected waiting time for an agent
with queue position $\ell$. Given the primitives, this can be pinned down
exactly.

\begin{lemma}
\label{lem:waiting-time-FCFS} For any $\ell=1,..., K^*$, $\tau^*_{\ell}= {%
\ell}/{\mu _{\ell}}$. $\tau^*_{\ell}$ is
nondecreasing in $\ell$. If $2\mu_1>\mu_2$, then $\tau^*_{\ell}$ is strictly
increasing in $\ell$.
\end{lemma}


From now on, we denote the no information rule by $I^*\in
\mathcal{I}$. Recall that, under this rule, \emph{no information is provided to
each agent both at the time of joining the queue and after joining the
queue, beyond
recommendations to join or stay in
the queue}.
This means that when he joins the queue, he
forms a belief about his position $\ell$, or the length of the queue, based on his prior belief (given by the invariant distribution) and the recommendation to join the queue.  From
then on, he updates the belief about his queue position at each $t>0$
according to Bayes rule without any further information (given that he is
recommended to stay from then on). In practice, the no-information rule can be implemented by sending a message consisting of either ``join'' or ``leave,'' or by providing a coarse (i.e., binary) estimate of the ``expected'' waiting time, to an arriving agent.

Given the cutoff policy $x^*$ and the queueing and information rules $(q^*, {I}^*)$, the incentive
constraint at time $t$ is given by
\begin{equation*}
V- C \sum_{\ell=1}^{K^*} \tilde{\gamma}_{\ell}^{t} \cdot \tau^*_{\ell}\ge 0, %
\leqno(IC_{t})
\end{equation*}%
where $\tilde{\gamma}^{t}=(\tilde{\gamma}_1^{t}, ...., \tilde{\gamma}%
_{K^*}^{t}) \in \Delta (\{1,..., K^*\}) $ is the belief on his position in
the queue after spending time $t$ on the queue.\footnote{Note that $\sigma_{k,\ell}=1$ for all $k,\ell$ since, by definition of the cutoff policy,
the designer never removes agents from the queue.} Since the expected waiting
time depends only on one's position, the belief on other variables such as
the queue length $k$ does not affect the agent's incentive to join or stay
in the queue.

Given the information rule ${I}^*$, the belief at the time of joining the
queue must be:
\begin{equation}
\tilde{\gamma}_{\ell }^{0}=\left\{%
\begin{array}{ll}
\frac{p^*_{\ell-1 }\tilde\lambda_{\ell-1 }}{\sum_{i=0}^{K^*-1}p^*
_{i}\tilde\lambda_{i} } & \mbox{ if } \ell=1,..., K^* \\
\label{eq:belief0} 0 & \mbox{ if } \ell >K^*, \cr%
\end{array}
\right.
\end{equation}
where $\tilde \lambda_k$ is an ``effective'' arrival rate given by: $\tilde
\lambda_k\triangleq \lambda_k$ for $k=0,...,K^*-2$, and $\tilde%
\lambda_{K^*-1}\triangleq x^*_{K^*-1}\lambda_{K^*-1}$.\footnote{Recall that the
optimal cutoff policy may involve random entry at $k=K^*-1$; recall that $%
x^*_{K^*-1}\in (0,1]$ stands for the optimal randomization at $k=K^*-1$.} This formulation
rests on the consistency of an agent's belief about the rule in
place---namely, $(x^*, q^*,I^*)$---as well as the invariant distribution $%
p^*$. Specifically, \cref{eq:belief0} computes the probability of an agent
occupying position $\ell$ conditional on entering the queue. Its numerator
is the probability that an agent joins the queue in state $\ell-1$, which
equals the probability of there being $\ell-1$ agents already in the queue
multiplied by the probability of entry per unit time in that state $%
\tilde\lambda_{\ell-1 }$.\footnote{The formula in \cref{eq:belief0} is justified as follows. Recall that (by the PASTA property---\cite{wolff1982}),  one can think of an agent's {\it unconditional} belief about the state  as  given by the invariant distribution over states. The conditional belief is then obtained by conditioning based on the entry $\{x_k\}$ policy as well as the heterogeneity in the arrival rate $\{\lambda_k\}$.}
Its denominator is the total probability of entering the queue per unit of time.

It is easy to show that the candidate policy $(q^*, I^*)$ provides the agents with incentives to enter the queue, i.e., it satisfies $(IC_0)$.  In fact, $(IC_0)$  follows from \Cref{IR} {\it regardless of} the queueing rules, under no information $I^* $.\footnote{\Cref{fn:IR and IC_0} shows how this is implied by \cref{IR}. See also online appendix \Cref{online_app:IC_0} for an alternative argument using directly the characterization of waiting times under FCFS given in \Cref{lem:waiting-time-FCFS}.}  By contrast, it is more challenging to show that   $(q^*, I^*)$ satisfies $(IC_t)$ for $t>0$, namely, that the agents have the incentive to stay in the queue once they join it.  To
examine the latter, we need to study how an agent's belief evolves after he joins the queue. Since no agent is
recommended to abandon the queue, $(IC_t)$ for $t>0$ boils down to whether
an agent's belief about his queue position becomes (at least weakly) more
favorable---or put more probability at lower $\ell$'s---as time passes.

Suppose that an agent has belief $\tilde{\gamma}^t$ after spending time $%
t\ge 0$ in the queue. By Bayes rule, after time $t+dt$, his belief is
updated to:\footnote{  \Cref{online_app:infty} derives this belief recursion equation rigorously.}
\begin{equation*}
\tilde{\gamma}_{\ell }^{t+dt}=\frac{\tilde{\gamma}_{\ell
}^{t}(1-\sum_{i=1}^{\ell } {q}^{\ast }_i dt)+\tilde{\gamma}_{\ell
+1}^{t}\sum_{i=1}^{\ell } {q}^{\ast }_i dt }{\sum_{i=1}^{K^* }\tilde{ \gamma}%
_{i}^{t}(1- {q}^{\ast }_idt)} +o(dt)\text{.}
\end{equation*}
The numerator is the probability that his queue position is $\ell$ after
staying in the queue for length $t+dt$ of time. This event occurs if either
(i) the agent already has position $\ell$ in the queue at time $t$ and none
of the agents ahead of him and himself have been served during time increment $dt$; or
(ii) if he has position $\ell+1$ at $t$ and one agent ahead of him is served
by $t+dt$.\footnote{%
The probability of multiple agents ahead of him being served during $[t,
t+dt)$ has a lower order of magnitude denoted by $o(dt)$.} The denominator in turn gives the total probability that the agent has
not been served by time $t$. Hence, given that an agent has not been served
by $t $, the above expression gives the conditional belief that his position
in the queue is $\ell $ at time $t+dt$. By the definition of FCFS, we have $\sum_{i=1}^{\ell }{q}^{\ast }_i=\mu _{\ell}$, so we can
rewrite the belief updating rule as:

\begin{equation}
\tilde{\gamma}_{\ell }^{t+dt}=\frac{(1-\mu _{\ell}dt)\tilde{\gamma}_{\ell
}^{t }+\mu _{\ell}dt \tilde{\gamma}_{\ell +1}^{t }}{\sum_{i=1}^{K^* }\tilde{
\gamma}_{i}^{t}(1- {q}^{\ast }_idt)}+o(dt)\text{.}  \label{Eq: cond beliefs}
\end{equation}

We now study how the belief updates dynamically over time under $(q^*, I^*)$. The statistic we focus on is the \textbf{likelihood ratio} $%
r^t_{\ell}\triangleq \frac{\tilde \gamma^t_{\ell}}{\tilde \gamma^t_{\ell-1}}$
in beliefs of being in queue position $\ell$ to being in queue
position $\ell-1$ after spending time $t$ on the queue. One can use %
\cref{Eq: cond beliefs} to derive a system of ordinary differential
equations (ODEs) on the likelihood ratios:
\begin{equation}  \label{eq:lr}
\dot r_{\ell}^{t}=r^{t}_{\ell} \left(\mu _{\ell-1}-\mu _{\ell} -\mu
_{\ell-1} r^{t}_{\ell} +\mu _{\ell} r_{\ell +1}^{t}\right),
\end{equation}
where $\ell=2,..., K^*$. Further, the invariant distribution $p^*$ can be
used to obtain the boundary conditions, $r_{\ell}^{0}= \frac{%
\tilde\lambda_{\ell-1}}{\mu _{ \ell-1}},$ for $\ell=2,..., K^*$, where we recall that $%
\tilde\lambda_k$ is the effective arrival rate. \Cref{app:dyn-incentives} derives this system
of ODEs and establishes the existence of a unique solution.

We will argue that the regularity of the primitive process (in particular part
(ii) of \Cref{def: regular}) is sufficient for these likelihood ratios---the solution to the above
ODEs---to decline over time, meaning one's belief about his position becomes
progressively favorable under $(q^*, I^*)$. At first glance, this seems
obvious under FCFS: conditional on starting at any position $\ell$ at $t=0$,
an agent's queue position can \emph{only} improve as time passes. Since the
agent begins with no information, however, this is not the only event about
which the agent updates his beliefs. The agent is also updating his belief
about his initial position $\ell$. On this account, however, the time $t$ spent on the queue is ``bad'' news, as it suggests that he may have been too
optimistic about his position initially, causing him to revise his initial
queue position pessimistically as time passes.

\begin{figure}[t]
\caption{Belief about position $\ell=1$}
\label{fig:top-position}\centering
\includegraphics[height=3in]{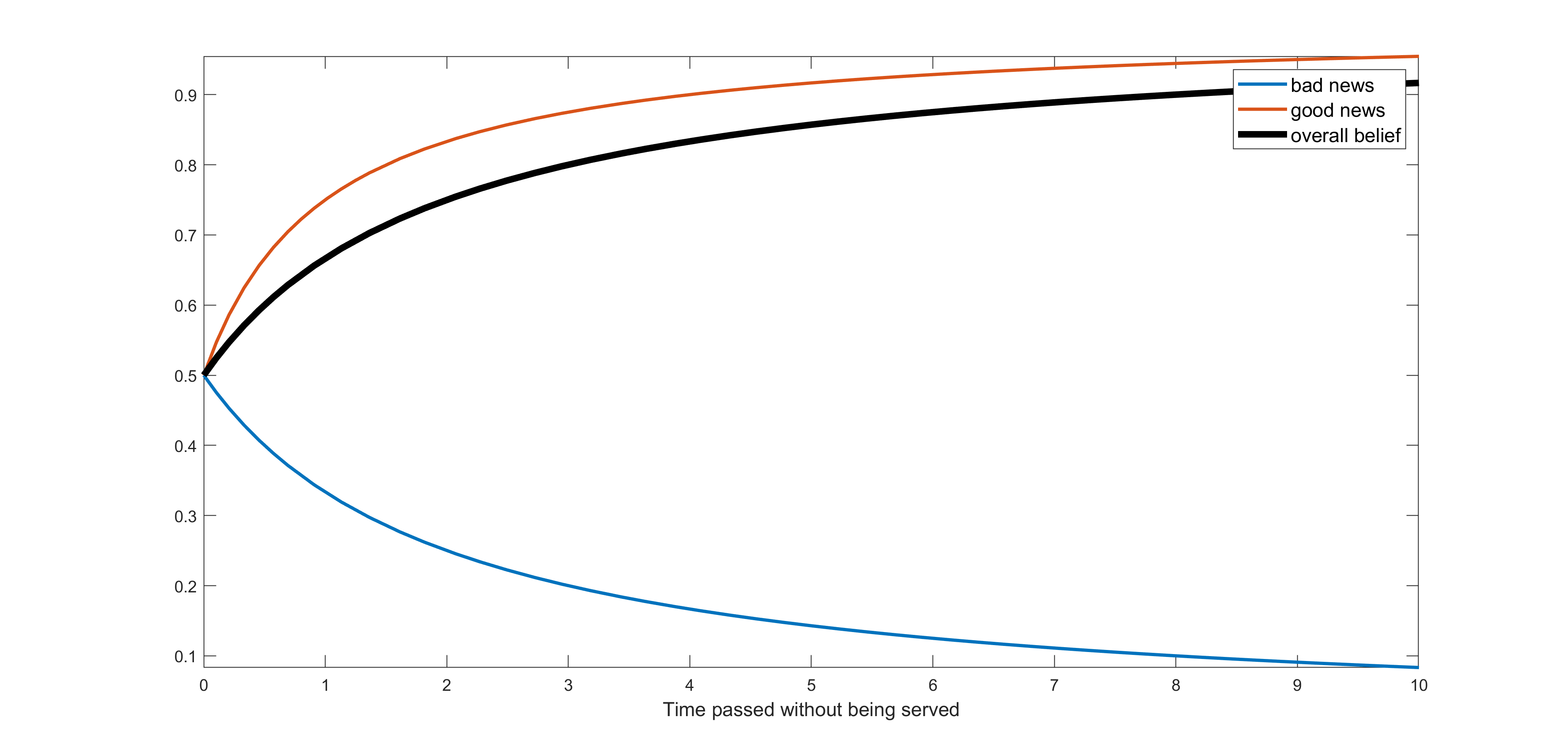} Note: $M/M/1$ with $K^*=2$%
; $\lambda=\mu =1$.
\end{figure}

\cref{fig:top-position} displays these two competing effects in an $M/M/1$
queue with $K^*=2$. Its top graph depicts the good news effect: an agent's
belief about being at the top position ($\ell=1$) is improving over time
\textit{when the belief about his initial queue position is held fixed at
the prior}. The bottom graph depicts the bad news effect: the belief about
his initial queue position being $\ell=1$ falls over time. The middle graph
displays the overall evolution of the belief---namely about $\ell=1$
conditional on not being served by $t$. Its increase means that the former
``position-improvement'' effect dominates the worsening posterior about the
initial position.

The regularity of the primitive process is sufficient for the good news
effect to dominate the bad news effect:

\begin{lemma}
\label{lem:dyn-incentives} Assume that the primitive process $(\lambda,\mu )
$ is regular. Then, for all $\ell\in \{2,..., K^*\}$, $r^{t}_{\ell}$ is
nonincreasing in $t$ for all $t \geq 0$.
\end{lemma}


Intuitively, regularity ensures that the arrival rate does not rise faster
than the service rate as the queue length increases. This keeps the adverse
inference about initial position from worsening one's belief about the
residual waiting time.\footnote{ {Our proof method differs from the standard queuing analysis which focuses on the increasing or decreasing hazard rate of an agent's waiting time in the  $M/M/1$ and $M/M/c$ queue models (see \cite{gnedenko1989}).  Analyzing the evolution of hazard rates appears difficult in our general Markovian model. We believe that the current method that tracks the evolution of posterior beliefs are of independent analytical interest for queuing theory.}
} We can now state the following theorem.

\begin{theorem}
\label{thm:dyn-fcfs} Assume that the primitive process is regular. Then,
FCFS with no information $(q^*, I^*)$ implements the optimal outcome $%
(x^*,p^*)$ where $p^*$ solves $[P']$. Consequently, $(x^*, q^*,I^*)$ is an optimal solution of $[P]$.
\end{theorem}


 We close this section with two remarks.
First, the above result relies on the designer's ability to stop an agent
from entering a queue. While the designer does have such a power in many
settings, the power is unnecessary if $V \mu_{K^*}\le K^* C $, which holds
for instance if \cref{IR} is binding at the optimal outcome; the latter in turn holds when $1-\alpha$, the weight in the designer's objective on the service provider's profit, is large enough. In that case,
the designer can simply issue a ``recommendation'' not to enter when $k=K^*$%
, and the agent will follow that recommendation.\footnote{%
Given the length $K^*$ (which the agent infers from the recommendation not
to enter), he expects to wait for $\tau^*_{K^*}=K^*/\mu_{K^*}$ (recall %
\Cref{lem:waiting-time-FCFS}).}
	
Second, to the extent that regularity is extremely mild, one may view this theorem
as suggesting that the combination of FCFS and No Information is optimal in
a broad set of circumstances. Nevertheless, the dynamic incentives provided
by FCFS, or the role played by regularity conditions, should not be taken
for granted. Intuitively, with the failure of regularity, delay is more of a
signal about the initial queue length being long than about predecessors
having been served, and thus one's belief, and therefore one's incentive to stay in the queue,
may get worse over time. We provide an example in \Cref{Example_non-regular} of the online appendix
where regularity fails and as a consequence the optimal solution to $[P^{\prime }]$ is
not implementable under $(q^*, I^*)$.

\section{Necessity of FCFS for Optimality in a Rich Domain}

\label{sec:necessity}

We have shown that FCFS with no information is optimal in all regular
environments. This result raises the question of whether a different
queueing/information policy may be also optimal in some (or all) environments. While some other policies may be also optimal in some environments,\footnote{For instance, one can show that, when $\alpha =1$, FCFS is optimal under full
information, with the entry controlled optimally. See our generalization of \cite%
{naor1969regulation} in appendix \Cref{sec:Naor generalization}.  In the same environment,  \cite{hassin1985} and \cite{su-zenios2004}%
) have shown that versions of LCFS, possibly \emph{with preemption} (i.e.,
where a newly arriving agent replaces one under service), are optimal under
full information when $\alpha=1$.}
we show below none of them  can be optimal in {\it all}
regular environments. Specifically, we show that of all feasible queueing rules, FCFS is the
only queueing rule that is optimal for all (regular) queueing environments.
Or equivalently, for any queueing rule differing from FCFS, we exhibit a
(regular) environment in which this rule is suboptimal under any information
rule.

For this purpose, we focus on the  simplest
environment: the M/M/1 environment in which a uniquely optimal solution to $%
[P^{\prime }]$ involves (i) $K^*=2$, (ii) no rationing when $k=K^*-1=1$, and
(iii) a binding \cref{IR}. Specifically, we fix any service rate $\mu>0$. We
then consider a sufficiently small arrival rate $\lambda$ by letting it
approach zero. When we do this, we simultaneously adjust the values of $(V,
C,\alpha)$ to ensure that properties (i), (ii), and (iii) continue to hold.%
\footnote{\label{fn:y zeros copy(2)}
These requirements can be met by
choosing $V/C=\frac{2\lambda +\mu }{(\lambda +\mu )\mu }$ and $\alpha =0$.
In that case, there is a unique optimal solution $p$ to $[P^{\prime }]$ and
any outcome $(x,y,z)$ implementing $p$ satisfies (i), (ii) and (iii). Note
that assumption (iii) precludes $\alpha =1$ under which \cref{IR} is
non-binding at the optimal policy as long as the value of the objective may
be strictly positive.}

Since $K^{\ast }=2$,
there are only three relevant \textquotedblleft states,\textquotedblright\ $%
(k,\ell )=(1,1),(2,1),(2,2)$, based on the queue length $k$ and one's queue
position $\ell $. Hence, we can denote a queueing rule by $%
q=(q_{1,1},q_{2,1},q_{2,2})$. Recall that FCFS corresponds to $q^{\ast
}=(\mu ,\mu ,0)$. For any feasible work-conserving queueing rule, we must
have $q_{1,1}=\mu $ and $q_{2,1}+q_{2,2}=\mu $. Hence, a queueing rule $q \in \mathcal{Q}$
can differ from FCFS $q^{\ast }$ if and only if $q_{2,1}<\mu $, or
equivalently, $q_{2,2}>0$.
 Formally, we say that a queueing rule \textbf{differs from FCFS}
if $q_{2,2}$ is bounded away from $0$ for all possible values of  $\lambda$ (recall that we have fixed the value of $\mu$).\footnote{A standard queueing rule does not depend on the arrival rate of agents. An exception is Load-Independent Expected Wait (LIEW)  considered by \cite{leshno2019dynamic}, which adjusts priorities based on the arrival rates. Nevertheless, LIEW has $q_{2,2}$ bounded away from $0$, so it satisfies our definition of a queueing rule differing from FCFS.}
All
queueing rules studied in the literature such as SIRO, LCFS, and LIEW differ
from FCFS in this sense. 
We are now in a position to state the main result of this section:

\begin{theorem}
\label{thm:necessity} Fix any queuing rule $q$ that differs from FCFS. Then,
there exists a regular (in particular $M/M/1$)
 queueing environment  with values $(V,C, \alpha, \lambda, \mu)$ such
that the queueing rule $q$ fails $(IC_{t})$ for some $t>0$ under any
information policy. Hence, $q$ cannot implement the optimal cutoff policy
under any information policy.
\end{theorem}


The intuition for this result is most clear under LCFS. Under this rule, an
agent loses his service priority when another agent enters. So, if an agent
were initially indifferent to queueing,
he will definitely wish to abandon the queue once a new agent enters.
Consequently, $(IC_{t})$ fails at time $t$ when a new entry occurs if he had
full information. Even with no information, as time passes without getting
served, an agent will suspect that a new entry is increasingly likely and he
will lose his priority as a consequence. This feature destroys his dynamic
incentive.\footnote{A similar problem arises with LIEW, the queueing rule that equalizes the expected waiting times upon entry, to maximize the incentive to join the queue.  Note the latter goal is achieved under all queueing rules once the no-information policy is adopted. More problematic is the incentive to stay resulting from LIEW.  The equalization of
waiting time across queue lengths means that an agent who enters an empty
queue must be ``penalized'' in service priority later when a new agent
enters. This very feature undermines the dynamic incentive of an agent. The
root cause of the problem under these rules is: $q_{2,1}<q_{1,1}=\mu$%
---namely, the loss of priority an agent suffers when a new agent arrives.}
Although LCFS is extreme in this regard, any rule that assigns $%
q_{2,1}<\mu =q_{1,1}$, including SIRO, suffers from the same fundamental
issue.  As mentioned in the introduction, the issue is traced to the dispersed wait times arising from these rules, compared with FCFS. A dispersion of wait times creates unfavorable conditional beliefs for agents since the elapse of time on the queue (without being served) signals a longer residual wait time.

 This point is illustrated in \cref{fig:waitingtime}, which plots the expected waiting times
against the time-on-the-queue under five queueing disciplines: FCFS, SIRO, LIEW, LCFS, and LCFS-PR, where LCFS-PR is the LCFS with ``preemption,'' namely, a rule in which an old agent leaves when a new agent enters the queue.
As is clearly
seen, and consistent with \Cref{thm:necessity}, as time passes, an agent in
the queue expects to wait increasingly longer under all these disciplines,
except for FCFS under which his expected wait decreases.
\begin{figure}[h]
\caption{Expected waiting times under alternative values of $q$.}
\label{fig:waitingtime}\centering
\includegraphics[height=3in]{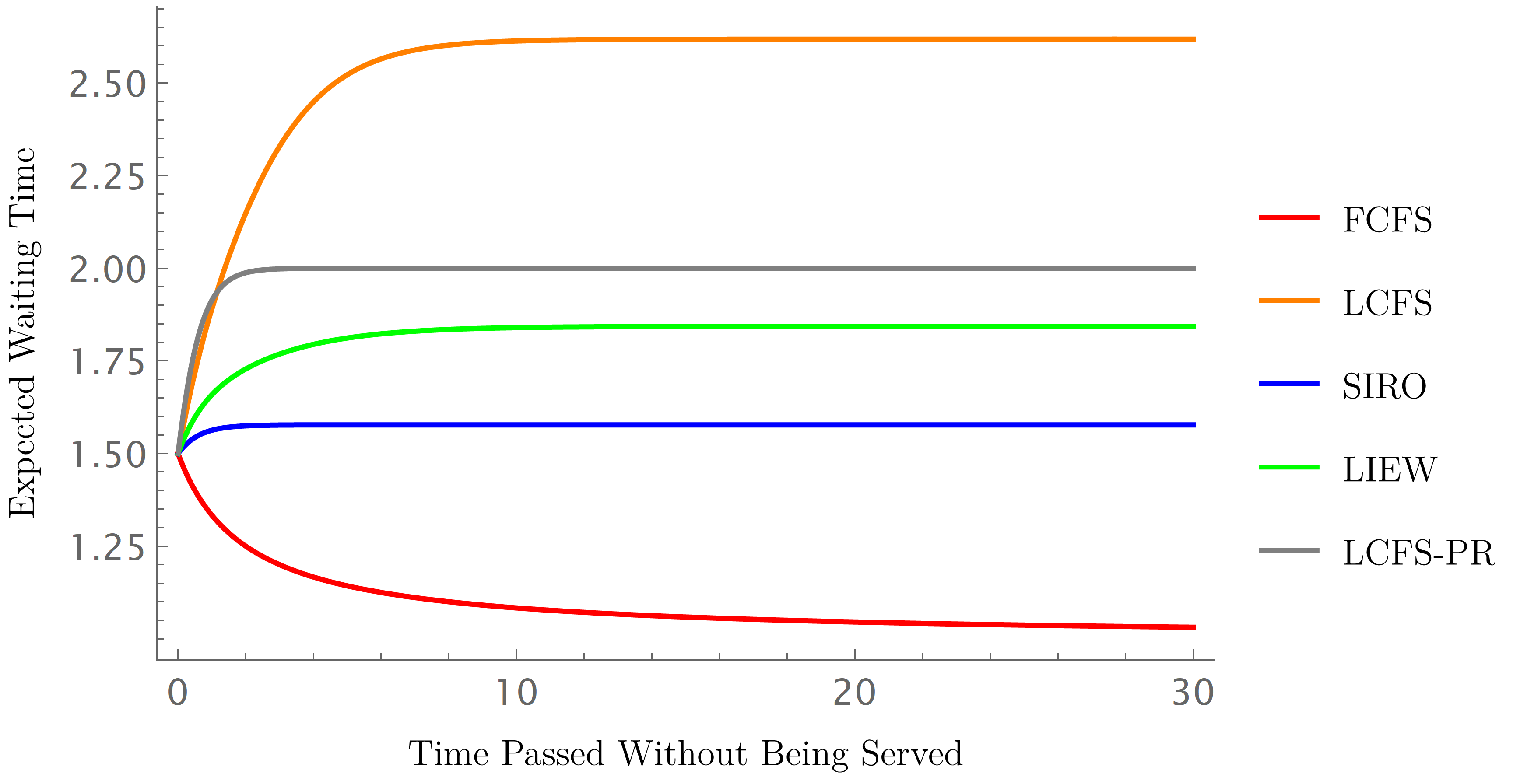}\newline
Note: $M/M/1$ with $K^*=2$; $\lambda=\mu =1$.
\end{figure}

\section{Concluding Remarks}

\label{sec:conclusion}

While we have focused on a canonical queueing model, the insights we obtain appear general and apply beyond our model. Here we
discuss how one may extend our analysis to other settings of potential
interest. \smallskip

\indent\textbf{Dynamic two-sided matching.} A topic closely related to
queueing is dynamic matching; see \cite{akbarpour2017thickness}, %
\citet{akbarpour2020unpaired}, \cite{baccara2020optimal}, \cite%
{leshno2019dynamic}, \cite{doval2018efficiency}, and \cite%
{ashlagi2019imbalance}, among others. The primary focus of this literature is
the optimal timing of matching and assignment, rather than queueing
incentives. Exceptions are \cite{leshno2019dynamic} and \cite%
{baccara2020optimal}, who study incentives for two different types of agents
for queueing to match with either two different types of objects (e.g.,
housing) or agents. In such a model, efficiency calls for accumulating
agents in a queue until the right type of object or agent arrives, to avoid
mismatching.
\cite{leshno2019dynamic} assumes overloaded demand so that the
planner wishes to incentivize the agents to queue as much as possible, and shows
that, given complete information, SIRO outperforms FCFS in
this regard and LIEW outperforms all other mechanisms.  This result rests crucially on his assumption of complete information. In fact, the main problem of his model is captured precisely by an M/M/1 version of our model with $\alpha=0$, where the designer wishes to maximize the incentive for queuing just as in his model. As has been shown in the current paper, with optimal information design,  the FCFS could do just as well as any other mechanism, including LIEW, in incentivizing agents to enter a
queue. Meanwhile, if the dynamic incentives are the problem, which the existing authors ignored, then FCFS does strictly better than other
queueing disciplines.\footnote{ Despite the ostensible difference in modeling,  \Cref{app: formal arg	discussion sec} in the online appendix shows that our analysis applies
without much modification to Leshno's model, and points out that the main
results from \cite{leshno2019dynamic} rest on his full information
assumption.  Strictly speaking, \cite%
{leshno2019dynamic} assumes the value of outright exit to be very low (e.g., in
comparison with the value of a mismatched object), so the dynamic incentives may not be a problem.  If the value of the outside option is significant, however, as we assume in \Cref{app: formal arg discussion sec} in the online appendix,  then the dynamic
incentives will matter just as they do in our model. Note also that the dynamic incentive issue does not arise in SIRO or FCFS under complete
information: any agent who joins the queue will have the incentive to stay
in the queue. But recall that neither discipline would implement the
optimum under complete information. Under \emph{no information} (which is
optimal), dynamic incentives will be an issue.}
\cite{baccara2020optimal}'s model is similar to that
of \cite{leshno2019dynamic}, except that there are agents on both sides.
Here again, our main insights in \Cref{thm: main} and \Cref{thm:necessity} apply.\footnote{Unlike \cite{leshno2019dynamic}, agents' incentives to enter a
queue may be excessive under FCFS or LCFS with full information. While this is an
issue in their decentralized matching, in our setting the designer can
easily solve the problem by preventing an agent from entering a queue, as is
often done in practice.}

\smallskip

\indent\textbf{Monopolist problem with endogenously set fee.}
If $\alpha=0$, one could interpret
the service provider/designer as a
monopolist who provides the service. We treated
 the fee $R$ as exogenous, representing the shadow value of addressing customer
needs. In many contexts, however, one may think of this profit as a monetary fee
collected and set by the service provider. In this case, the designer/monopolist  chooses this fee $R$ and the net surplus of customers for
service now equals $V-R$. Hence, we can rewrite problem $[P]$ assuming that $%
R$ is part of the decision variables and incorporating the new net surplus of customers
into the \cref{IC}
condition.
Our framework can be easily adapted to this environment. Indeed, one can
write $[P^{\prime }]$ assuming that the designer chooses both the invariant
distribution and the fee level.  Given the optimal choice of fee, the rest of the proof applies without any modification.
Namely, a cutoff policy is optimal. Clearly, \Cref{lem:dyn-incentives} must still hold, so \Cref{thm:dyn-fcfs} (and so \Cref{thm: main}) extends to
this context.  Incidentally, one can also characterize the
optimal fee in this context.
Intuitively, when choosing the fee, the monopolist should
consider both its impact on his profit and also on the incentives of agents
to join (and stay) in the queue. For instance, a higher fee
increases profit but may also discourage agents from joining the queue, which
increases the probability that the servers go idle and thus jeopardizes the opportunity to collect that fee. 
The optimal fee must balance this tradeoff.


\smallskip
\indent\textbf{Time preferences.} 
The current model follows the standard convention of the queueing literature
by assuming linear waiting cost. This convention is useful for analytical tractability and  comparability with existing
queueing models. It serves another purpose in our model: it isolates the effect of dynamic incentives generated by alternative
queueing rules. Given linear waiting costs, we find that the differences in waiting time
{\it distributions}  across alternative queueing rules matter for agents' dynamic incentives for queuing.  In particular, the fact that FCFS induces the least dispersed waiting times in comparison with other queueing rules helps to minimize the adverse updating from a ``missing'' an early service.  Introducing nonlinear time preferences will confound this effect by rendering the waiting-time
distribution under an queueing rule {\it directly} payoff-relevant. A reasonable conjecture is, though, that risk-averse time preferences will reinforce the optimality of FCFS whereas risk-loving time preferences (such as exponential
discounting) will counteract it. \smallskip

\indent\textbf{Heterogenous preferences.}
Following the standard queueing models, we have assumed that agents have
homogeneous preferences.  It will be interesting to
allow agents to differ in their waiting costs, value of service, or in their
service requirements. Such heterogeneities will introduce the need by the
designer to treat agents differently based on their types, for instance
prioritizing service toward those agents with high waiting costs, high value
of service and small service requirements.\footnote{%
See \cite{Anunrojwong2020} for a simple model of heterogenous waiting
costs---i.e., zero cost and positive costs.} This will again confound the
analysis by making allocation of service priority directly payoff-relevant,
above and beyond making it relevant from the perspective of dynamic
incentives---the central focus of the current study. In particular, if the
agents' characteristics are unobservable, one must deal with additional
incentive issues with screening agents based on this additional
informational asymmetry. Such an extension is therefore beyond the scope of
the current paper. Nevertheless, we expect that the main logic and
thrust of the current paper will extend to such a model. At least within each type of agents, allocating service according to FCFS contributes to their dynamic incentives for queueing, and
will be desirable.

We leave these and other worthy extensions of the current model for future
research.

\bibliographystyle{economet}
\bibliography{references}



\renewcommand{\theequation}{\Alph{section}.\arabic{equation}}

\renewcommand{\thetheorem}{\Alph{section}.\arabic{theorem}}

\renewcommand{\theproposition}{\Alph{section}.\arabic{proposition}}

\renewcommand{\thelemma}{\Alph{section}.\arabic{lemma}}
\renewcommand{\theclaim}{A.\arabic{claim}}

\renewcommand{\thecorollary}{\Alph{section}.\arabic{corollary}}

\renewcommand{\thedefinition}{\Alph{section}.\arabic{definition}}

\renewcommand{\theexample}{\Alph{section}.\arabic{example}}

\renewcommand{\thefootnote}{\Alph{section}.\arabic{footnote}}

\renewcommand{\thetable}{\Alph{section}.\arabic{table}}

\renewcommand{\thefigure}{\Alph{section}.\arabic{figure}}


\appendix

\section*{Appendix}

\section{Proof of \cref{thm:cutoff}}

\label{app:theorem1}


Rewrite problem $[P^{\prime }]$ as:
\begin{equation*}
\max_{p\in M}\sum_{k=0}^{\infty }p_{k}\left[ \mu _{k}((1-\alpha )R+\alpha
V)-\alpha Ck\right] \text{ s.t. } \sum_{k=0}^{\infty }p_{k}\left[ \mu
_{k}V-Ck\right] \geq 0, \leqno{[P']}
\end{equation*}%
%
%
%
where $M\triangleq \{p\in \Delta (\mathbb{Z}_{+}):p\ $satisfies \cref{B'}$\}$%
. (Recall our convention that, $\mu _{0}=0$).

Note that, assuming  $p_{k+1}>0$, \cref{B'} binds at $k$ if and only if \cref{B} is satisfied for $%
x_k=1,z_{k,\ell}=0$ for all $\ell=1,...,k$ and $y_{k+1,\ell}=0$ for all $\ell=1,...,k+1$. This means that an invariant distribution $p$ is generated by a cutoff policy $%
(x,y,z)$ with maximal length $K$ (possibly infinite) if and only if $\supp%
(p)=\{0,...,K\}$ and \cref{B'} binds for all $k=0,...,K-2$ and holds for $%
k=K-1$ (with weak inequality). Hence, in the sequel, if a distribution $p$ satisfies the latter
property, we will simply say that it exhibits a cutoff policy.  Our goal in this section is therefore to show that the above LP problem has an optimal solution
that exhibits that property.

Below we use a Langrangian characterization of the LP problem. Unlike finite
dimensional LP problems, this characterization is not automatically valid in
infinite dimensional LP problems.\footnote{{Countably infinite linear
programs (CILPs) are linear optimization problems with a countably infinite
number of variables and a countably infinite number of constraints. It is
well-known that many of the nice properties of finite dimensional linear
programming may fail to hold in these problems. Indeed, while in finite
dimensional LP problems, zero duality gap is ensured provided that the
primal problem is feasible, necessary conditions for zero duality gap for
CILPs are much more demanding and may often fail. See \cite%
{kipp/ryan/stern:16} and references therein.}} In order to overcome the
difficulty, we first study a finite dimensional truncation of $[P^{\prime }]$
where the state space contains finitely many states, say $K$, where $K$ can
potentially be \textquotedblleft large\textquotedblright . In this
environment, we will show that an optimal solution $p^{K}$ exhibits a cutoff
policy (\Cref{sec:finite dim analysis}). In a second step, we show that as $%
K $ gets large, a limit point of $\{p^{K}\}$ is an optimal solution of $%
[P^{\prime }]$ and exhibits a cutoff policy. The proof of this second step,
in essence, uses a continuity argument---and so uses fairly routine
arguments. Hence it is sketched in \Cref{sec:infinite dim
analysis} but the formal argument is relegated to the online appendix %
\cref{sup: proof final prop}.

\subsection{Finite dimensional analysis}

\label{sec:finite dim analysis}

In the sequel, we fix an integer $K\geq 0$. We consider the following
\textquotedblleft truncated\textquotedblright\ version of $[P^{\prime }]$,
say $[P_{K}^{\prime }]$%
\begin{equation*}
\max_{p\in M_{K}}\sum_{k=0}^{K}p_{k}\left[ \mu _{k}((1-\alpha )R+\alpha
V)-\alpha Ck\right] \text{ s.t. } \sum_{k=0}^{K}p_{k}\left[ \mu _{k}V-Ck%
\right] \geq 0, \leqno{[P'_K]}
\end{equation*}%
%
%
%
where $M_{K}\triangleq \{p\in \Delta (\{0,1,...,K\}):p\ $satisfies $\cref{B'}%
\}$.

Let us fix $\xi \geq 0$ and consider the problem $[\mathcal{L}_{\xi }]$%
\begin{equation*}
\max_{p\in M_{K}}\mathcal{L}(p,\xi )\leqno{[\mathcal{L}_{\xi }]}
\end{equation*}%
where
\begin{eqnarray*}
\mathcal{L}(p,\xi ) &\triangleq &\sum_{k=0}^{K}p_{k}\left[ \mu _{k}(\left(
1-\alpha \right) R+\alpha V)-\alpha Ck\right] +\xi \sum_{k=0}^{K}p_{k}\left[
\mu _{k}V-Ck\right] \\
&=&\sum_{k=0}^{K}p_{k}f(k;\xi ),
\end{eqnarray*}%
where $f(k;\xi )\triangleq \mu _{k}((1-\alpha )R+\left( \alpha +\xi \right)
V)-\left( \alpha +\xi \right) Ck$.

The Lagrangian dual of problem $[P_{K}^{\prime }]$ is taking the $\inf $
over $\xi \geq 0$ of the value of $[\mathcal{L}_{\xi }]$. Since $M_{K}$ is a
convex set, the problem constitutes a finite dimensional linear program, so
strong duality applies. Hence, $p^{\ast }$ is an optimal solution if and
only if there is (a Lagrange multiplier) $\xi ^{\ast }\geq 0$ such that $%
(p^{\ast },\xi ^{\ast })$ is a saddle point of the function $\mathcal{L}%
(\cdot ,\cdot )$, i.e.,
\begin{equation*}
\mathcal{L}(p,\xi ^{\ast })\leq \mathcal{L}(p^{\ast },\xi ^{\ast })\leq
\mathcal{L}(p^{\ast },\xi )
\end{equation*}%
for any $\xi \geq 0$ and $p\in M_{K}$. We fix a saddle point $(p^{\ast },\xi
^{\ast })$ of function $\mathcal{L}(\cdot ,\cdot )$ and show that it
exhibits a cutoff policy.

In this section, we will show a finite-dimensional version of %
\cref{thm:cutoff} stated below.

\begin{proposition}
\label{prop:cutoff finite dim} If $\mu $ is regular, then there is an
optimal solution for $[P_{K}^{\prime }]$ which exhibits a cutoff policy. In
addition, $p _{k}^{\ast }>0$ for each $k\leq \min \{k^{\ast },K\}$ where $%
k^{\ast }\triangleq\min \arg \max f(k;\xi ^{\ast })$.
\end{proposition}

In order to prove this proposition, we need to first establish several
lemmas. To begin, we say a function $f:\mathbb{Z}_{+}\rightarrow \mathbb{R}$
is \textit{single-peaked} if $f(k-1)<f(k)$ for all $k\leq \min \arg
\max_{k\in \mathbb{Z}_{+}}f(k)$ while $f(k)>f(k+1)$ for all $k\geq \max \arg
\max_{k\in \mathbb{Z}_{+}}f(k)$. Our convention is that if $\arg \max_{k\in
\mathbb{Z}_{+}}f(k)$\ is empty, then $\min \arg \max_{k\in \mathbb{Z}%
_{+}}f(k)$ is set to $+\infty $. We now show that the regularity of $\mu $
implies that $f(\cdot ;\xi )$ is single-peaked.

\begin{lemma}
\label{lem: single-peakedness} If $\mu $ is regular, then for any $\xi \geq
0 $, function $f(\cdot ;\xi )$ is single-peaked.
\end{lemma}

\begin{proof}
Fix any $\xi \geq 0$. It is easily checked that $f(\cdot ;\xi )$ is
single-peaked if and only if $f(k;\xi )\geq (>)f(k+1;\xi )$ then $%
f(k^{\prime };\xi )\geq (>)f(k^{\prime }+1;\xi )$ for any $k^{\prime }\geq k$%
. Assume that $f(k;\xi )\geq f(k+1;\xi )$, i.e.,%
\begin{equation*}
\mu _{k}(\left( 1-\alpha \right) R+(\alpha +\xi )V)-(\alpha +\xi )Ck\geq \mu
_{k+1}(\left( 1-\alpha \right) R+(\alpha +\xi )V)-(\alpha +\xi )C(k+1).
\end{equation*}%
Simple algebra shows that this is equivalent to%
\begin{equation*}
\mu _{k+1}-\mu _{k}\leq \frac{(\alpha +\xi )C}{\left( 1-\alpha \right)
R+(\alpha +\xi )V}\text{.}
\end{equation*}%
Since $\mu $ is regular, $\mu _{k+1}-\mu _{k}$ is nonincreasing and so, for $%
k^{\prime }\geq k$, we must have%
\begin{equation*}
\mu _{k^{\prime }+1}-\mu _{k^{\prime }}\leq \mu _{k+1}-\mu _{k}\leq \frac{%
(\alpha +\xi )C}{\left( 1-\alpha \right) R+(\alpha +\xi )V}\text{.}
\end{equation*}%
Hence, $f(k^{\prime };\xi )\geq f(k^{\prime }+1;\xi )$. The same argument
holds to show that $f(k;\xi )>f(k+1;\xi )$ implies $f(k^{\prime };\xi
)>f(k^{\prime }+1;\xi )$ for any $k^{\prime }\geq k$. \end{proof}

We will also use the following lemma.

\begin{lemma}
\label{claim:increasing part} Suppose%
\begin{equation*}
f(\ell ;\xi ^{\ast })<f(\ell +1;\xi ^{\ast })
\end{equation*}%
for some $\ell \leq K-1$. Then, $\lambda _{\ell }p_{\ell }^{\ast }=\mu
_{\ell +1}p_{\ell +1}^{\ast }$.
\end{lemma}

\begin{proof}
Fix $\ell $ satisfying the properties of the lemma. Since $p^{\ast }\ $is an
optimal solution of $[P_{K}^{\prime }]$---and so satisfies $\cref{B'}$---we know that $\mu _{\ell +1}p_{\ell
+1}^{\ast }\leq \lambda _{\ell }p_{\ell }^{\ast }$. Toward a contradiction,
assume that $\mu _{\ell +1}p_{\ell +1}^{\ast }<\lambda _{\ell }p_{\ell
}^{\ast }$. Now, simply consider $\hat{p}$ defined as%
\begin{equation*}
\hat{p}_{k}=\left\{
\begin{array}{c}
p_{k}^{\ast }+\varepsilon \text{ if }k=\ell +1 \\
p_{k}^{\ast }-\varepsilon \text{ if }k=\ell  \\
p_{k}^{\ast }\text{ otherwise}%
\end{array}%
\right.
\end{equation*}%
and note that we can choose $\varepsilon >0$ so that $\mu _{\ell +1}\hat{p}%
_{\ell +1}=\lambda _{\ell }\hat{p}_{\ell }$ while ensuring $\hat{p}_{\ell },%
\hat{p}_{\ell +1}\in (0,1)$.\footnote{%
Indeed, at $\varepsilon =0$, we have $\mu _{\ell +1}\hat{p}_{\ell
+1}<\lambda _{\ell }\hat{p}_{\ell }$. In addition, for $\varepsilon =p_{\ell
}>0$ we have $\hat{p}_{\ell +1}=p_{\ell +1}+\varepsilon =p_{\ell +1}+p_{\ell
}\leq 1$ and $\mu _{\ell +1}\hat{p}_{\ell +1}>\lambda _{\ell }\hat{p}_{\ell
}=0$. Hence, by the Intermediate Value Theorem, there must exist $%
\varepsilon \in (0,p_{\ell })$ so that $\mu _{\ell +1}\hat{p}_{\ell
+1}=\lambda _{\ell }\hat{p}_{\ell }$ and $\hat{p}_{\ell },\hat{p}_{\ell +1}$
are in $(0,1)$.} Clearly, $\sum_{k=0}^{K}\hat{p}_{k}=1.$ Now, let us show
that $\mu _{k+1}\hat{p}_{k+1}\leq \lambda _{k}\hat{p}_{k},\forall k=0,...K-1$%
. Since these inequalities holds at $p^{\ast }$ (because $p^{\ast }$ is an
optimal solution of $[P_{K}^{\prime }]$ and so satisfies $\cref{B'}$), by construction of $\hat{p}$, we
only need to check this constraint for $k=\ell +1$ and $k=\ell -1$. For $%
k=\ell +1$, we have%
\begin{equation*}
\mu _{\ell +2}\hat{p}_{\ell +2}=\mu _{\ell +2}p_{\ell +2}^{\ast }\leq
\lambda _{\ell +1}p_{\ell +1}^{\ast }\leq \lambda _{\ell +1}\hat{p}_{\ell
+1}.
\end{equation*}%
Similarly, for $k=\ell -1$,%
\begin{equation*}
\mu _{\ell }\hat{p}_{\ell }\leq \mu _{\ell }p_{\ell }^{\ast }\leq \lambda
_{\ell -1}p_{\ell -1}^{\ast }=\lambda _{\ell -1}\hat{p}_{\ell -1}.
\end{equation*}%
Now, we show that the value of the objective of $[\mathcal{L}_{\xi ^{\ast }}]
$ strictly increases when we replace solution $p^{\ast }$ by $\hat{p}$. We
have
\begin{eqnarray*}
\sum_{k=0}^{K}\hat{p}_{k}f(k;\xi ^{\ast })-\sum_{k=0}^{K}p_{k}^{\ast
}f(k;\xi ^{\ast }) &=&\hat{p}_{\ell }f(\ell ;\xi ^{\ast })-p_{\ell }^{\ast
}f(\ell ;\xi ^{\ast })+\hat{p}_{\ell +1}f(\ell +1;\xi ^{\ast })-p_{\ell
+1}^{\ast }f(\ell +1;\xi ^{\ast }) \\
&=&-\varepsilon f(\ell ;\xi ^{\ast })+\varepsilon f(\ell +1;\xi ^{\ast
})=\varepsilon \left( f(\ell +1;\xi ^{\ast })-f(\ell ;\xi ^{\ast })\right) >0
\end{eqnarray*}%
where the inequality comes from the assumption in the lemma. To conclude, we
must have that $\mathcal{L}(\hat{p},\xi ^{\ast })>\mathcal{L}(p^{\ast },\xi
^{\ast })$ which contradicts the fact that $(p^{\ast },\xi ^{\ast })$ is a
saddle point of the function $\mathcal{L}(\cdot ,\cdot )$. \end{proof}

Finally, in the proof of \cref{prop:cutoff finite dim}, we will need the
following simple lemma which proof is relegated to \cref{sup: stoch dom} of
the online appendix.

\begin{lemma}
\label{lem: charact stoch dom} Assume that $p ^{\prime }$ stochastically
dominates $p $. Let $\varphi $ be a nondecreasing function. If there is $%
\kappa $ such that
\begin{equation*}
\sum_{k=\kappa }^{K}p _{k}^{\prime }>\sum_{k=\kappa }^{K}p _{k}
\end{equation*}%
and $\varphi (\kappa )>\varphi (\kappa -1)$ then
\begin{equation*}
\sum_{k=0}^{K}p _{k}^{\prime }\varphi (k)>\sum_{k=0}^{K}p _{k}\varphi (k)%
\text{.}
\end{equation*}
\end{lemma}

\begin{proof}
See \cref{sup: stoch dom} in the online appendix.
\end{proof}

\begin{proof}[Proof of \cref{prop:cutoff finite dim}]
Before proceeding, we make the following straightforward observations (1) $%
p_{0}^{\ast }>0$ (or else $p_{k}^{\ast }=0$ for all $k$ because, by
construction of $M_{K}$, $p$ satisfies $\cref{B'}$; this contradicts the
assumption that $p$ is a probability measure); (2) for all $\xi $, $f(0;\xi
)=0$. Using these two facts, we claim that \cref{prop:cutoff finite
dim} holds whenever $f(k;\xi ^{\ast })=f(k^{\prime };\xi ^{\ast })$ for all $%
k,k^{\prime }$ in the support of $p^{\ast }$. Indeed, since $p_{0}^{\ast }>0$%
, $f(k;\xi ^{\ast })=0$ for all states $k$ in the support of $p^{\ast }$. In
that case, $\sup_{p}\mathcal{L}(p,\xi ^{\ast })=0$. Thus, the value of the
problem $[P_{K}^{\prime }]$ is $0$. Clearly, the distribution $p$
corresponding to the Dirac measure on state $0$ yields the same value and is
a cutoff policy.$\ $Hence, in this very special case, \cref{thm:cutoff}\
holds true. Thus, in the sequel, we assume that there is a pair of states $k$
and $k^{\prime }$\ in the support of $p^{\ast }$ satisfying $f(k;\xi ^{\ast
})\neq f(k^{\prime };\xi ^{\ast })$.

Let $k^{\ast }$ be $\min \arg \max_{k}f(k;\xi ^{\ast })$ and $k^{\ast \ast }$
be $\max \arg \max_{k}f(k;\xi ^{\ast })$. Recall that $k^{\ast }$ can be
equal to $+\infty $. By \cref{lem: single-peakedness}, we know that $f(k;\xi
^{\ast })$ is strictly increasing up to $k^{\ast }$. Hence, %
\cref{claim:increasing part} implies that $\mu _{k}p_{k}^{\ast }=\lambda
_{k-1}p_{k-1}^{\ast }$ for each $k\leq \min \{k^{\ast },K\}$. Note that
(since $p_{0}^{\ast }>0$) this also implies that $p_{k}^{\ast }>0$ for each $%
k\leq \min \{k^{\ast },K\}$, as stated in \cref{prop:cutoff finite dim}. If $%
K\leq k^{\ast }$, we are done. So assume from now on that $K>k^{\ast }$;
note that this implies that $k^{\ast }<+\infty $.
By means of contradiction, let us assume that $p^{\ast }$ does not exhibit a
cutoff policy. This means that there is $k_{0}>k^{\ast }$ such that $\mu
_{k_{0}}p_{k_{0}}^{\ast }<\lambda _{k_{0}-1}p_{k_{0}-1}^{\ast }$ and $%
p_{k_{0}+1}^{\ast }>0$ (hence, $p_{k_{0}}^{\ast }>0$).\footnote{Indeed, given the above, by definition, $p^*$ exhibits a cutoff policy
if and only if $\mu_{k_{0}}p_{k_{0}}^{\ast }=\lambda _{k_{0}-1}p_{k_{0}-1}^{\ast }$ for all $k_0=k^*+1,\cdots K-1$, i.e., \cref{B'} binds for all $k=0,...,K-2$.} Without loss, assume
that for any $k<k_{0},$ we have $\mu _{k}p_{k}^{\ast }=\lambda
_{k-1}p_{k-1}^{\ast }$. We consider two cases.

\underline{Case 1 $:p_{k}^{\ast }>0$ for some $k>k^{\ast \ast }$.} Toward a
contradiction, we construct a $\hat{p}$ that would achieve a strictly higher
value than $p^{\ast }$ in $[\mathcal{L}_{\xi ^{\ast }}]$. Let $\hat{p}%
_{k}=p_{k}^{\ast }$ for $k\leq k_{0}-1$. For each $k\geq k_{0}$, build $\hat{%
p}$ inductively so that $\mu _{k_{0}}\hat{p}_{k_{0}}=\lambda _{k_{0}-1}\hat{p%
}_{k_{0}-1}$, $\mu _{k_{0}+1}\hat{p}_{k_{0}+1}=\lambda _{k_{0}}\hat{p}%
_{k_{0}}$... Since the total mass of $\hat{p}$ must be $1$, this may be
possible only up to a point $\hat{K}$ where, by construction, we will have $%
\mu _{\hat{K}}\hat{p}_{\hat{K}}\leq \lambda _{\hat{K}-1}\hat{p}_{\hat{K}-1}$%
. Finally, we set $\hat{p}_{k}=0$ for all $k>\hat{K}$. In order to show that
$\hat{p}\ $lies in $\Delta (\{0,1,...,K\})$, we need to show that $\hat{K}%
\leq K$. By a simple induction argument, $\hat{p}_{k}\geq p_{k}^{\ast }$ for
all $k\leq \hat{K}-1$ and so we must have that $\hat{K}\leq K$. To recap,
there is $\hat{K}\geq k_{0}$ (potentially equal to $K$) such that $\mu _{k}%
\hat{p}_{k}=\lambda _{k-1}\hat{p}_{k-1}$ for $k=0,...,\hat{K}-1$, and $\hat{p%
}_{k}=0$ for $k>\hat{K}$. One can show inductively that $\hat{p}%
_{k}>p_{k}^{\ast }$ for all $k=k_{0},...,\hat{K}-1$ while, by construction, $%
\hat{p}_{k}=p_{k}^{\ast }$ for all $k\leq k_{0}-1$. We claim that
distribution $p^{\ast }$ stochastically dominates distribution $\hat{p}$. To
see this, fix any $\kappa >\hat{K}$. Clearly, $\sum_{k=\kappa }^{K}\hat{p}%
_{k}=0\leq \sum_{k=\kappa }^{K}p_{k}^{\ast }$. Now, fix $\kappa \leq \hat{K}$%
.
\begin{equation}
\sum_{k=\kappa }^{K}\hat{p}_{k}=1-\sum_{k=0}^{\kappa -1}\hat{p}_{k}\leq
1-\sum_{k=0}^{\kappa -1}p_{k}^{\ast }=\sum_{k=\kappa }^{K}p_{k}^{\ast }
\label{eq:stoch dom}
\end{equation}%
where the inequality uses the fact that $\hat{p}_{k}\geq p_{k}^{\ast }$ for
all $k=0,...,\kappa -1$. Importantly, the above inequality is strict for all
$\kappa \in \{k_{0}+1,...,\hat{K}\}$ since $\hat{p}_{k}>p_{k}^{\ast }$ for
all $k=k_{0},...,\hat{K}-1$.\footnote{%
Recall that, by construction, $k_{0}+1\leq \hat{K}$.} It is also strict for
any $\kappa \geq \hat{K}+1$ as long as $p_{\kappa }^{\ast }>0$ since in that
case the LHS is simply $0$ while the RHS\ is strictly positive. In
particular, given our assumption that $p_{k}^{\ast }>0$ for some $k>k^{\ast
\ast }$, it must be that $p_{k^{\ast \ast }+1}^{\ast }>0$. Consequently,%
\begin{equation}
\sum_{k=\kappa }^{K}\hat{p}_{k}<\sum_{k=\kappa }^{K}p_{k}^{\ast }
\label{strict:stoch dom}
\end{equation}%
for $\kappa =\max \{k_{0}+1,k^{\ast \ast }+1\}$.

Now, we show that the value of the objective in $[\mathcal{L}_{\xi ^{\ast }}]
$ strictly increases when we replace solution $p^{\ast }$ by $\hat{p}$. We
have to show that%
\begin{equation*}
\sum_{k=0}^{K}\hat{p}_{k}f(k;\xi ^{\ast })>\sum_{k=0}^{K}p_{k}^{\ast
}f(k;\xi ^{\ast })\text{.}
\end{equation*}%
Since $\hat{p}_{k}=p_{k}^{\ast }$ for all $k\leq k_{0}-1$, this is
equivalent to showing%
\begin{equation}
\sum_{k=k_{0}}^{K}\hat{p}_{k}f(k;\xi ^{\ast })>\sum_{k=k_{0}}^{K}p_{k}^{\ast
}f(k;\xi ^{\ast })  \label{increasing obj}
\end{equation}

{Now, define a function $\varphi :\mathbb{Z}_{+}\rightarrow \mathbb{R}$ as
follows}%
\begin{equation*}
{\varphi (k)=\left\{
\begin{array}{c}
{f(k_{0};\xi ^{\ast })}\text{ if }k\leq k_{0}-1 \\
{f(k;\xi ^{\ast })}\text{ if }k\geq k_{0}\text{.}%
\end{array}%
\right. }
\end{equation*}%
Since $k_{0}>k^{\ast }$, by \cref{lem: single-peakedness}, this function is
weakly decreasing and it is strictly decreasing from $k$ to $k+1$ for any $%
k\geq \max \{k_{0},k^{\ast \ast }\}$. Thus, $\varphi (\kappa -1)>\varphi
(\kappa )$ for $\kappa =\max \{k_{0}+1,k^{\ast \ast }+1\}$. Now, we know
that $p^{\ast }$ stochastically dominates $\hat{p}$, that inequality %
\cref{strict:stoch dom} holds at $\kappa =\max \{k_{0}+1,k^{\ast \ast }+1\}$%
. and that $\varphi (\kappa -1)>\varphi (\kappa )$. Applying
\cref{lem:
charact stoch dom},%
\begin{equation*}
\sum_{k=0}^{K}\left( \hat{p}_{k}-p_{k}^{\ast }\right) \varphi (k)>0\text{.}
\end{equation*}%
Since $\hat{p}_{k}=p_{k}^{\ast }$ for all $k\leq k_{0}-1$, this is
equivalent to Equation \cref{increasing obj}. To conclude, $\mathcal{L}(\hat{%
p},\xi ^{\ast })>\mathcal{L}(p^{\ast },\xi ^{\ast })$ which contradicts the
fact that $(p^{\ast },\xi ^{\ast })$ is a saddle point of $\mathcal{L}(\cdot
,\cdot )$.

\bigskip

\underline{Case 2 $:p_{k}^{\ast }=0$ for all $k>k^{\ast \ast }.$} Recall our
assumption that there is a pair of states $k$ and $k^{\prime }$\ in the
support of $p^{\ast }$ satisfying $f(k;\xi ^{\ast })\neq f(k^{\prime };\xi
^{\ast })$. Hence, because $f(\cdot ;\xi ^{\ast })$ is single-peaked, $f$
must be weakly increasing on the support of $p^{\ast }$ and strictly
increasing from $k$ to $k+1$ for all $k<k^{\ast }$. In particular, this
holds at $k=0$, and so we have $f(0;\xi ^{\ast })<f(1;\xi ^{\ast })$ and $%
p_{0}^{\ast }>0$.\ Recall that $k_{0}$ is the smallest $k$ in $\{k^{\ast
}+1,...,k^{\ast \ast }-1\}$ such that $\mu _{k}p_{k}^{\ast }<\lambda
_{k-1}p_{k-1}^{\ast }$ and $p_{k+1}^{\ast }>0$.\ We now construct a measure $%
\hat{p}$ as follows

\begin{equation*}
\hat{p}_{k}=\left\{
\begin{array}{c}
p_{k}^{\ast }/Z_{1}\text{ if }k\leq k_{0}-1 \\
p_{k}^{\ast }+Z_{2}\text{ if }k=k_{0} \\
p_{k}^{\ast }\text{ if }k\geq k_{0}+1\text{,}%
\end{array}%
\right.
\end{equation*}
where $Z_{1}>1$ and $Z_{2}\triangleq \sum_{k=0}^{k_{0}-1}(p_{k}^{\ast }-\hat{%
p}_{k})$ so that $\hat{p}$ sums up to $1$. We pick $Z_{1}$ small enough so
that $\hat{p}_{k_{0}}$ remains between $0$ and $1$ for each $k$. We show
that, for $Z_{1}>1$ small enough, for each $k\leq K,$ $\mu _{k}\hat{p}%
_{k}\leq \lambda _{k-1}\hat{p}_{k-1}$. To see this, first fix $k\leq k_{0}-1$
and note that
\begin{equation*}
\mu _{k}\hat{p}_{k}=\mu _{k}p_{k}^{\ast }/Z_{1}\leq \lambda
_{k-1}p_{k-1}^{\ast }/Z_{1}=\lambda _{k-1}\hat{p}_{k-1}
\end{equation*}%
where the inequality follows from the fact that $p^{\ast }$ is a feasible
solution of $[P_{K}^{\prime }]$. Next,%
\begin{equation*}
\mu _{k_{0}}\hat{p}_{k_{0}}=\mu _{k_{0}}\left( p_{k_{0}}^{\ast
}+Z_{2}\right) \leq \lambda _{k_{0}-1}p_{k_{0}-1}^{\ast }/Z_{1}=\lambda
_{k_{0}-1}\hat{p}_{k_{0}-1}
\end{equation*}%
where the inequality holds if $Z_{1}$ is small enough since, by assumption, $%
\mu _{k_{0}}p_{k_{0}}^{\ast }<\lambda _{k_{0}-1}p_{k_{0}-1}^{\ast }$ (and $%
Z_{2}$ vanishes as $Z_{1}$ goes to $1$).\footnote{%
Indeed, by construction, for each $k\leq k_{0}-1$, $\hat{p}_{k}\rightarrow
p_{k}^{\ast }$ $\ $as $Z_{1}\rightarrow 1$. Since $Z_{2}=%
\sum_{k=0}^{k_{0}-1}(p_{k}^{\ast }-\hat{p}_{k})$, $Z_{2}$ converges to $0$
as $Z_{1}\rightarrow 1$. } Now, for $k=k_{0}+1$, we have
\begin{equation*}
\mu _{k_{0}+1}\hat{p}_{k_{0}+1}=\mu _{k_{0}+1}p_{k_{0}+1}^{\ast }\leq
\lambda _{k_{0}}p_{k_{0}}^{\ast }\leq \lambda _{k_{0}}(p_{k_{0}}^{\ast
}+Z_{2})=\lambda _{k_{0}}\hat{p}_{k_{0}}\text{.}
\end{equation*}%
Finally, by construction, for any $k>k_{0}+1$, $\mu _{k}\hat{p}_{k}\leq
\lambda _{k-1}\hat{p}_{k-1}$ must hold since $p^{\ast }$ and $\hat{p}$
coincide.

Now, we show that the value of the objective in $[\mathcal{L}_{\xi ^{\ast }}]
$ strictly increases when we replace solution $p^{\ast }$ by $\hat{p}$. To
see this, observe first that $\hat{p}$ must stochastically dominate $p^{\ast
}$. Indeed, fix any $\kappa >k_{0}$. Clearly, since $\hat{p}_{k}=p_{k}^{\ast
}$ for all $k\geq k_{0}+1$, $\sum_{k=\kappa }^{K}\hat{p}_{k}=\sum_{k=\kappa
}^{K}p_{k}^{\ast }$. Now, fix $\kappa \leq k_{0}$.
\begin{equation}
\sum_{k=\kappa }^{K}\hat{p}_{k}=1-\sum_{k=0}^{\kappa -1}\hat{p}%
_{k}>1-\sum_{k=0}^{\kappa -1}p_{k}^{\ast }=\sum_{k=\kappa }^{K}p_{k}^{\ast }
\label{strict stoch dom 2}
\end{equation}%
where the inequality uses the fact that $\hat{p}_{k}=p_{k}^{\ast
}/Z_{1}<p_{k}^{\ast }$ for all $k=0,...,\kappa -1$ (since $Z_{1}>1$ and $%
p_{k}^{\ast }>0$ for such $k$). Now, we show that the value of the objective
in $[\mathcal{L}_{\xi ^{\ast }}]$ strictly increases when we replace
solution $p^{\ast }$ by $\hat{p}$, i.e.,
\begin{equation*}
\sum_{k=0}^{K}\hat{p}_{k}f(k;\xi ^{\ast })>\sum_{k=0}^{K}p_{k}^{\ast
}f(k;\xi ^{\ast })\text{.}
\end{equation*}%
We know that $\hat{p}$ stochastically dominates $p^{\ast }$, that inequality %
\cref{strict stoch dom 2} holds at $\kappa =1$ and that $f(0;\xi ^{\ast
})<f(1;\xi ^{\ast })$. In addition, $f(\cdot ;\xi ^{\ast })$ is
nondecreasing on the support of $p^{\ast }$ and $\hat{p}$. Hence, this
follows from \cref{lem: charact stoch dom}. \end{proof}

\subsection{Infinite dimensional analysis}

\label{sec:infinite dim analysis}

Let us consider the sequence $\{ p ^{K}\}_{K}$ where for each $K$, $p ^{K} $
is an optimal solution of problem $[P_{K}^{\prime }]$. If $\mu$ is regular,
we assume each $p ^{K} $ exhibits a cutoff policy which is well-defined by
\cref{prop:cutoff finite
dim}. For each $K$, we see $p ^{K}$ as a point in $\mathbb{R}^{\mathbb{Z}%
_{+}}$ with value $0$ on states weakly greater than $K+1$. We will be
interested in the limit points of sequence $\{ p ^{K}\}_{K}$. Together with
the result showing that $[P^{\prime }]$ has an optimal solution, the
following statement implies \cref{thm:cutoff}.

\begin{proposition}
\label{prop: final} Assume $\mu $ is regular. Sequence $\{ p ^{K}\}_{K}$ has
a subsequence which converges to a distribution $p ^{\ast }$ which is an
optimal solution to $[P^{\prime }]$ and exhibits a cutoff policy. Further,
it satisfies $p _{k}^{\ast }>0$ for each $k\leq \min \arg \max_{k}\mu
_{k}V-Ck$.
\end{proposition}

This result is shown in the online appendix \cref{sup: proof final prop}
through the following steps. First, we show that the infinite-dimensional
problem $[P^{\prime }]$ admits an optimal solution (%
\cref{thm: existence
infinite}). Then, we show that the set of feasible distributions of $%
[P^{\prime }]$ exhibiting a cutoff-policy is sequentially compact, which in
turn implies that (when $\mu$ is regular) \ $\{p^{K}\}_{K}$ has a
subsequence converging to a point which exhibits a cutoff policy (%
\cref{lem:compactness
final}). Finally, we argue that any limit point of $\{p^{K}\}_{K}$ must be
an optimal solution of $[P^{\prime }]$ (\cref{prop:final step}).

\section{Proofs from \Cref{sec:FCFS}: FCFS with No Information}

\subsection{Proof of \cref{lem:waiting-time-FCFS}}

\label{app:waiting-time-FCFS} The expected waiting time satisfies the following recursion. The agent in the first position has expected waiting
time
\begin{align*}
\tau^*_{1} & =( q^*_{1}dt)dt + [1-q^*_{1}dt](\tau^*_{1}+dt) +o(dt),
\end{align*}
since he waits for $dt$ period with probability $q^*_{1}dt$ and for $%
\tau^*_{1}+dt$ periods with the remaining probability. Letting $dt\to 0$, we
get
\begin{equation*}
\tau^*_1= 1/q_1^*=1/\mu _1.
\end{equation*}

More generally, the agent in queue position $\ell$ waits for
\begin{align*}
\tau^*_{\ell} & =( q^*_{\ell}dt)dt + \left[1-\sum_{j=1}^{\ell} q^*_j dt%
\right](\tau^*_{\ell}+dt) +\left(\sum_{j=1}^{\ell-1}
q^*_jdt\right)(\tau^*_{\ell-1}+dt) +o(dt),
\end{align*}
since he is served in $dt$ period with probability $q^*_{\ell}dt$, in $%
\tau^*_{\ell}+dt$ periods with probability $1-\sum_{j=1}^{\ell} q^*_j dt$
(when nobody before him is served), and in $\tau^*_{\ell-1}+dt$ periods with
probability $\sum_{j=1}^{\ell-1} q^*_jdt$ (when somebody before him is
served).\footnote{Again, the probability that multiple agents are served during $[t,
t+dt)$ has a lower order of magnitude denoted by $o(dt)$.}

The recursion equations yield a unique solution:
\begin{equation*}
\tau^*_{\ell}= \frac{\ell}{\sum_{j=1}^{\ell} q^*_j}= \frac{\ell}{ \mu _{
\ell}},
\end{equation*}
where the last equality follows from feasibility.

Part (ii) of regularity implies that $q^*_{\ell}$ is nonincreasing in $\ell$%
. Therefore, for each $\ell$
\begin{equation*}
\tau^*_{\ell+1}-\tau^*_{\ell}=\frac{\sum_{j=1}^{\ell} q^*_j-\ell
q^*_{\ell+1} }{(\sum_{j=1}^{\ell} q^*_j)(\sum_{j=1}^{\ell+1} q^*_j)} \ge 0.
\end{equation*}
Hence, it follows that $\tau^*_{\ell}$ is nonincreasing in $\ell$. Further,
if $2\mu_1>\mu_2$, then $q_1^*>q^*_2\ge q^*_{\ell}$ for all $\ell\ge 2$.
Then, the above inequality becomes strict for all $\ell$, which proves the
last statement. \qed

\subsection{Proof of \cref{lem:dyn-incentives}}

\label{app:dyn-incentives}

We let $\bar{K}$ be the largest state in the support of $p^{\ast }$ (which
can potentially be infinite). We first study the dynamics for the case with $%
\bar{K}<\infty $. For $\bar{K}=\infty $, we show that the dynamics can be
approximated by the dynamics for $\bar{K}<\infty $ when $\bar{K}$ goes to
infinity.\ While it requires some care, the argument for $\bar{K}=\infty $
essentially relies on the case with $\bar{K}<\infty $. Hence, we defer the
proof to online appendix \Cref{online_app:infty}, which also derives the recursion equation for belief evolution more rigorously. In the sequel, we assume
that $\bar{K}<\infty $.


Using \cref{Eq: cond beliefs}, we write for each such $\ell \geq 2$,
\begin{equation*}
r_{\ell }^{t+dt}=\frac{\tilde{\gamma}_{\ell }^{t+dt}}{\tilde{\gamma}_{\ell
-1}^{t+dt}}=\frac{(1-\mu _{\ell }dt)\tilde{\gamma}_{\ell }^{t}+\mu _{\ell }dt%
\tilde{\gamma}_{\ell +1}^{t}}{(1-\mu _{\ell -1}dt)\tilde{\gamma}_{\ell
-1}^{t}+\mu _{\ell -1}dt\tilde{\gamma}_{\ell }^{t}}+o(dt)=\frac{1-\mu _{\ell
}dt+\mu _{\ell }dtr_{\ell +1}^{t}}{(1-\mu _{\ell -1}dt)\frac{1}{r_{\ell }^{t}%
}+\mu _{\ell -1}dt}+o(dt).
\end{equation*}%
Rearranging, we get
\begin{equation*}
\frac{r_{\ell }^{t+dt}-r_{\ell }^{t}}{dt}=\frac{\mu _{\ell -1}-\mu _{\ell
}-\mu _{\ell -1}r_{\ell }^{t}+\mu _{\ell }r_{\ell +1}^{t}}{(1-\mu _{\ell
-1}dt)\frac{1}{r_{\ell }^{t}}+\mu _{\ell -1}dt}+o(dt)/dt.
\end{equation*}%
Letting $dt\rightarrow 0$, we obtain%
\begin{equation}
\dot{r}_{\ell }^{t}=r_{\ell }^{t}\left( \mu _{\ell -1}-\mu _{\ell }-\mu
_{\ell -1}r_{\ell }^{t}+\mu _{\ell }r_{\ell +1}^{t}\right) .  \label{eq:r-t}
\end{equation}%
\cref{eq:r-t} forms a system of ordinary differential equations. The
boundary condition is defined as follows. Recall that the effective arrival
rate be $\tilde{\lambda}_{k}\triangleq \lambda _{k}x_{k}^{\ast }$ for each $%
k $. For $\ell \leq \bar{K}$,
\begin{equation}
r_{\ell }^{0}=\frac{\tilde{\gamma}_{\ell }^{0}}{\tilde{\gamma}_{\ell -1}^{0}}%
=\frac{p_{\ell }^{\ast }\mu _{\ell }}{p_{\ell -1}^{\ast }\mu _{\ell -1}}=%
\frac{\tilde{\lambda}_{\ell -1}}{\mu _{\ell -1}},  \label{Eq r0}
\end{equation}%
where the second equality uses the fact that $\tilde{\gamma}_{\ell
}^{0}=p_{\ell }^{\ast }\mu _{\ell }\left\backslash \sum_{i=1}^{\infty
}p_{i}^{\ast }\mu _{i}\right. $ for each $\ell $, while the third one uses %
\cref{B} whereby $\frac{p_{\ell }^{\ast }}{p_{\ell -1}^{\ast }}=\frac{\tilde{%
\lambda}_{\ell -1}}{\mu _{\ell }}$.\footnote{One can obtain the expression for
$\tilde{\gamma}_{\ell}^{0}$ as follows. The optimality of the cutoff policy means $%
x^*_k=1$ for all $k=0,..., K^*-2$, $x^*_k=0$ for all $k> K^*-1$, and $%
y^*_{k,\ell}=z^*_{k,\ell}=0$ for all $(k,\ell)$. Substituting these into %
\cref{B}, one obtains the expression by rewriting \cref{eq:belief0}.} It is routine to see that the system of
ODEs \cref{eq:r-t} together with the boundary condition \cref{Eq r0} admits
a unique solution $(r_{\ell }^{t})_{\ell }$ for all $t\geq 0$.\footnote{%
This follows from the observation that the RHS of \cref{eq:r-t} is locally
Lipschitzian in $r$ (a fact implied by the continuous differentiability of
RHS in $r_{\ell }^{t}$'s). See Hale p. 18, Theorem 3.1, for instance.}

We first claim that $\dot{r}_{\ell }^{0}\leq {0}$ for all $\ell =2,...,\bar{K%
}$. It follows from \cref{eq:r-t} that, for $\ell =2,...,\bar{K}$, $\dot{r}%
_{\ell }^{0}\leq 0$ if and only if
\begin{equation}
\mu _{\ell -1}-\mu _{\ell }\leq \mu _{\ell -1}r_{\ell }^{0}-\mu _{\ell
}r_{\ell +1}^{0}.  \label{Eq r1}
\end{equation}

Consider any $\ell =2,...,\bar{K}$. Substituting \cref{Eq r0} into
\cref{Eq
r1}, the condition simplifies to:
\begin{equation*}
\mu _{\ell -1}-\mu _{\ell }\leq \tilde{\lambda}_{\ell -1}-\tilde{\lambda}%
_{\ell },
\end{equation*}%
which holds by regularity of $(\lambda ,\mu )$ and the fact that $%
x_{k}^{\ast }$ is nonincreasing in $k$.

Having established that $\dot{r}_{\ell }^{0}\leq 0$ for each $\ell =2,...,%
\bar{K}$, we next prove that $\dot{r}_{\ell }^{t}\leq 0$ for all $t>0$. To
this end, suppose this is not the case. Then, there exists
\begin{equation*}
\ell \in \arg \min_{\ell ^{\prime }=2,...,\bar{K}}T_{\ell ^{\prime }},
\end{equation*}%
where
\begin{equation*}
T_{\ell ^{\prime }}\triangleq \inf \{t^{\prime }:\dot{r}_{\ell ^{\prime
}}^{t^{\prime }}>0\}
\end{equation*}%
if the infimum is well defined, or else $T_{\ell ^{\prime }}\triangleq
\infty $. Let $t=T_{\ell }<\infty ,$ by the hypothesis. Then, we must have
\begin{equation*}
\ddot{r}_{\ell }^{t}>0;\dot{r}_{\ell ^{\prime }}^{t}\leq 0,\forall \ell
^{\prime }\neq \ell ;\mbox{ and }\dot{r}_{\ell }^{t}=0.
\end{equation*}%
Differentiating \cref{eq:r-t} on both sides, we obtain
\begin{equation*}
0<\ddot{r}_{\ell }^{t}=\dot{r}_{\ell }^{t}\left( \mu _{\ell -1}-\mu _{\ell
}-\mu _{\ell -1}r_{\ell }^{t}+\mu _{\ell }r_{\ell +1}^{t}\right) -r_{\ell
}^{t}(\mu _{\ell -1}\dot{r}_{\ell }^{t}-\mu _{\ell }\dot{r}_{\ell
+1}^{t})=r_{\ell }^{t}\mu _{\ell }\dot{r}_{\ell +1}^{t}\leq 0,
\end{equation*}%
a contradiction. We thus conclude that $\dot{r}_{\ell }^{t}\leq 0$, for all $%
\ell =2,...,\bar{K}$, for all $t\geq 0$.


\subsection{Proof of \cref{thm:dyn-fcfs}}

This theorem is a consequence of \cref{lem:dyn-incentives}.  Indeed, it suffices to prove that, under FCFS with no information,  $(IC_{t})$ holds for all $t\ge 0$.  Note first that, as we already stated (see \cref{lem:ic0-fcfs} in the online appendix), $(IC_{0})$ holds.  Next consider
	$(IC_{t})$ for any $t> 0$.
	\cref{lem:dyn-incentives} proves that $r^t_{\ell}\le r^{0}_{\ell}$ for each $\ell$.  Since $\tau^*_{\ell}$ is nondecreasing in $\ell$ (\cref{lem:waiting-time-FCFS}), this  means that	$$ \sum_{\ell=1}^{K^*} \tilde{\gamma}_{\ell}^{t} \cdot \tau^*_{\ell}\le \sum_{\ell=1}^{K^*} \tilde{\gamma}_{\ell}^{0} \cdot \tau^*_{\ell},$$
	so we have
	$$ V- C\sum_{\ell=1}^{K^*} \tilde{\gamma}_{\ell}^{t} \cdot \tau^*_{\ell}\ge V-C\sum_{\ell=1}^{K^*} \tilde{\gamma}_{\ell}^{0} \cdot \tau^*_{\ell}\ge 0, $$
	where the last inequality follows from  $(IC_{0})$ being satisfied.   Hence,  $(IC_{t})$ holds for any $t>0$.   



\section{Proof of \cref{thm:necessity}}

\label{app:critical-pi}

Fix a queuing rule $q$ which differs from FCFS. We consider the information
policy that provides no information (beyond the recommendations) for all $%
t\geq 0$. This is without loss since, if a queueing rule $q$ fails $(IC_{t})$%
, for some $t\geq 0$, under no information, it would fail $(IC_{t})$ under
\emph{any} information policy.

Recall that we have fixed the service rate $\mu $. While arrival rate $%
\lambda $ is yet to be fixed, for each $\lambda $, we can choose parameters $%
V,C$ and $\alpha $ to ensure that the optimal outcome $(x^{\ast
},y^{\ast },z^{\ast },p^{\ast })$ (i) involves a maximal length $K^{\ast }=2$
(i.e., $x_{2}^{\ast }=0$ or $z_{2,1}^{\ast }+z_{2,2}^{\ast }=1$), (ii) no
rationing at $k=1$ (i.e., $x_{1}^{\ast }=1$ and $z_{1,1}^{\ast }=0$), and
(iii) \cref{IR} is binding at $p^{\ast }$.\footnote{%
If $V/C=\frac{2\lambda +\mu }{(\lambda +\mu )\mu }$ and $\alpha =0$, one can
easily show that there is a unique optimal solution $p$ to $[P^{\prime }]$
and any outcome $(x,y,z)$ implementing $p$ satisfies (i), (ii) and (iii).}
Importantly, assumption (ii) implies that $y_{k,\ell }^{\ast }$ are all
zeros.\footnote{%
Indeed, in that case, $x_{0}^{\ast }=x_{1}^{\ast }=1\ $and $\sum_{\ell
=1}^{0}z_{0,\ell }^{\ast }=\sum_{\ell =1}^{1}z_{1,\ell }^{\ast }=0$.
Further, $(x^{*},y^{*},z^{*},p^{*})$ satisfies \cref{B}, i.e., for each $k$%
\begin{equation*}
p_{k}^{*}\lambda _{k}x_{k}^{*}(1-\sum_{\ell =1}^{k}z_{k,\ell
}^{*})=p_{k+1}^{*}(\sum_{\ell =1}^{k+1}y_{k+1,\ell }^{*}+\mu _{k+1})\text{.}
\end{equation*}%
From the above equation, it is easily checked that if $x_{k}^{*}=1$ and $%
\sum_{\ell =1}^{k}z_{k,\ell }^{*}=0$, given that $p_{k}^{*}\lambda _{k}\leq
p_{k+1}^{*}\mu _{k+1}$ since $p^{*}$ satisfies \cref{B'}, we must have that $%
y_{k+1,\ell }^{*}=0$ for each $\ell $. Thus, we must have that $y_{1,\ell
}^{\ast }=y_{2,\ell }^{\ast }=0$ for each $\ell $.} In the sequel, we fix
such an outcome $(x^{\ast },y^{\ast },z^{\ast },p^{\ast })$. Note that $%
x_{2}^{\ast }>0$ implies that $z_{2,1}^{\ast }+z_{2,2}^{\ast }=1$ and since
the values of $z_{2,1}^{\ast }$ and $z_{2,2}^{\ast }$ are irrelevant when $%
x_{2}^{\ast }=0$, without loss, we will assume that $z_{2,1}^{\ast
}+z_{2,2}^{\ast }=1$. While the variables we study below do depend on $\mu $
and $\lambda $, for simplicity, we omit the dependence in notations.

We then study an agent's expected utility with elapse of time $t\geq 0$ on
the queue:
\begin{equation}
U(t)\triangleq S(t)V-W(t)C.  \label{Eq_U(t)}
\end{equation}%
$W(t)$ stands for the residual waiting time, conditional on having spent
time $t\geq 0$ on the queue, i.e.,
\begin{equation*}
W(t)\triangleq \gamma _{1,1}^{t}\tau _{1,1}+\gamma _{2,1}^{t}\tau
_{2,1}+\gamma _{2,2}^{t}\tau _{2,2}
\end{equation*}%
where $\gamma ^{t}=(\gamma_{1,1}^{t},\gamma _{2,1}^{t},\gamma _{2,2}^{t})$
is the belief an agent has about alternative states $(k,\ell)$ and $\tau
=(\tau_{1,1},\tau _{2,1},\tau_{2,2})$ are his expected waiting times at
alternative states, both under the queueing rule $q$. Similarly, $S(t)$ is
the probability of eventually getting served and writes as:
\begin{equation*}
S(t)\triangleq \gamma _{1,1}^{t}\sigma _{1,1}+\gamma _{2,1}^{t}\sigma
_{2,1}+\gamma _{2,2}^{t}\sigma _{2,2}
\end{equation*}%
where $\sigma=(\sigma_{1,1}, \sigma _{2,1},\sigma_{2,2})$ are the
probabilities of an agent getting eventually served at alternative states $%
(k,\ell )$, again under the queueing rule $q$. (Throughout, we suppress the
dependence on $q$ for notational ease.)

Since $U(0)=0$ (as implied by a binding \cref{IR}), it suffices to show that
$U(t)$ decreases strictly in the neighborhood of $t=0$ which will then prove
that $q$ fails $(IC_{t})$ for some small $t>0$. We establish this for a
sufficiently small value $\lambda >0$.\footnote{%
Recall we adjust the values of $C,V$ and $\alpha $ so as to ensure that %
\cref{IR} is binding at the optimal cutoff policy that solves $[P^{\prime }]$%
.} Specifically, we focus on $\dot{U}(0)$---the change in utility
\textquotedblleft right after joining the queue\textquotedblright ---as $%
\lambda \rightarrow 0$. As it turns out, $\dot{U}(0)\rightarrow 0$ as $%
\lambda \rightarrow 0$.
Hence, one must consider how \textquotedblleft slowly\textquotedblright\ $%
\dot{U}(0)$ converges to 0, or more precisely, the limit behavior of $\dot{U}%
(0)/\lambda $ as $\lambda \rightarrow 0$.

Hence, we will show that $\dot{U}(0)/\lambda \ $converges to a strictly
negative number as $\lambda \rightarrow 0$. For our purpose, it is enough to
show that, as $\lambda $ vanishes, $S^{\prime }(0)/\lambda \ $converges to $%
0 $ while $W^{\prime }(0)/\lambda \ $converges to a strictly positive
number. To this end, it is necessary to characterize the limit behaviors of $%
(\tau_{k,\ell}), (\sigma_{k,\ell})$ and $(\dot{\gamma}_{k,\ell}^{0})$. We do
this first. \smallskip

\noindent \textbf{Limit behavior of $(\tau _{k,\ell })$.}
The expected waiting time $\tau _{1,1}$ must satisfy:%
\begin{equation*}
\tau _{1,1}=\left( \mu dt\right) dt+\lambda dt\left( dt+\tau _{2,1}\right)
+\left( 1-\mu dt-\lambda dt\right) (dt+\tau _{1,1})+o(dt),
\end{equation*}%
since, for a small time increment $dt$, the sole agent in the queue waits
for time $dt$ if he is served during $[t,t+dt)$ (which occurs with
probability $\mu dt$), for $dt+\tau _{2,1}$ if another agent arrives during $%
[t,t+dt)$ (which occurs with probability $\lambda dt$), and for $dt+\tau
_{1,1}$ if neither event arises (which occurs with probability $1-\mu
dt-\lambda dt$). By a similar reasoning, we have:%
\begin{equation*}
\tau _{2,1}=\left( q_{2,1}dt+\lambda x_{2}^{\ast }z_{2,1}^{\ast }dt\right)
dt+q_{2,2}dt(dt+\tau _{1,1})+\left( 1-\mu dt-\lambda x_{2}^{\ast
}z_{2,1}^{\ast }dt\right) (dt+\tau _{2,1})+o(dt)
\end{equation*}%
and%
\begin{equation*}
\tau _{2,2}=\left( q_{2,2}dt+\lambda x_{2}^{\ast }z_{2,2}^{\ast }dt\right)
dt+q_{2,1}dt(dt+\tau _{1,1})+\lambda x_{2}^{\ast }z_{2,1}^{\ast }dt\left(
dt+\tau _{2,1}\right) +\left( 1-\mu dt-\lambda x_{2}^{\ast }dt\right)
(dt+\tau _{2,2})+o(dt).
\end{equation*}%
Letting $dt\rightarrow 0$ and simplifying, we obtain:%
\begin{equation*}
\left( \mu +\lambda \right) \tau _{1,1}=\lambda \tau _{2,1}+1,\text{ }\left(
\mu +\lambda x_{2}^{\ast }z_{2,1}^{\ast }\right) \tau _{2,1}=q_{2,2}\tau
_{1,1}+1\text{ and }\left( \mu +\lambda x_{2}^{\ast }\right) \tau
_{2,2}=\lambda x_{2}^{\ast }z_{2,1}^{\ast }\tau _{2,1}+q_{2,1}\tau _{1,1}+1%
\text{.}
\end{equation*}%
Thus, we have that, as $\lambda \rightarrow 0$,%
\begin{equation}
\tau _{1,1}\rightarrow \frac{1}{\mu },\tau _{2,1}\rightarrow \frac{q_{2,2}}{%
\mu }\frac{1}{\mu }+\frac{1}{\mu }\text{ and }\tau _{2,2}\rightarrow \frac{%
q_{2,1}}{\mu }\frac{1}{\mu }+\frac{1}{\mu }  \label{Eq:lim_tau}
\end{equation}%
where we abuse notations and simply note $q_{2,2}$ for the limit as $\lambda
$ vanishes of $q_{2,2}\ $(and similarly for $q_{2,1}$). We assume
here that this limit is well-defined and take a subsequence of our vanishing
sequence of $\lambda $ if necessary. \smallskip

\noindent\textbf{Limit behavior of $(\sigma_{k,\ell})$.}
We have%
\begin{equation*}
\sigma _{1,1}=\mu dt+\lambda dt\sigma _{2,1}+(1-\mu dt-\lambda dt)\sigma
_{1,1}+o(dt)
\end{equation*}%
since, for a small time increment $dt$, the sole agent in the queue is
served with probability $\mu dt$; the agent is eventually served with
probability $\sigma _{2,1}$ if another agent arrives (which occurs with
probability $\lambda dt$), and the agent is served with probability $\sigma
_{1,1}$ if neither event arises (which occurs with probability $1-\mu
dt-\lambda dt$). Similar reasoning yields the following expressions for $%
\sigma _{2,1}$ and $\sigma _{2,2}$%
\begin{equation*}
\sigma _{2,1}=q_{2,1}dt+(1-\mu dt-\lambda x_{2}^{\ast }dt)\sigma
_{2,1}+q_{2,2}dt\sigma _{1,1}+\lambda x_{2}^{\ast }dtz_{2,2}^{\ast }\sigma
_{2,1}+o(dt),
\end{equation*}%
and%
\begin{equation*}
\sigma _{2,2}=q_{2,2}dt+(1-\mu dt-\lambda x_{2}^{\ast }dt)\sigma
_{2,2}+q_{2,1}dt\sigma _{1,1}+\lambda x_{2}^{\ast }dtz_{2,1}^{\ast }\sigma
_{2,1}+o(dt).
\end{equation*}%
We obtain
\begin{align*}
(\mu +\lambda )\sigma _{1,1}& =\mu +\lambda \sigma _{2,1} \\
(\mu +\lambda x_{2}^{\ast }(1-z_{2,2}^{\ast }))\sigma _{2,1}&
=q_{2,1}+q_{2,2}\sigma _{1,1} \\
(\mu +\lambda x_{2}^{\ast })\sigma _{2,2}& =q_{2,2}+q_{2,1}\sigma
_{1,1}+\lambda x_{2}^{\ast }z_{2,1}^{\ast }\sigma _{2,1}.
\end{align*}%
Hence, we obtain that
\begin{equation}  \label{Eq:lim_sigma}
\sigma _{1,1},\sigma _{2,1},\sigma _{2,2}\rightarrow 1 \text{ as } \lambda
\rightarrow 0.
\end{equation}
\smallskip

\noindent\textbf{Limit behavior of $(\dot{\gamma}_{k,\ell}^{0})$.}
We study the dynamics of beliefs. An agents' beliefs evolve during $[t,t+dt)$
according to Bayes rule. For instance, for state $(k,\ell )=(1,1)$, we obtain%
\begin{equation*}
\gamma _{1,1}^{t+dt}=\frac{\gamma _{1,1}^{t}\left[ 1-\mu dt-\lambda dt\right]
+\gamma _{2,2}^{t}\left[ q_{2,1}dt\right] +\gamma _{2,1}^{t}\left[ q_{2,2}dt%
\right] }{\gamma _{1,1}^{t}\left[ 1-\mu dt\right] +\gamma _{2,1}^{t}\left[
1-q_{2,1}dt-\lambda x_{2}^{\ast }z_{2,1}^{\ast }dt\right] +\gamma _{2,2}^{t}%
\left[ 1-q_{2,2}dt-\lambda x_{2}^{\ast }z_{2,2}^{\ast }dt\right] }+o(dt)
\end{equation*}%
where the numerator is the probability that the agent's state is $(k,\ell
)=(1,1)$ after staying in the queue for length $t+dt$ of time. This event
occurs if either (i) the agent is already in state $(1,1)$ in the queue at
time $t$, the agent is not served and no agent arrives in the queue during
time increment $dt$; or (ii) his state is $(2,2)$ or $(2,1)$ at $t$ and the
other agent in the queue is served by $t+dt$. The denominator in turn gives
the probability that the agent has not been served or removed from the queue
by time $t+dt$. Hence, given that an agent has not been served or removed
from the queue by $t$, the above expression gives the conditional belief
that his state is $(1,1)$ at time $t+dt$.

Similar reasoning yields the following expressions for the evolution of
beliefs for state $(2,1)$ and $(2,2)$%
\begin{equation*}
\gamma _{2,1}^{t+dt}=\frac{\gamma _{2,1}^{t}\left[ \lambda x_{2}^{\ast
}z_{2,2}^{\ast }dt+1-\mu dt-\lambda x_{2}^{\ast }dt\right] +\gamma _{2,2}^{t}%
\left[ \lambda x_{2}^{\ast }z_{2,1}^{\ast }dt\right] +\gamma _{1,1}^{t}\left[
\lambda dt\right] }{\gamma _{1,1}^{t}\left[ 1-\mu dt\right] +\gamma
_{2,1}^{t}\left[ 1-q_{2,1}dt-\lambda x_{2}^{\ast }z_{2,1}^{\ast }dt\right]
+\gamma _{2,2}^{t}\left[ 1-q_{2,2}dt-\lambda x_{2}^{\ast }z_{2,2}^{\ast }dt%
\right] }+o(dt)
\end{equation*}%
and%
\begin{equation*}
\gamma _{2,2}^{t+dt}=\frac{\gamma _{2,2}^{t}\left[ 1-\mu dt-\lambda
x_{2}^{\ast }dt\right] }{\gamma _{1,1}^{t}\left[ 1-\mu dt\right] +\gamma
_{2,1}^{t}\left[ 1-q_{2,1}dt-\lambda x_{2}^{\ast }z_{2,1}^{\ast }dt\right]
+\gamma _{2,2}^{t}\left[ 1-q_{2,2}dt-\lambda x_{2}^{\ast }z_{2,2}^{\ast }dt%
\right] }+o(dt).
\end{equation*}

From these, we can derive ODEs that describe belief evolutions:
\begin{eqnarray*}
\dot{\gamma}_{1,1}^{t} &=&-\gamma _{1,1}^{t}\left[ \mu +\lambda \right]
+\gamma _{2,2}^{t}\left[ q_{2,1}\right] +\gamma _{2,1}^{t}\left[ q_{2,2}%
\right] +\left( \gamma _{1,1}^{t}\right) ^{2}\left[ \mu \right]  \\
&&+\gamma _{1,1}^{t}\gamma _{2,1}^{t}\left[ q_{2,1}+\lambda x_{2}^{\ast
}z_{2,1}^{\ast }\right] +\gamma _{1,1}^{t}\gamma _{2,2}^{t}\left[
q_{2,2}+\lambda x_{2}^{\ast }z_{2,2}^{\ast }\right] ,
\end{eqnarray*}%
\begin{eqnarray*}
\dot{\gamma}_{2,1}^{t} &=&-\gamma _{2,1}^{t}\left[ \mu +\lambda x_{2}^{\ast
}(1-z_{2,2}^{\ast })\right] +\gamma _{2,2}^{t}\left[ \lambda x_{2}^{\ast
}z_{2,1}^{\ast }\right] +\gamma _{1,1}^{t}\left[ \lambda \right]  \\
&&+\gamma _{2,1}^{t}\gamma _{1,1}^{t}\left[ \mu \right] +\left( \gamma
_{2,1}^{t}\right) ^{2}\left[ q_{2,1}+\lambda x_{2}^{\ast }z_{2,1}^{\ast }%
\right] +\gamma _{2,1}^{t}\gamma _{2,2}^{t}\left[ q_{2,2}+\lambda
x_{2}^{\ast }z_{2,2}^{\ast }\right]
\end{eqnarray*}%
and%
\begin{eqnarray*}
\dot{\gamma}_{2,2}^{t} &=&-\gamma _{2,2}^{t}\left[ \mu +\lambda x_{2}^{\ast }%
\right] +\gamma _{2,2}^{t}\gamma _{1,1}^{t}\left[ \mu \right]  \\
&&+\gamma _{2,2}^{t}\gamma _{2,1}^{t}\left[ q_{2,1}+\lambda x_{2}^{\ast
}z_{2,1}^{\ast }\right] +\left( \gamma _{2,2}^{t}\right) ^{2}\left[
q_{2,2}+\lambda x_{2}^{\ast }z_{2,2}^{\ast }\right]
\end{eqnarray*}%
with a boundary condition at $t=0$ satisfying $\gamma _{2,1}^{0}=0\ $and
\begin{equation*}
\gamma _{1,1}^{0}=\frac{\lambda p_{0}}{\lambda p_{0}+\lambda p_{1}+\lambda
x_{2}^{\ast }p_{2}\left( z_{2,1}^{\ast }+z_{2,2}^{\ast }\right) }=\frac{1}{1+%
\frac{\lambda }{\mu }+x_{2}^{\ast }\left( \frac{\lambda }{\mu }\right) ^{2}},
\end{equation*}%
and
\begin{equation*}
\gamma _{2,2}^{0}=\frac{\lambda p_{1}+\lambda x_{2}^{\ast }p_{2}\left(
z_{2,1}^{\ast }+z_{2,2}^{\ast }\right) }{\lambda p_{0}+\lambda p_{1}+\lambda
x_{2}^{\ast }p_{2}\left( z_{2,1}^{\ast }+z_{2,2}^{\ast }\right) }=\frac{%
\frac{\lambda }{\mu }+x_{2}^{\ast }\left( \frac{\lambda }{\mu }\right) ^{2}}{%
1+\frac{\lambda }{\mu }+x_{2}^{\ast }\left( \frac{\lambda }{\mu }\right) ^{2}%
},
\end{equation*}%
where we used the fact that $p_{1}\mu =\lambda p_{0}\text{ and }p_{2}\mu
=\lambda p_{1}=\lambda \frac{\lambda }{\mu }p_{0}$ at the invariant
distribution together with $z_{2,1}^{\ast }+z_{2,2}^{\ast }=1$ since state $%
k\geq 3$ have mass $0$ at the invariant distribution. (Recall that we
assumed, wlog, that $z_{2,1}^{\ast }+z_{2,2}^{\ast }=1$).

Observe that%
\begin{equation*}
\frac{\gamma _{1,1}^{0}}{\lambda }-\frac{1}{\lambda }\rightarrow -\frac{1}{%
\mu }\text{, }\frac{\gamma _{2,2}^{0}}{\lambda }\rightarrow \frac{1}{\mu }%
\text{ and }\frac{\gamma _{2,1}^{0}}{\lambda }=0
\end{equation*}%
In addition,
\begin{equation}  \label{Eq:lim_gamma}
\frac{\dot{\gamma}_{1,1}^{0}}{\lambda }\rightarrow -1<0\text{, }\frac{\dot{%
\gamma}_{2,1}^{0}}{\lambda }\rightarrow 1>0\text{ and }\frac{\dot{\gamma}%
_{2,2}^{0}}{\lambda }\rightarrow 0\text{.}
\end{equation}

\begin{proof}[Completion of the proof of \Cref{thm:necessity}]
As we already mentioned, for our purpose, it is enough to show that as $%
\lambda $ vanishes, $S^{\prime }(0)/\lambda \ $converges to $0$ while $%
W^{\prime }(0)/\lambda \ $converges to a strictly positive number. We have
that
\begin{equation*}
\frac{W^{\prime }(0)}{\lambda }=\frac{\dot{\gamma}_{1,1}^{t}}{\lambda }\tau
_{1,1}+\frac{\dot{\gamma}_{2,1}^{0}}{\lambda }\tau _{2,1}+\frac{\dot{\gamma}%
_{2,2}^{0}}{\lambda }\tau _{2,2}\rightarrow -\frac{1}{\mu }+\left( \frac{%
q_{2,2}}{\mu }\frac{1}{\mu }+\frac{1}{\mu }\right) =\left( \frac{q_{2,2}}{%
\mu }\right) \frac{1}{\mu }>0
\end{equation*}%
where the limit result comes from \Cref{Eq:lim_tau} and \Cref{Eq:lim_gamma}
while the strict inequality holds given our assumption that $q$ differs from
FCFS and so $q_{2,2}>0$. Further, we have
\begin{equation*}
\frac{S^{\prime }(0)}{\lambda }=\frac{\dot{\gamma}_{1,1}^{t}}{\lambda }%
\sigma _{1,1}+\frac{\dot{\gamma}_{2,1}^{0}}{\lambda }\sigma _{2,1}+\frac{%
\dot{\gamma}_{2,2}^{0}}{\lambda }\sigma _{2,2}\rightarrow 0
\end{equation*}%
where the limit result comes from \Cref{Eq:lim_sigma} and \Cref{Eq:lim_gamma}%
. Thus, as claimed, $\dot{U}(0)/\lambda \ $converges to a strictly negative
number as $\lambda \rightarrow 0$. \end{proof}


\newpage
\section*{Online Appendix}

 \section{General Queueing and Information Rules} \label{app-sec: general}

Our results do not rest on the set of queueing and information rules assumed in the main text but rather hold in much more general environments.  Here we specify this more general class of queueing and information rules and argue how our results continue to hold under them.  The key argument is established in  \Cref{lem:general} below. Incidentally, this lemma   also implies  that $[P']$ is a relaxed program of $[P]$, which was argued informally in the text.

 \paragraph{$\bullet$ Queueing rule:}

 Here we consider a fully general queueing rule that allocates  service priority (including possible sharing of service) for each agent in the queue as a (random) function of every possible event observable by the designer.   The only payoff-relevant aspect of such a queueing rule is the {\it eventual service probability} and the {\it expected residual waiting time} it induces for each agent after each possible time $t\ge 0$ she has spent in the queue. (Note that given the linear waiting costs, residual waiting time matters only through its expectation.)  Hence, we formally proceed as follows.  Fix $(\l,\mu, x,y,z)$. This induces a Markov chain on the length of the queue with an arbitrary initial state, say with length $k=0$. We then  specify a queueing rule by a calendar-time indexed stochastic process $\theta=(\theta^t_s)_{s,t\ge 0}$, where  $\theta^t_s=(\sigma^t_s,\tau^t_s)\in \Delta([0,1]\times \mathbb{R}_+)$ consists of  the {\it distribution} of an agent's {\it eventual service probability} $\sigma_s^t$  and the distribution of the {\it expected residual waiting time} $\tau_S^t$, both conditional on the agent having entered the queue at calendar time $s\ge 0$ and  spent time $t\ge 0$ in the queue.

The process $\theta$ must be such that, for each $s\ge 0$, $(\theta_s^t)_t$ must form a filtration, and $\theta^{[\cdot]}_s$ must be a filtration with respect to $s$. Further, $\theta$ must be compatible with $(\l,\m, x,y,z)$ and the Markov chain on the queue length it induces. For our purpose, we do not need to specify or characterize the compatibility precisely.  Since we focus on the invariant distribution $p=(p_k)$ on the Markov chain, we will require the process $\theta$ to satisfy only the following two conditions:
 \begin{align*}
 & \sum_{k} p_k \l_k x_k   \bar \sigma^0_s = \sum_{k} \mu_k p_k, \tag{WC} \label{cap}\\
 & \sum_{k} p_k \l_k x_k  \bar \tau^0_s =\sum_{k}k p_k,  \tag{Little} \label{little}
 \end{align*}
where  $$  \bar \sigma^0_s:=\int \sigma  \theta^0_{s}(d\sigma d w)$$
is an agent's ex-ante expectation of the eventual service probability and
$$\bar \tau^0_s:=\int \tau  \theta_s^0(d\sigma d \tau)$$
is her ex-ante expectation of the waiting time, both at the time of her entry when she has entered the queue at the calendar time $s\ge 0$.

Condition \Cref{cap} is work-conservation/non-wastefulness requiring that the expected service rate enjoyed by an arriving agent (LHS) must equal the maximal total service rate available (RHS).  Condition \Cref{little} follows from Little's law  that the expected rate of arrival multiplied by the expected waiting equals the expected size of the queue. Hence, one should think of this as a (minimal) feasibility condition.
 These conditions in turn imply that  both $  \bar \sigma^0_s$ and $  \bar \tau^0_s$ do not depend on $s$.

 Let $\Theta$ denote the set of all $\theta$'s that satisfy  \Cref{cap} and \Cref{little}.  The queueing rules accommodated by $\Theta$ are fully general as long as they satisfy work conservation.  For example, any (Markovian) queueing rule in $\Q$ assumed in the main text trivially satisfies them.  First, \cref{little} holds independently of the queueing rule. Second, \Cref{cap} holds  since the service rate equals the maximal service rate at each state.  For instance, any hybrid rule that mixes multiple standard disciplines, a rule that permutes priorities following some event observable by the designer as envisioned by \cite{leshno2019dynamic}, and even a non-stationary and non-anonymous rule that changes service priorities based on the realized history are included in $\Theta$ as long as they are work-conserving. The reason is that any such rule, as long as it admits an invariant distribution, induces a well-defined process $\theta$ obeying \Cref{cap}  and  \Cref{little}.  Hence, they are accommodated by  $\Theta$.  Interestingly, the {\it feasibility} condition we impose in $\Q$ is not required for the process $\theta$ to satisfy \Cref{cap}  and  \Cref{little}. Since  FCFS belongs to $\Q$ and hence satisfies feasibility,  our \Cref{thm:dyn-fcfs} means that feasibility associated with FCFS imposes no cost for FCFS to be optimal among all queueing rules allowed in $\Theta$.

 \paragraph{$\bullet$ Information rule:}

We  can accordingly generalize an information rule as a feasible signal about a queueing rule $\theta\in \Theta$.  An information rule is simply specified as a (calendar-time indexed) distribution of signals on $(\sigma^t_s, \tau^t_s)$ received/observed by each agent.  Specifically, an \textsf{information rule} is given by  $\gamma=(\gamma_{s}^t)_{t,s}$, where the distribution $\gamma_s^t\in \Delta([0,1]\times \mathbb{R}_+)$ of the signals received by an agent who entered the queue at calendar time $s$ and spent time $t$ is a mean-preserving contraction of $\theta_{s}^t$.  Let $\Gamma$ denote the set of all such $\gamma$'s. That is, we impose no restriction on the signals received by the agent.  For our purpose, it is not necessary to characterize the conditions for $\gamma$ to be a posterior belief distribution of $\theta$; it will suffice to use only that $\gamma_s^0$ shares the same ex-ante mean as $\theta_s^0$.

 \begin{remark} In the case of the Markovian queueing rule assumed in the main text, the memorylessness of the rule $(\lambda, \mu, x, y, z, q)$ means that a pair $(k,\ell)$ of the queue length and position is a sufficient statistic for the eventual service probability as well as the expected residual waiting time.  Hence, the information can be succinctly characterized by a posterior belief distribution $(\gamma_{k, \ell})$ on $(k,\ell)$, which justifies our specification in the main text.
 \end{remark}

 Obviously, $\gamma$ includes \textbf{full information}, in which case $ \gamma$ coincides with $\theta$, and \textbf{no information}, in which case the belief distribution $\gamma_{s}^t$ is degenerate on the mean of $\theta_{s}^t$.  We allow any possible information in between the two, including the one that varies with time, as allowed for by the indexation by the calendar time.

 \paragraph{$\bullet$ Incentives.} The incentive constraint can be stated for the general class as follows:
\begin{equation*}
 V \cdot \sigma -C \cdot \tau  \geq 0,\forall (\sigma,\tau)\in \supp(\gamma_s^t),\forall
t\geq 0, s\ge 0. \tag{$\widetilde{IC}$}  \label{IC'}
\end{equation*}
The new incentive constraint \cref{IC'} states that an agent, when recommended to join or
stay in the queue, must find the prospect of being
served to be high enough to justify the expected remaining waiting cost, given each possible belief $(\gamma_{s}^t)$.  Recall that we do not require the constraint to be satisfied for an agent who is being removed from the queue. Any such agent will face $\tau_s^t=0$, however, so the constraint will hold trivially. Note also that the dependence of the distribution $\tilde \gamma^t_{s}$ on the calendar time $s$ may not vanish if either the queueing rule or the informational rule is nonstationary.
Clearly, the constraint \Cref{IC'} simplifies to \Cref{IC} in the case of the Markovian queueing and information rules in the text, upon noting that ``state'' variable $(k,\ell)$ is the sufficient statistic for the stochastic process of  eventual service probability and expected residual waiting time.

\paragraph{$\bullet$ Generalization Lemma.}  Consider now the new program $[\widetilde P]$ which is the same as $[P]$ except that the queueing and information rules are now chosen from the larger sets $\Theta$ and $\Gamma$, and \Cref{IC} is now replaced by \Cref{IC'}.  Note that \Cref{B} remains a valid balance equation for our queueing rule since it satisfies \Cref{cap}.  Clearly, the value of $[\widetilde P]$ is no less than that of $[P]$.

The crucial observation we now make is that $[P']$ is still a relaxed program of $[\widetilde P]$; this  will imply that all subsequent results, \Cref{thm:cutoff} and \Cref{thm:dyn-fcfs}, remain valid under the general class of queueing and information rules we now allow for.  The following lemma will establish this.

\begin{lemma} \label{lem:general} The constraint \Cref{IC'} implies \Cref{IR}, and thus the value of $[P']$ is no less than the value of $[\widetilde P]$.
\end{lemma}
\begin{proof}  Fix any feasible solution to $[\widetilde P]$. If $p_0=1$ at the feasible solution to $[\widetilde P]$, then \Cref{IR} holds trivially.  Hence, assume $p_0<1$.

Fix any calendar time $s$.
By integrating  both sides of \Cref{IC'} for $t=0$ over all $(\sigma,\tau)\in \supp(\gamma^0_s)$ according to the distribution $\gamma^0_s$, we obtain that
\begin{align*}
0\le &   V \int_{\sigma,\tau}\sigma \gamma^0_s(d s dw)-C \int_{\sigma,\tau} \tau \gamma^{0}_s(d  \sigma,d\tau) \\
= & V \int_{\sigma,\tau}\sigma \theta^0_s(d s dw)-C \int_{\sigma,\tau} \tau \theta^{0}_s(d  \sigma,d\tau) \\
      =& V \bar \sigma^0_s  - C \bar \tau^0_s \\
         =& V \frac{\sum_{k} \mu_k p_k}{\sum_{k} p_k \l_k x_k}   - C \frac{\sum_{k} k p_k}{\sum_{k} p_k \l_k x_k}    \\
         =& \frac{1}{\sum_{k} p_k \l_k x_k}\sum_{k} p_k\left(    \mu_k  V -  k C\right).
\end{align*}
The   inequality follows from \Cref{IC'}  for $t=0$; the first equality follows from the fact that $\gamma^0_s$ is a mean-preserving contraction of $\tau^0_s$, so they have the same means for $\sigma$ and $\tau$; the second equality uses the definitions of $\bar \sigma^0_s$  and $\bar \tau^0_s$; the fourth  equality follows from \Cref{cap} and \Cref{little}; and the last equality is a simple rearrangement of terms.  Finally, we conclude that $\sum_{k} p_k \l_k x_k>0$.  Suppose  not.  Then,   $x_k=0$ for all $k$ with $p_k\lambda_k>0$. But this means that $p_0=1$, contrary to our assumption.  Since $\sum_{k} p_k \l_k x_k>0$, $$\sum_{k} \left(    \mu_k  V -  k C\right)\ge 0,$$ so  \cref{IR} is satisfied.  Note also that \cref{IR} holds regardless of $s$, which means that if the incentive constraint is required only for the  limit as  $s\to \infty$, it will still imply \cref{IR}.

We  already established in the text that \Cref{B} implies \Cref{B'}. Hence, we have shown that $[P']$ is a relaxation of $[\widetilde P]$, from which the second statement follows.
\end{proof}

Since our $\Q$ and $\I$ are nested by $\Theta$ and $\Gamma$, respectively, the following corollary holds.

\begin{corollary} The constraint \Cref{IC} implies \Cref{IR}, and thus the value of $[P']$ is no less than the value of $[P]$.
\end{corollary}

\section{Axiomatic Foundation for Regular Service Rates}\label{online_sec:axiomatic}

As defined in the text, for each $k\in \mathbb{N}$, $\mu_k$ represents the
maximal service rate that any set of $k$ agents may receive. By definition, $%
\mu_k$ is nondecreasing in $k$ since if $\mu_k>\mu_m$ for $k<m$, we can
simply redefine $\mu_m \triangleq \mu_k$. In this sense, we view $\mu=(\mu_k)_{k\in
\mathbb{N}}$ as ``effective'' maximal service rates.\footnote{%
Indeed, we can characterize $\mu$ as arising from more primitive service
constraints. Say there are upper bound constraints $(c_k)_k$ for each group
of $k$ agents. We do not impose any condition on $(c_k)_k$, except that
there exists $B>0$ such that $c_k\le B$ for all $k$ and it is
nondecreasing. The effective service rate $\mu_n$ for $n$ agents can be
defined as the value:
\begin{equation*}
\sup_{q\in \mathbb{R}_+^n} \sum_{i\in [n]} q_i \leqno[C_n]
\end{equation*}
subject to
\begin{equation*}
\sum_{i\in S}q_i\le c_k, \forall k\in \mathbb{N}, \forall S\subset [n] \,
s.t.\,|S|= k.
\end{equation*}%
}

Below we provide a more primitive definition of FCFS based on the concept
that the priority must be assigned greedily to maximize the service rates
for earlier arriving agents. Under regularity of $\mu$, this definition will then produce the formula we
presented in the main paper as the definition of FCFS.

\paragraph{FCFS:}

Specifically, for each $k\in \mathbb{N}$, we define the service rates $%
(q_{k,1}^{\ast },...,q_{k,k}^{\ast })\in \mathbb{R}_{+}^{k}$ that agents in
the queue of $k$ length receive under FCFS.

To begin, let $Q_{0}\triangleq \mathbb{R}_{+}^{k}$, and consider a sequence
of the following problems:

In step $j\in \lbrack k]\triangleq \{1,..., k\}$, we choose
\begin{equation*}
Q_{j}=\arg \max_{q\in Q_{j-1}}\sum_{i\in \lbrack j]}q_{i}\leqno[C_j^*]
\end{equation*}%
subject to
\begin{equation*}
\sum_{i\in S}q_{i}\leq \mu _{m},\forall m \in [k],\forall S\subset
\lbrack k]\,s.t.\,|S|=m.
\end{equation*}

In words, the first agent's service rate is maximized subject to the
constraint that he can never receive more than $\mu_1$ the maximal service
rate any single agent can ever receive. Taking that as constraint, we next
maximize first and second agents' service rate now only
subject to $\mu_2$ the maximal total service rate that any two agents can
ever receive, and so on.

The FCFS service rates $(q_{k,1}^{\ast },...,q_{k,k}^{\ast })$ are then
defined to be an optimal solution for step $k$---i.e., an element of $Q_{k}$%
. While it is in principle possible that $Q_k$ has multiple elements, it is easily seen that
$Q_{k}$ is a singleton. We let  $\mu _{k}^{\ast }$ denote the
maximized value of $[C_{k}^{\ast }]$. Next, observe that, for any $i\leq
k,k^{\prime }$, we have $q_{k,i}^{\ast }=q_{k^{\prime },i}^{\ast }$. Hence,
we henceforth write $q_{i}^{\ast }\ $for $q_{k,i}^{\ast }$.

We now derive the optimal solution $(q_j^*)_{j\in [k]}$ explicitly. The
resulting formula will resemble the one we defined for the service rates
under FCFS.

\begin{lemma}
\label{lem:fcfs-formula} Fix $k$. The optimal value of $[C_{j}^{\ast }]$ is $%
\mu _{j}^{\ast }$, where $\mu _{1}^{\ast }=\mu _{1}$, and for $j=2,...,k$,
\begin{equation*}
\mu _{j}^{\ast }=\mu _{j-1}^{\ast }+\min \{\mu _{j}-\mu _{j-1}^{\ast },\mu
_{j-1}^{\ast }-\mu _{j-2}^{\ast }\}.
\end{equation*}%
Agent $j\in \lbrack k]$ receives service rates $q_{j}^{\ast }=\mu _{j}^{\ast
}-\mu _{j-1}^{\ast }$, which is nonincreasing in $j$, for $j=1,...k$, where $%
\mu _{0}^{\ast }\triangleq 0$.
\end{lemma}

\begin{proof}
The  proof is inductive. First, it is trivial to note that $\mu
_{1}^{\ast }=\mu _{1}$ is indeed the value of $[C_{1}^{\ast }]$ and $%
q_{1}^{\ast }=\mu _{1}^{\ast }=\mu _{1}^{\ast }-\mu _{0}^{\ast }$. Suppose
next that $\mu _{i}^{\ast }$ is the value of $[C_{i}^{\ast }]$ for all $%
i=1,...,j-1$, and these steps pin down $q_{i}^{\ast }:=\mu _{i}^{\ast }-\mu
_{i-1}^{\ast }$. We make several observations: (i) Since $\mu _{i}^{\ast }$
is the value of $[C_{i}^{\ast }]$ for $i=j-2,j-1$, any $q\in Q_{j-1}$ has $%
q_{\ell }\leq \mu _{j-1}^{\ast }-\mu _{j-2}^{\ast }=q_{j-1}^{\ast }$ for all
$\ell \geq j-1$. (Suppose to the contrary that $q_{\ell }>\mu _{j-1}^{\ast
}-\mu _{j-2}^{\ast }$ for some $q\in Q_{j-1}$, then swapping $q_{j-1}$ and $%
q_{\ell }$ between $j-1$ and $\ell $ is feasible and strictly improves the
value of $[C_{j-1}^{\ast }]$, a contradiction.) (ii) By construction, we
have $\mu _{i}^{\ast }\leq \mu _{i}$ for all $i=1,...,j-1$. (iii) By
construction, we have $q_{i}^{\ast }\leq q_{i^{\prime }}^{\ast }$ for $%
i^{\prime }\leq i\leq j-1$ (which follows from the fact that $\mu _{j}^{\ast
}-\mu _{j-1}^{\ast }$ is nonincreasing in $j$.).

Consider problem $[C_{j}^{\ast }]$. We will argue that its value is given by
the formula $\mu _{j}^{\ast }=\mu _{j-1}^{\ast }+\min \{\mu _{j}-\mu
_{j-1}^{\ast },\mu _{j-1}^{\ast }-\mu _{j-2}^{\ast }\}$, and it pins down $%
q_{j}^{\ast }=\mu _{j}^{\ast }-\mu _{j-1}^{\ast }$. To this end, note first
that the value $\mu _{j}^{\ast }$ of $[C_{j}^{\ast }]$ cannot exceed:
\begin{equation*}
\mu _{j-1}^{\ast }+\min \{\mu _{j}-\mu _{j-1}^{\ast },\mu _{j-1}^{\ast }-\mu
_{j-2}^{\ast }\}\text{.}
\end{equation*}%
To see this, simply observe that the above term can take two values, either $%
\mu _{j}$ or $\mu _{j-1}^{\ast }+\mu _{j-1}^{\ast }-\mu _{j-2}^{\ast }$.
Since, by definition of $[C_{j}^{\ast }]$, $\mu _{j}^{\ast }\leq \mu _{j}$
the result holds in the former case. Since by (i) above $q_{\ell }\leq \mu
_{j-1}^{\ast }-\mu _{j-2}^{\ast }$ for all $\ell \geq j-1$ for any $q\in
Q_{j-1}$, and since, by definition, the value of $[C_{j}^{\ast }]$ equals $%
\mu _{j-1}^{\ast }+q_{j}^{\ast }$, the result also holds in the latter case.

We next prove that the value is actually attained. Construct $\hat{q}$ such
that $\hat{q}_{i}=q_{i}^{\ast }$ for all $i\leq j-1$, $\hat{q}_{j}=\mu
_{j}^{\ast }-\mu _{j-1}^{\ast }$ and $\hat{q}_{i}=0$ for all $i\geq j+1$.
Note that, since $\mu _{i}^{\ast }-\mu _{i-1}^{\ast }$ is nonincreasing in $i
$, $\hat{q}_{j}\leq \hat{q}_{i}$ for all $i\leq j$. Take any $S\subset
\lbrack k]$ such that $|S|=\ell <j$. Then,
\begin{equation*}
\sum_{i\in S}\hat{q}_{i}\leq \sum_{i\in \lbrack \ell ]}\hat{q}_{i}=\mu
_{\ell }^{\ast }\leq \mu _{\ell },
\end{equation*}%
where the first inequality follows from (iii) and the second follows from
(ii). Next, take any $S\subset \lbrack k]$ such that $|S|=k$. Then,
\begin{equation*}
\sum_{i\in S}\hat{q}_{i}\leq \sum_{i\in \lbrack j]}\hat{q}_{i}=\mu
_{j}^{\ast }\leq \mu _{j},
\end{equation*}%
where the first follows from (iii) and the fact that $q_{j}^{\ast }\leq
q_{i}^{\ast }$ for all $i\leq j$, and the second follows from our prior
observation that the value $\mu _{j}^{\ast }$ of $[C_{j}^{\ast }]$ must be
smaller than $\mu _{j}$. Lastly, it is trivial that $\sum_{i\in S}\hat{q}%
_{i}\leq \mu _{\ell },$ for any $S$ with $|S|=\ell $, where $\ell >j$. We
thus conclude $\hat{q}\in Q_{j}$ and $\mu _{j}^{\ast }$ is the value of $%
[C_{j}^{\ast }]$.
\end{proof}

It is easy to verify that the optimal solution $(q_j^*)_{j\in [k]}$ is unique. More
importantly, one can see  that the solution coincides with the service rate
we define for FCFS in the main text,  provided that FCFS is work conserving.
To see this note from \Cref{lem:fcfs-formula} that  $\sum_{j\in \lbrack
k]}q_{j}^{\ast }=\mu _{k}^*$. Hence, if FCFS is work conserving, we must
have $\mu_k^*=\mu_k$ for each $k$ (since $\mu_k^* \leq \mu_k$ for each $k$). In that case, we get $q_j^*=
\mu_{j}-\mu_{j-1}$, precisely as we defined in the text.

\paragraph{Axiomatic Characterization:}

We now prove that regularity of $\mu $ is a necessary and sufficient
condition for FCFS to be work-conserving, i.e., $\sum_{i\in \lbrack
k]}q_{k,i}^{\ast }=\mu _{k}$ for all $k$.

\begin{theorem}
\label{thm:axiomatization} FCFS is work-conserving if and only if $\mu $ is
regular.
\end{theorem}

\begin{proof}  By \Cref{lem:fcfs-formula},   for all $k\in \mathbb{N}$, $\sum_{i\in \lbrack k]}q_{i}^{\ast }=\mu
_{k}^{\ast }$,  and by feasibility $\mu _{k}^{\ast }\leq \mu _{k}$.  Hence, FCFS is work-conserving
if and only if $\mu _{k}^{\ast }=\mu _{k}$ for all $k\in \mathbb{N}$. Thus, it suffices to prove that
$\mu $ is regular if and only if $\mu
_{k}^{\ast }=\mu _{k}$ for all $k$.

To prove the ``only if'' direction, suppose $\mu $ is regular. We argue inductively that $\mu _{k}^{\ast }=\mu
_{k}$ for all $k$.\ First, by definition, $\mu _{1}^{\ast }=\mu _{1}$.
Suppose $\mu _{i}^{\ast }=\mu _{i}$ for all $i\in \lbrack k-1]$. Then,
\begin{eqnarray*}
\mu _{k}^{\ast } &=&\mu _{k-1}^{\ast }+\min \{\mu _{k}-\mu _{k-1}^{\ast
},\mu _{k-1}^{\ast }-\mu _{k-2}^{\ast }\} \\
&=&\mu _{k-1}+\min \{\mu _{k}-\mu _{k-1},\mu _{k-1}-\mu _{k-2}\}=\mu _{k},
\end{eqnarray*}%
where the first equality is by definition of $\mu _{k}^{\ast }$, the second
follows from the induction hypothesis, and the last follows from the
regularity.

The converse, the ``if'' direction, follows from the fact that $\mu _{k}-\mu _{k-1}=\mu
_{k}^{\ast }-\mu _{k-1}^{\ast }=q_{k}^{\ast }$ and $q_{k}^{\ast }$ is
nonincreasing in  $k$ by \Cref{lem:fcfs-formula}.  \end{proof}

We have focused only on FCFS, but the LCFS can be defined analogously, and a
similar result is obtained.

\section{Proof of \Cref{lem: charact stoch dom}}
\label{sup: stoch dom}

We have
\begin{eqnarray*}
\sum_{k=0}^{K}p _{k}^{\prime }\varphi (k) &=&\varphi
(K)-\sum_{L=0}^{K-1}\left( \sum_{k=0}^{L}p _{k}^{\prime }\right) \left(
\varphi (L+1)-\varphi (L)\right) \\
&=&\varphi (K)-\sum_{L=0}^{K-1}\left( 1-\sum_{k=L+1}^{K}p _{k}^{\prime
}\right) \left( \varphi (L+1)-\varphi (L)\right) \\
&>&\varphi (K)-\sum_{L=0}^{K-1}\left( 1-\sum_{k=L+1}^{K}p _{k}\right) \left(
\varphi (L+1)-\varphi (L)\right) \\
&=&\varphi (K)-\sum_{L=0}^{K-1}\left( \sum_{k=0}^{L}p _{k}\right) \left(
\varphi (L+1)-\varphi (L)\right) \\
&=&\sum_{k=0}^{K}p _{k}\varphi (k),
\end{eqnarray*}%
where the first and the last equalities hold by Abel's formula for summation
by parts while the strict inequality uses the fact that (1) $p ^{\prime }$
stochastically dominates $p $; (2) $\varphi $ is a nondecreasing function
and (3) there is $\kappa \geq 1$ such that
\begin{equation*}
\sum_{k=\kappa }^{K}p _{k}^{\prime }>\sum_{k=\kappa }^{K}p _{k}\text{ and }%
\varphi (\kappa )>\varphi (\kappa -1)\text{.}
\end{equation*}

\section{Remaining Proof of   \Cref{thm:cutoff}}

\label{sup: proof final prop}

In this section we prove \Cref{prop: final} which completes the proof of %
\Cref{thm:cutoff}.

\subsubsection{Existence of a solution in the infinite-dimensional problem}

Our problem $[P^{\prime }]$ can be written as%
\begin{equation*}
\max_{p \in M^{\prime }}\sum_{k=0}^{\infty }p _{k}\left[ \mu _{k}((1-\alpha
)R+\alpha V)-\alpha Ck\right] \leqno{[P']}
\end{equation*}%
where $M^{\prime }\triangleq\{ p \in \Delta (\mathbb{Z}_{+}):\sum_{k=0}^{%
\infty }p _{k}\left[ \mu _{k}V-Ck\right] \geq 0,\lambda _{k}p _{k}\geq \mu
_{k+1}p _{k+1},\, \forall k\}$. We prove the following result.

\begin{proposition}
\label{thm: existence infinite}The set of optimal solutions of $[P^{\prime
}] $ is nonempty.
\end{proposition}

We start by showing that the objective of the optimization problem is upper
semi-continuous (\cref{upper semi-cont obj}). We endow $\mathbb{Z}_{+}$ with
the discrete topology and $\Delta (\mathbb{Z}_{+})$ with the weak topology.
Since $\mathbb{Z}_{+}$ endowed with the discrete topology is a (separable)
metric space, $\Delta (\mathbb{Z}_{+})$ is metrizable by Prokhorov's
Theorem. We next show that set $M^{\prime }$ is compact (%
\cref{prop:compactness}). This enough for our purpose. Indeed, by the
Extreme Value Theorem for upper semi-continuous functions, optimization
problem $[P^{\prime }]$ has an optimal solution.

\begin{proposition}
\label{upper semi-cont obj} The function%
\begin{equation*}
\sum_{k=0}^{\infty }p_{k}\left[ \mu _{k}((1-\alpha )R+\alpha V)-\alpha Ck%
\right]
\end{equation*}%
is upper semi-continuous in $p\in \Delta (\mathbb{Z}_{+})$.
\end{proposition}

\begin{proof}
Consider a sequence $\{ p  ^{n}\}$ in $\Delta  (\mathbb{Z}_{+})$ converging
to $p  ^{\ast }$. Since the function $k\mapsto \mu  _{k}((1-\alpha
)R+\alpha V)-\alpha Ck$ is continuous (in the discrete topology) and upper
bounded,\footnote{%
 {Recall our assumption that $\mu _k$ is uniformly bounded.}} by Portmanteau's Theorem, $%
\lim \sup \sum_{k=0}^{\infty }p  _{k}^{n}\left[ \mu  _{k}((1-\alpha
)R+\alpha V)-\alpha Ck\right] \leq \sum_{k=0}^{\infty }p  _{k}^{\ast }\left[
\mu  _{k}((1-\alpha )R+\alpha V)-\alpha Ck\right] $ and so we get the
upper semi-continuity of our function. \end{proof}

\begin{proposition}
\label{prop:compactness} Set $M^{\prime }$ is compact.
\end{proposition}

\begin{proof}
The proof is based on the two lemmas proved below.

\begin{lemma}
The set $M^{\prime }$ is tight.
\end{lemma}

\begin{proof}
We need to show that for any $\varepsilon >0$, there is $n$ large enough so
that any probability measure $p  \in M^{\prime }\ $has $\sum_{k=n+1}^{%
\infty }p  _{k}<\varepsilon $. Suppose to the contrary that
there is $\varepsilon >0$ and a sequence $\{ p  ^{n}\}_{n}$ in $M^{\prime }$
(which satisfies $\sum_{k=0}^{\infty }p  _{k}^{n}\left[ \mu  _{k}V-Ck%
\right] \geq 0$) such that $\sum_{k=n+1}^{\infty }p  _{k}^{n}>\varepsilon $
for all $n$. This implies
\begin{eqnarray*}
\sum_{k=0}^{\infty }p  _{k}^{n}(\mu  _{k}V-Ck) &=&V\sum_{k=0}^{\infty
}p  _{k}^{n}\mu  _{k}-C\sum_{k=0}^{\infty }p  _{k}^{n}k \\
&\leq & \sup_k \mu_k V-C\sum_{k=n+1}^{\infty }p  _{k}^{n}k \\
&\leq & \sup_k \mu_k V-C(n+1)\sum_{k=n+1}^{\infty }p  _{k}^{n} \\
&\leq & \sup_k \mu_k V-C(n+1)\varepsilon \text{.}
\end{eqnarray*}%
Note that for $n$ large enough, using our assumption that $\sup_k \mu_k<+\infty$, the above term must be strictly negative.
This contradicts the fact that $\sum_{k=0}^{\infty }p  _{k}^{n}(\mu
_{k}V-Ck)\geq 0$ for all $n$. \end{proof}

\begin{lemma}
The set $M^{\prime }$ is closed.
\end{lemma}

\begin{proof}
To show that $M^{\prime }$ is closed, we need to show that it contains all its limit points.
Recall that since $\Delta  (\mathbb{Z}_{+})\ $is a metric space, $p  \in
\Delta  (\mathbb{Z}_{+})$ is a limit point of $M^{\prime }$ if and only if there is a
sequence of points in $M^{\prime }\backslash \{ p  \}$ converging to $p  $.
Take any sequence $\{ p  ^{n}\}_{n}$ in $M^{\prime }$ converging to $p
^{\ast }$. We need to show that (1) $\sum_{k=0}^{\infty }p  _{k}^{\ast
}(\mu  _{k}V-Ck)\geq 0$ and (2) for all $k,$ $\lambda _{k}p  _{k}^{\ast
}\geq \mu  _{k+1}p  _{k+1}^{\ast }$.

\underline{(1) $\sum_{k=0}^{\infty }p  _{k}^{\ast }(\mu  _{k}V-Ck)\geq 0$.%
} Proceed by contradiction and assume that $\sum_{k=0}^{\infty }p
_{k}^{\ast }(\mu  _{k}V-Ck)<0$. By Portmanteau's Theorem, since the
function $k\mapsto \mu  _{k}V-Ck$ is bounded above (and trivially
continuous in the discrete topology), we must have that $\lim \sup
\sum_{k=0}^{\infty }p  _{k}^{n}(\mu  _{k}V-Ck)\leq \sum_{k=0}^{\infty
}p  _{k}^{\ast }(\mu  _{k}V-Ck)$. Hence, since, by assumption, $%
\sum_{k=0}^{\infty }p  _{k}^{\ast }(\mu  _{k}V-Ck)<0$, it must be that
for $n$ large enough, $\sum_{k=0}^{\infty }p  _{k}^{n}(\mu  _{k}V-Ck)<0$,
a contradiction with the fact that $p  ^{n}\in M^{\prime }$.

\underline{(2) For all $k,$ $\lambda _{k}p  _{k}^{\ast }\geq \mu
_{k+1}p  _{k+1}^{\ast }$.}\ By contradiction, assume that for some $k$, $%
\lambda _{k}p  _{k}^{\ast }<\mu  _{k+1}p  _{k+1}^{\ast }$. Since $p
_{k}^{n}$ and $p  _{k+1}^{n}$ converge pointwise to $p  _{k}^{\ast }$ and $%
p  _{k+1}^{\ast }$, for $n$ large enough we have $\lambda _{k}p
_{k}^{n}<\mu _{k+1}p  _{k+1}^{n}$ which contradicts the fact that $p
^{n} $ is in $M^{\prime }$. \end{proof}

Since $M^{\prime }$ is closed and tight, by Prokhorov Theorem, $M^{\prime }$
must be sequentially compact. Since $\Delta  (\mathbb{Z}_{+})\ $is a metric
space, this implies that $M^{\prime }$ is compact, as claimed. \end{proof}

\subsubsection{Completion of the proof of \cref{prop:
final}}

Let $M^{\prime \prime }$ be the set of $p$'s in $M^{\prime }$ which exhibits
a cutoff policy. That is any $p\in M^{\prime \prime }$ satisfies for some $%
\hat{K}$, $\lambda _{k}p_{k}=\mu _{k+1}p_{k+1},\forall k=0,..\hat{K}-1$ and $%
p_{k}=0$ for all $k\geq \hat{K}+1$. We define the sequence $\{p^{K}\}_{K}$
where, for each $K$, $p^{K}$ is an optimal solution of $[P_{K}^{\prime }]$.
If $\mu $ is regular, we assume that $p^{\ast }$ exhibits a cutoff policy
which is well-defined by \cref{prop:cutoff finite
dim}. In addition, for each $K$, we see $p^{K}$ as a point in $\mathbb{R}^{%
\mathbb{Z}_{+}}$ with $p_{k}^{K}=0$ for all $k\geq K+1$. Clearly $%
\{p^{K}\}_{K}$ is a sequence in $M^{\prime \prime }$. In the next
proposition we show that $M^{\prime \prime }$ is (sequentially) compact.
This will show that $\{p^{K}\}_{K}$ must have a subsequence converging to a
point that exhibits a cutoff policy. In the sequel, we assume that $\mu $ is
regular.

\begin{proposition}
\label{lem:compactness final} $\{p^{K}\}_{K}$ must have a subsequence
converging to a feasible point $p^{\ast }$ of $[P^{\prime }]$ that exhibits
a cutoff policy. In addition, $p_{k}^{\ast }>0$ for each $k\leq \min \arg
\max_{k}\mu _{k}V-Ck$.
\end{proposition}

\begin{proof}
For the first part of the statement, it suffices to show that $M^{\prime
\prime }$ is (sequentially) compact. Since $M^{\prime \prime }$ is a subset
of $M^{\prime }$ which is compact (\cref{prop:compactness}), we only need to
show that $M^{\prime \prime }$ is closed. Consider a sequence $\{p^{n}\}$ in
$M^{\prime \prime }$ converging to $p^{\ast }$. We show that $p^{\ast }\in
M^{\prime \prime }$. Since $M^{\prime }$ is (sequentially) compact, we
already know that $p^{\ast }\in M^{\prime }$. Letting $\hat{K}$ be the
largest state in the support of $p^{\ast }$ (which is potentially $+\infty $
if the support is unbounded), we proceed by contradiction and assume that
there exists $k_{0}<\hat{K}$ such that $\lambda
_{k_{0}-1}p_{k_{0}-1}^{_{\ast }}>\mu _{k_{0}}p_{k_{0}}^{_{\ast }}$. Now,
simply pick $n$ large enough so that (1) $p_{k}^{n}>0$ for all $%
k=0,...,k_{0}+1$ and (2) $\lambda _{k_{0}-1}p_{k_{0}-1}^{n}>\mu
_{k_{0}}p_{k_{0}}^{n}$. This contradicts the assumption that $p^{n}$ is in $%
M^{\prime \prime }$. We thus conclude that $p^{\ast }\in M^{\prime \prime }$.

We now show the second part of the statement. We just proved that $%
\{p^{K}\}_{K}$ must have a subsequence converging to a feasible point $%
p^{\ast }$ of $[P^{\prime }]$. We show that $p^{\ast }$ satisfies $%
p_{k}^{\ast }>0$ for each $k\leq \min \arg \max_{k}\mu _{k}V-Ck$. First, we
simply observe that for any $\xi \geq 0$, $\min \arg \max_{k}\mu
_{k}V-Ck\leq \min \arg \max f(k;\xi )$.\footnote{%
Straightforward algebra show that $\mu _{k+1}V-C(k+1)>\mu _{k}V-Ck$ if
and only if $\mu _{k+1}-\mu _{k}>C/V$.\ Similarly,\ given $\xi \geq 0$, we
have that $f(k+1;\xi )>f(k;\xi )$ if and only if $\mu _{k+1}-\mu
_{k}>C/[(1-\alpha )/(\alpha +\xi )+V]$. Hence, whenever $\mu _{k}V-Ck$ is
strictly increasing from $k$ to $k+1$, so is $f(k;\xi )$. Since by %
\cref{lem: single-peakedness}, these functions are single-peaked, we must
have $\min \arg \max_{k}\mu _{k}V-Ck\leq \min \arg \max f(k;\xi )$.} Now, we
proceed by contradiction and assume that there is $k_{0}\leq \min \arg
\max_{k}\mu _{k}V-Ck$ such that $p_{k_{0}}^{\ast }=0$. Let us assume that $%
k_{0}$ is the smallest state satisfying this property, so, in particular, $%
p_{k_{0}-1}^{\ast }>0$. This implies that $p_{k_{0}}^{\ast }\mu
_{k_{0}}<p_{k_{0}-1}^{\ast }\lambda _{k_{0}-1}$. Since $\{p^{K}\}_{K}$
converges to $p^{\ast }$, for $K$ large enough, $p_{k_{0}}^{K}\mu
_{k_{0}}<p_{k_{0}-1}^{K}\lambda _{k_{0}-1}$. Since $k_{0}\leq \min \arg \max
f(k;\xi _{K}^{\ast })$, using single-peakedness of $f(\cdot ;\xi _{K}^{\ast
})$, we must have $f(k_{0}-1;\xi _{K}^{\ast })<f(k_{0};\xi _{K}^{\ast })$
(where we use the notation $(p^{K},\xi _{K}^{\ast })$ for the saddle point
of the Lagrangian in $[P_{K}^{\prime }]$). This contradicts %
\cref{claim:increasing part}. \end{proof}

Finally, we complete the proof of \cref{prop: final}$\ $via the following
proposition.

\begin{proposition}
\label{prop:final step} Take any subsequence of $\{p^{K}\}_{K}$ converging
to a limit $p^{\ast }$. Then, $p^{\ast }$ must be an optimal solution of $%
[P^{\prime }]$.
\end{proposition}

\begin{proof}
In the sequel, we let $p^{\ast }$ be the limit of an arbitrary converging
subsequence $\{p^{K}\}_{K}$. We proceed by contradiction and assume that $%
p^{\ast }$ is not a solution to the infinite dimensional problem. By %
\cref{thm: existence infinite}, we know that there is a solution to this
problem. Let us call it $\bar{p}$. By assumption,
\begin{equation}
\sum_{k=0}^{\infty }\bar{p}_{k}\left[ \mu _{k}\left( \left( 1-\alpha \right)
R+\alpha V\right) -\alpha Ck\right] >\sum_{k=0}^{\infty }p_{k}^{\ast }\left[
\mu _{k}\left( \left( 1-\alpha \right) R+\alpha V\right) -\alpha Ck\right]
\text{.}  \label{ineq contrad}
\end{equation}%
Now, let us note by $\bar{p}^{K}$ the distribution $\bar{p}$ conditional on $%
\{0,...,K\}$, i.e., $\bar{p}_{k}^{K}=0$ for all $k\geq K+1$ while $\bar{p}%
_{k}^{K}=\bar{p}_{k}\left/ \sum_{k=0}^{K}\bar{p}_{k}\right. $ for all $k\leq
K$. We claim that
\begin{equation*}
\lim \sum_{k=0}^{\infty }\bar{p}_{k}^{K}\left[ \mu _{k}\left( \left(
1-\alpha \right) R+\alpha V\right) -\alpha Ck\right] =\sum_{k=0}^{\infty }%
\bar{p}_{k}\left[ \mu _{k}\left( \left( 1-\alpha \right) R+\alpha V\right)
-\alpha Ck\right] .
\end{equation*}%
Indeed, by construction, for each $K$,%
\begin{equation*}
\sum_{k=0}^{\infty }\bar{p}_{k}^{K}\left[ \mu _{k}\left( \left( 1-\alpha
\right) R+\alpha V\right) -\alpha Ck\right] =\sum_{k=0}^{K}\bar{p}_{k}\left[
\mu _{k}\left( \left( 1-\alpha \right) R+\alpha V\right) -\alpha Ck\right]
\left/ \sum_{k=0}^{K}\bar{p}_{k}\right. .
\end{equation*}%
Taking limits on both sides as $K\rightarrow \infty $ (and using the fact
that $\lim_{K\rightarrow \infty }\sum_{k=0}^{K}\bar{p}_{k}=1$), we obtain%
\begin{equation*}
\lim_{K\rightarrow \infty }\sum_{k=0}^{\infty }\bar{p}_{k}^{K}\left[ \mu
_{k}\left( \left( 1-\alpha \right) R+\alpha V\right) -\alpha Ck\right]
=\sum_{k=0}^{\infty }\bar{p}_{k}\left[ \mu _{k}\left( \left( 1-\alpha
\right) R+\alpha V\right) -\alpha Ck\right] ,
\end{equation*}%
as claimed.

Now, using Equation \cref{ineq contrad}, for $K$ large enough, we must have
\begin{equation}
\sum_{k=0}^{\infty }\bar{p}_{k}^{K}\left[ \mu _{k}\left( \left( 1-\alpha
\right) R+\alpha V\right) -\alpha Ck\right] >\sum_{k=0}^{\infty }p_{k}^{\ast
}\left[ \mu _{k}\left( \left( 1-\alpha \right) R+\alpha V\right) -\alpha Ck%
\right] +\varepsilon   \label{Ineq 1}
\end{equation}%
for some $\varepsilon >0$. Now, since $\{p^{K}\}_{K}$ converges weakly to $%
p^{\ast }$, by \cref{upper semi-cont obj},%
\begin{equation*}
\lim_{K\rightarrow \infty }\sup \sum_{k=0}^{\infty }p_{k}^{K}\left[ \mu
_{k}\left( \left( 1-\alpha \right) R+\alpha V\right) -\alpha Ck\right] \leq
\sum_{k=0}^{\infty }p_{k}^{\ast }\left[ \mu _{k}\left( \left( 1-\alpha
\right) R+\alpha V\right) -\alpha Ck\right] .
\end{equation*}%
Hence, we must have that for $K$ large enough,
\begin{equation}
\sum_{k=0}^{\infty }p_{k}^{\ast }\left[ \mu _{k}\left( \left( 1-\alpha
\right) R+\alpha V\right) -\alpha Ck\right] +\varepsilon >\sum_{k=0}^{\infty
}p_{k}^{K}\left[ \mu _{k}\left( \left( 1-\alpha \right) R+\alpha V\right)
-\alpha Ck\right] \text{.}  \label{Ineq 2}
\end{equation}%
Using Equation \cref{Ineq 1} and \cref{Ineq 2}, we conclude that for $K$
large enough,
\begin{equation*}
\sum_{k=0}^{\infty }\bar{p}_{k}^{K}\left[ \mu _{k}\left( \left( 1-\alpha
\right) R+\alpha V\right) -\alpha Ck\right] >\sum_{k=0}^{\infty }p_{k}^{K}%
\left[ \mu _{k}\left( \left( 1-\alpha \right) R+\alpha V\right) -\alpha Ck%
\right] \text{.}
\end{equation*}%
This contradicts the fact that $p^{K}$ is an optimal solution of $%
[P_{K}^{\prime }]$ since $\bar{p}^{K}$ is feasible in this problem.
\end{proof}

\section{Proof that FCFS with no information satisfies $(IC_0)$.}

\label{online_app:IC_0}

Recall that policy $(q^*, I^*)$ stands for FCFS queueing rule and the no
information (beyond recommendations) rule.

\begin{lemma}
\label{lem:ic0-fcfs} The queueing/information policy $(q^*,I^*)$ satisfies $%
(IC_0)$.
\end{lemma}

\begin{proof}
Recall the optimality of the cutoff policy means $%
x^*_k=1$ for all $k=0,..., K^*-2$ and $x^*_k=0$ for all $k> K^*-1$, and $%
y^*_{k,\ell}=z^*_{k,\ell}=0$ for all $(k,\ell)$. Substitute these into %
\cref{B}. Use the resulting equations to rewrite \cref{eq:belief0}:
\begin{equation*}
\tilde \gamma_{\ell}^0= \frac{p^* _{\ell}\mu _{\ell}}{ \sum_{i=1}^{K^*}
p_i^* \mu _i}, \forall \ell=1,..., K^*.
\end{equation*}
An agent's expected payoff when joining the queue after being recommended to
do so is:
\begin{align*}
V - C \sum_{k=1}^{K^*} \tilde{\gamma}_{k}^{0} \cdot \tau^*_{k} = &V - C
\frac{\sum_{k=1}^{K^*} p^* _{k}\mu _{k} \cdot \tau^*_{k}}{ \sum_{i=1}^{K^*}
p_i^* \mu _i} \\
= & V - C \frac{\sum_{k =1}^{K^*} p^* _{k} k }{ \sum_{i=1}^{K^*} p_i^* \mu _i%
} \\
= & \left(\frac{1}{ \sum_{i=1}^{K^*} p_i^* \mu _i}\right) \sum_{k=1}^{K^*}
p_k^* \left( \mu _k V- kC \right),
\end{align*}
where the first equality is from the preceding observation and the second
equality follows from \cref{lem:waiting-time-FCFS}. Since $\sum_{i=1}^{K^*}
p_i^* \mu _i>0$, $(IC_{0})$ holds if and only if \cref{IR} holds.
\end{proof}

\section{The analyis of the belief ODEs and the Proof of \cref{lem:dyn-incentives} when $\bar{K}=\infty $.}
\label{online_app:infty}

We first derive the infinite system of ODEs in terms of agents' belief of
occupying queue position $\ell =1,...,\infty $ at time $t$.  We first derive \cref{Eq: cond beliefs}.    Define the following events:
\begin{align*}
    A^t & :=\{ \mbox{not served by time } t\}\\
    A^t_{\ell} & :=\{ \mbox{not served by time } t \mbox{ and has position } \ell \mbox{ at time } t\}\\
        B^{dt}_{-k} & :=\{ \mbox{no agent with position } j\le k \mbox{ is served during } [t, t+dt)\}\\
\end{align*}

Observe first
$$\Pr\{ A^{t+dt}_{\ell}\}   =\Pr\{ A^{t+dt}_{\ell}|A^{t+dt}\} \Pr\{ A^{t+dt}\}   =\gamma_{\ell}^{t+dt} \Pr\{ A^{t+dt}\}.$$  Next,
\begin{align*}
    \Pr\{ A^{t+dt}_{\ell}\}
       & =  \Pr\{ A^{t+dt}_{\ell}|A^{t}\} \Pr\{ A^{t}\} \\
    & = \Pr\{B^{dt}_{-\ell}\cap A^{t}_{\ell} |A^{t}_{\ell}\} \Pr\{A^{t}_{\ell}|A^{t}\} \Pr\{ A^{t}\} \\
      & \quad + \Pr\{\neg B^{dt}_{-\ell}\cap A^{t}_{\ell+1} |A^{t}_{\ell+1}\} \Pr\{A^{t}_{\ell+1}|A^{t}\} \Pr\{ A^{t}\}+o(dt) \\
         & =\left( \left(1-\sum_{i=1}^{\ell}q_idt\right)\gamma_{\ell}^t  +  \left(\sum_{i=1}^{\ell}q_i\right) dt\gamma_{\ell+1}^t\right)\Pr\{ A^{t}\}+o(dt).
\end{align*}

Equating the two and rearranging, we get
\begin{align*}
   \gamma_{\ell}^{t+dt} & =\left( \left(1-\sum_{i=1}^{\ell}q_idt\right)\gamma_{\ell}^t  +  \left(\sum_{i=1}^{\ell}q_i\right) dt\gamma_{\ell+1}^t\right)\frac{\Pr\{ A^{t}\}}{\Pr\{ A^{t+dt}\}}+o(dt) \\
   & =\left( \left(1-\sum_{i=1}^{\ell}q_idt\right)\gamma_{\ell}^t  +  \left(\sum_{i=1}^{\ell}q_i\right) dt\gamma_{\ell+1}^t\right)\frac{1}{\Pr\{ A^{t+dt}| A^{t}\}}+o(dt)\\
      & =\left( \left(1-\sum_{i=1}^{\ell}q_idt\right)\gamma_{\ell}^t  +  \left(\sum_{i=1}^{\ell}q_i\right) dt\gamma_{\ell+1}^t\right)\frac{1}{\left(1-\sum_{i=1}^{K}\gamma_{i}^t q_idt\right)}+o(dt), \\
\end{align*}
where the second equality holds since
$\Pr\{ A^{t+dt}\}=\Pr\{ A^{t+dt}| A^{t}\} \Pr\{ A^{t}\}.$
We thus obtain \cref{Eq: cond beliefs}.

It follows from %
\cref{Eq: cond beliefs}, together with ${q}_{i}^{\ast }=\mu _{i}-\mu _{i-1}$%
, that
\begin{equation*}
\tilde{\gamma}_{\ell }^{t+dt}=\frac{(1-\mu _{\ell }dt)\tilde{\gamma}_{\ell
}^{t}+\mu _{\ell }dt\tilde{\gamma}_{\ell +1}^{t}}{\sum_{i=1}^{\bar{K}}\tilde{%
\gamma}_{i}^{t}(1-{q}_{i}^{\ast }dt)}+o(dt).
\end{equation*}%
\begin{eqnarray*}
\frac{\tilde{\gamma}_{k}^{t+dt}-\tilde{\gamma}_{k}^{t}}{dt} &=&\frac{(1-\mu
_{k}dt)\tilde{\gamma}_{k}^{t}+\mu _{k}dt\tilde{\gamma}_{k+1}^{t}}{%
dt\sum_{i=1}^{\infty }\tilde{\gamma}_{i}^{t}(1-{q}_{i}^{\ast }dt)}-\frac{%
\tilde{\gamma}_{k}^{t}}{dt}+\frac{o(dt)}{dt} \\
&=&\frac{(1-\mu _{k}dt)\tilde{\gamma}_{k}^{t}+\mu _{k}dt\tilde{\gamma}%
_{k+1}^{t}}{dt\left[ 1-dt\sum_{i=1}^{\infty }\tilde{\gamma}_{i}^{t}{q}%
_{i}^{\ast }\right] }-\frac{\tilde{\gamma}_{k}^{t}}{dt}+\frac{o(dt)}{dt} \\
&=&\frac{(1-\mu _{k}dt)\tilde{\gamma}_{k}^{t}+\mu _{k}dt\tilde{\gamma}%
_{k+1}^{t}}{dt\left[ 1-dt\sum_{i=1}^{\infty }\tilde{\gamma}_{i}^{t}(\mu
_{i}-\mu _{i-1})\right] }-\frac{\tilde{\gamma}_{k}^{t}}{dt}+\frac{o(dt)}{dt}
\\
&=&\frac{(1-\mu _{k}dt)\tilde{\gamma}_{k}^{t}+\mu _{k}dt\tilde{\gamma}%
_{k+1}^{t}-\tilde{\gamma}_{k}^{t}\left[ 1-dt\sum_{i=1}^{\infty }\tilde{\gamma%
}_{i}^{t}(\mu _{i}-\mu _{i-1})\right] }{dt\left[ 1-dt\sum_{i=1}^{\infty }%
\tilde{\gamma}_{i}^{t}(\mu _{i}-\mu _{i-1})\right] }+\frac{o(dt)}{dt} \\
&=&\frac{-\mu _{k}\tilde{\gamma}_{k}^{t}+\mu _{k}\tilde{\gamma}_{k+1}^{t}+%
\tilde{\gamma}_{k}^{t}\left[ \sum_{i=1}^{\infty }\tilde{\gamma}_{i}^{t}(\mu
_{i}-\mu _{i-1})\right] }{\left[ 1-dt\sum_{i=1}^{\infty }\tilde{\gamma}%
_{i}^{t}(\mu _{i}-\mu _{i-1})\right] }+\frac{o(dt)}{dt}.
\end{eqnarray*}%
Letting $dt\rightarrow 0$, we obtain: for all $k\in \mathbb{N}$,
\begin{equation}
\dot{\tilde{\gamma}}_{k}^{t}=-\mu _{k}\tilde{\gamma}_{k}^{t}+\mu _{k}\tilde{%
\gamma}_{k+1}^{t}+\tilde{\gamma}_{k}^{t}\left[ \sum_{i=1}^{\infty }\tilde{%
\gamma}_{i}^{t}(\mu _{i}-\mu _{i-1})\right] \triangleq f_{k}(\tilde{\gamma}%
^{t}),  \label{Eq: ODEs}
\end{equation}%
and let $f\triangleq (f_{k})_{k\in \mathbb{N}}$. The following proposition
states that, given an initial condition, this system of ODEs has a unique
solution.

\begin{proposition}
\label{prop:existence_uniqueness} For any initial condition in $\Delta (%
\mathbb{N})$, there is a unique solution to the system of ODEs given by %
\cref{Eq: ODEs}.
\end{proposition}

\begin{proof}
Let $\mathbf{X}$ be the set of sequences in $\ell ^{1}$-space endowed with $%
\ell ^{1}$-norm. As is well-known, this is a Banach space. Clearly, $\Delta (%
\mathbb{N})\subseteq \mathbf{X}$. Further, we can see that $f$ maps from $%
\mathbf{X}$ to $\mathbf{X}$. Indeed, for any $\tilde{\gamma}^{t}\in \mathbf{X%
}:$
\begin{eqnarray*}
||f(\tilde{\gamma}^{t})|| &=&\sum_{k=1}^{\infty }\left\vert f_{k}(\tilde{%
\gamma}_{k}^{t})\right\vert \\
&=&\sum_{k=1}^{\infty }\left\vert -\mu _{k}\tilde{\gamma}_{k}^{t}+\mu _{k}%
\tilde{\gamma}_{k+1}^{t}+\tilde{\gamma}_{k}^{t}\left[ \sum_{i=1}^{\infty }%
\tilde{\gamma}_{i}^{t}(\mu _{i}-\mu _{i-1})\right] \right\vert \\
&\leq &\sum_{k=1}^{\infty }\left\vert -\mu _{k}\tilde{\gamma}%
_{k}^{t}\right\vert +\sum_{k=1}^{\infty }\left\vert \mu _{k}\tilde{\gamma}%
_{k+1}^{t}\right\vert +\sum_{k=1}^{\infty }\left\vert \tilde{\gamma}_{k}^{t}%
\left[ \sum_{i=1}^{\infty }\tilde{\gamma}_{i}^{t}(\mu _{i}-\mu _{i-1})\right]
\right\vert \\
&\leq &\bar{\mu}\sum_{k=1}^{\infty }\left\vert \tilde{\gamma}%
_{k}^{t}\right\vert +\bar{\mu}\sum_{k=1}^{\infty }\left\vert \tilde{\gamma}%
_{k+1}^{t}\right\vert +\bar{\mu}\left( \sum_{k=1}^{\infty }\left\vert \tilde{%
\gamma}_{k}^{t}\right\vert \right) \left( \sum_{i=1}^{\infty }\left\vert
\tilde{\gamma}_{i}^{t}\right\vert \right) <\infty
\end{eqnarray*}%
where we recall that $\bar{\mu}\triangleq \sup_{k}\mu _{k}<\infty $ and use
the fact that $\tilde{\gamma}^{t}\in \mathbf{X}$. Hence, we have $f(\tilde{%
\gamma}^{t})\in \mathbf{X}$.

\begin{lemma}\label{lemma:Lipschitz}
Consider the restriction of $f$ defined as follows $f:U\rightarrow \mathbf{X}
$ where $U\triangleq \{\{x_{k}\}_{k\geq 1}\in \mathbf{X}:\sum_{k=1}^{\infty
}\left\vert x_{k}\right\vert <1+\varepsilon \}\subset \mathbf{X}$, for some $%
\varepsilon >0$, is an open set containing $\Delta (\mathbb{N})$. Mapping $f$
(restricted to $U$) is Lipschitz continuous.
\end{lemma}

\begin{proof}
Indeed, for any $\tilde{\gamma}$ and $\tilde{\gamma}^{\prime
}$ in $U$,
\begin{eqnarray*}
\left\Vert f(\tilde{\gamma}^{\prime })-f(\tilde{\gamma})\right\Vert
&=&\sum_{k=1}^{\infty }\left\vert f_{k}(\tilde{\gamma}^{\prime })-f_{k}(%
\tilde{\gamma})\right\vert \\
&\leq &\sum_{k=1}^{\infty }\left\vert -\mu _{k}\tilde{\gamma}_{k}^{\prime
}+\mu _{k}\tilde{\gamma}_{k}\right\vert +\sum_{k=1}^{\infty }\left\vert \mu
_{k}\tilde{\gamma}_{k+1}^{\prime }-\mu _{k}\tilde{\gamma}_{k+1}\right\vert \\
&&+\sum_{k=1}^{\infty }\left\vert \tilde{\gamma}_{k}^{\prime }\left[
\sum_{i=1}^{\infty }\tilde{\gamma}_{i}^{\prime }(\mu _{i}-\mu _{i-1})\right]
-\tilde{\gamma}_{k}\left[ \sum_{i=1}^{\infty }\tilde{\gamma}_{i}(\mu
_{i}-\mu _{i-1})\right] \right\vert \\
&\leq &\sum_{k=1}^{\infty }\mu _{k}\left\vert \tilde{\gamma}_{k}^{\prime }-%
\tilde{\gamma}_{k}\right\vert +\sum_{k=1}^{\infty }\mu _{k}\left\vert \tilde{%
\gamma}_{k+1}^{\prime }-\tilde{\gamma}_{k+1}\right\vert \\
&&+\max \{\sum_{i=1}^{\infty }\left\vert \tilde{\gamma}_{i}^{\prime
}\right\vert (\mu _{i}-\mu _{i-1}),\sum_{i=1}^{\infty }\left\vert \tilde{%
\gamma}_{i}\right\vert (\mu _{i}-\mu _{i-1})\}\sum_{k=1}^{\infty }\left\vert
\tilde{\gamma}_{k}^{\prime }-\tilde{\gamma}_{k}\right\vert \\
&\leq &\bar{\mu}\left\Vert \tilde{\gamma}^{\prime }-\tilde{\gamma}%
\right\Vert +\bar{\mu}\left\Vert \tilde{\gamma}^{\prime }-\tilde{\gamma}%
\right\Vert +(1+\varepsilon )\bar{\mu}\left\Vert \tilde{\gamma}^{\prime }-%
\tilde{\gamma}\right\Vert .
\end{eqnarray*}
Thus, $f$ restricted to $U$ is Lipschitz continuous with Lipschitz constant
equal to $\bar{\mu}(3+\varepsilon )$.
\end{proof}
In order to complete the proof of \Cref{prop:existence_uniqueness}, let us consider the system
of ODEs given by \cref{Eq: ODEs} where the vector field $f$ is the mapping
from $\mathbf{X}$ to $\mathbf{X}$. Since $f$ is bounded and, by \Cref{lemma:Lipschitz}, Lipschitz
continuous on $\Delta (\mathbb{N})$ and $\Delta (\mathbb{N})$ is positively
invariant, existence and uniqueness of a solution for our system of ODEs
with initial condition in $\Delta (\mathbb{N})$ follows from Picard-Lindel%
\"{o}f Theorem on Banach spaces.\footnote{%
Recall that a subset $S$ of $\mathbf{X}$ is positively invariant if no
solution starting inside $S$ can leave $S$ in the future.}
\end{proof}

In the sequel, we consider solutions to the system of ODEs when the entry
rule $x^{\ast }$ is \textquotedblleft truncated\textquotedblright to $x^K$, i.e.,
where $x_{k}^{K}=x_{k}^{\ast }=1$ for all $k\leq K$ and $x_{k}^{K}=0$
otherwise. We show that solutions to the system of ODEs under the truncated
cutoff policy $(x^{K},y^{\ast },z^{\ast })$ approximate solutions to the
system under the original cutoff policy $(x^{\ast },y^{\ast },z^{\ast })$.
More specifically, we let $\tilde{\gamma}^{K}(t)=(\tilde{\gamma}%
_{k}^{K}(t))_{k\in \mathbb{N}}$ denote a solution to the system given by %
\cref{Eq: ODEs}
\begin{equation*}
\dot{\tilde{\gamma}}^{t}=f(\tilde{\gamma}^{t}),
\end{equation*}%
when $\tilde{\gamma}^{0}=\tilde{\gamma}^{K}(0)\triangleq (\tilde{\gamma}%
_{k}^{K}(0))_{k\in \mathbb{N}}$ where $\tilde{\gamma}_{k}^{K}(0)$ is an
agent's belief of entering the queue with position $k$ at $t=0$ under the
truncated cutoff policy.\footnote{%
Note that $\tilde{\gamma}_{k}^{K}(0)=0$ for all $k>K$.} Meanwhile, $\tilde{%
\gamma}^{\infty }(t)=(\tilde{\gamma}_{k}^{\infty }(t))_{k\in \mathbb{N}}$
denotes a solution to this system of ODEs when $\tilde{\gamma}^{0}=\tilde{%
\gamma}^{\infty }(0)\triangleq (\tilde{\gamma}_{k}^{\infty }(0))_{k\in
\mathbb{N}}$ where $\tilde{\gamma}_{k}^{\infty }(0)$ is an agent's belief of
entering the queue with position $k$ at $t=0$ under the original cutoff
policy. We show that solution $\tilde{\gamma}^{K}$ converges to solution $%
\tilde{\gamma}^{\infty }$ when $K\ $goes to infinity.



\begin{lemma}
\label{lemma:convergence_solutions} The solution $\tilde{\gamma}^{K}$
converges pointwise to the solution $\tilde{\gamma}^{\infty }$, i.e., for
each $t>0$,
\begin{equation*}
\lim_{K\rightarrow \infty }|| \tilde{\gamma}^{K}(t)-\tilde{\gamma}^{\infty
}(t)|| =0.
\end{equation*}
\end{lemma}

\begin{proof}
The following two steps prove the lemma.

\begin{step}
\label{step1'} $\left\Vert \tilde{\gamma} ^{K}(0)-\tilde{\gamma} ^{\infty
}(0)\right\Vert \rightarrow 0$ as $K\to \infty$.
\end{step}

\begin{proof}
We know that for all $\ell =2,...,K:\tilde{\gamma}_{\ell
}^{K}(0)=\prod_{i=2}^{\ell }r_{i}^{0}\tilde{\gamma}_{1}^{K}(0)=\prod_{i=2}^{%
\ell }\frac{\tilde{\lambda} _{i-1}}{\mu _{i-1}}\tilde{\gamma}_{1}^{K}(0)$, where we
used \cref{Eq r0}, while $\tilde{\gamma}_{\ell }^{K}(0)=0$ for $\ell \geq
K+1 $. In addition, we know that $\tilde{\gamma}_{\ell }^{\infty
}(0)=\prod_{i=2}^{\ell }\frac{\tilde{\lambda}_{i-1}}{\mu _{i-1}}\tilde{\gamma%
}_{1}^{\infty }(0)$ and $\sum_{k=1}^{\infty }\prod_{i=2}^{k}\frac{\tilde{%
\lambda}_{i-1}}{\mu _{i-1}}\tilde{\gamma}_{1}^{\infty }(0)=1$ where our
convention is that $\prod_{i=2}^{1}\frac{\tilde{\lambda}_{i-1}}{\mu _{i-1}}%
\triangleq 1$. Thus,
\begin{equation*}
\tilde{\gamma}_{1}^{\infty }(0)=\frac{1}{\sum_{k=1}^{\infty }\prod_{i=2}^{k}%
\frac{\tilde{\lambda}_{i-1}}{\mu _{i-1}}}.
\end{equation*}%
Note that this implies that $\sum_{k=1}^{\infty }\prod_{i=2}^{k}\frac{\tilde{%
\lambda}_{i-1}}{\mu _{i-1}}<\infty $. Similar computation yields
\begin{equation*}
\tilde{\gamma}_{1}^{K}(0)=\frac{1}{\sum_{k=1}^{K}\prod_{i=2}^{k}\frac{\tilde{%
\lambda}_{i-1}}{\mu _{i-1}}}.
\end{equation*}%
Note that $\left\vert \tilde{\gamma}_{1}^{K}(0)-\tilde{\gamma}_{1}^{\infty
}(0)\right\vert \rightarrow 0$ as $K$ increases. We have
\begin{eqnarray*}
\left\Vert \tilde{\gamma}^{K}(0)-\tilde{\gamma}^{\infty }(0)\right\Vert
&=&\dsum\limits_{k=1}^{\infty }\left\vert \tilde{\gamma}_{k}^{K}(0)-\tilde{%
\gamma}_{k}^{\infty }(0)\right\vert \\
&=&\dsum\limits_{k=1}^{K}\left\vert \tilde{\gamma}_{k}^{K}(0)-\tilde{\gamma}%
_{k}^{\infty }(0)\right\vert +\dsum\limits_{k=K+1}^{\infty }\left\vert
\tilde{\gamma}_{k}^{\infty }(0)\right\vert \\
&=&\dsum\limits_{k=1}^{K}\left\vert \dprod\limits_{i=2}^{k}\frac{\tilde{%
\lambda}_{i-1}}{\mu _{i-1}}\tilde{\gamma}_{1}^{K}(0)-\dprod\limits_{i=2}^{k}%
\frac{\tilde{\lambda}_{i-1}}{\mu _{i-1}}\tilde{\gamma}_{1}^{\infty
}(0)\right\vert +\dsum\limits_{k=K+1}^{\infty }\left\vert \tilde{\gamma}%
_{k}^{\infty }(0)\right\vert \\
&=&\left\vert \tilde{\gamma}_{1}^{K}(0)-\tilde{\gamma}_{1}^{\infty
}(0)\right\vert \dsum\limits_{k=1}^{K}\dprod\limits_{i=2}^{k}\frac{\tilde{%
\lambda}_{i-1}}{\mu _{i-1}}+\dsum\limits_{k=K+1}^{\infty }\left\vert \tilde{%
\gamma}_{k}^{\infty }(0)\right\vert .
\end{eqnarray*}%
Since $\left\vert \tilde{\gamma}_{1}^{K}(0)-\tilde{\gamma}_{1}^{\infty
}(0)\right\vert \rightarrow 0$ as $K\rightarrow \infty $, $%
\sum_{k=1}^{\infty }\prod_{i=2}^{k}\frac{\tilde{\lambda}_{i-1}}{\mu _{i-1}}%
<\infty $, and $\sum_{k=K+1}^{\infty }\left\vert \tilde{\gamma}_{k}^{\infty
}(0)\right\vert $ goes to $0$ as $K\rightarrow \infty $, the result follows.
\end{proof}

\begin{step}
\label{step2'} For each $t> 0$,
\begin{equation*}
\lim_{K\rightarrow \infty }\sum_{k=1}^{\infty }\left\vert \tilde{\gamma}
_{k}^{K}(t)-\tilde{\gamma} _{k}^{\infty }(t)\right\vert =0.
\end{equation*}
\end{step}

\begin{proof} By Gr\"onwall's inequality,
\begin{equation*}
\left\Vert \tilde{\gamma} ^{K}(t)-\tilde{\gamma} ^{\infty }(t)\right\Vert \leq
e^{Ct}\left\Vert \tilde{\gamma} ^{K}(0)-\tilde{\gamma} ^{\infty }(0)\right\Vert \text{,}
\end{equation*}
where, by \Cref{lemma:Lipschitz}, $C\triangleq \bar\mu(3+\varepsilon)$ is the Lipschitz constant
for the Lipschitz continuous function
$f$ restricted to open set $U=\{\{x_{k}\}_{k\geq 1}\in
\mathbf{X}:\sum_{k=1}^{\infty }\left\vert x_{k}\right\vert <1+\varepsilon \}$. The result then follows from \Cref{step1'}.
\end{proof}
\end{proof}
We now complete the proof of \cref{lem:dyn-incentives} when $\bar{K}%
=\infty $ with the following lemma.

\begin{lemma}
\label{step3'} $\dot{r}_{\ell }^{\infty }(t)\leq 0$ for all $\ell \geq 2$
and $t$, where $r_{\ell }^{\infty}(t)=\tilde{\gamma} _{\ell }^{\infty}(t)/%
\tilde{\gamma} _{\ell -1}^{\infty}(t) $ for all $\ell \geq 2$.
\end{lemma}

\begin{proof}
Recall from Equation \cref{eq:r-t} in the main text that the system of ODEs is given by%
\begin{equation*}
\dot{r}_{\ell }^{\infty }(t)=r_{\ell }^{\infty }(t)\left( \mu _{\ell -1}-\mu
_{\ell }-\mu _{\ell -1}r_{\ell }^{\infty }(t)+\mu _{\ell }r_{\ell
+1}^{\infty }(t)\right)
\end{equation*}%
for all $\ell \geq 2$. Suppose to the contrary that $\dot{r}_{\ell }^{\infty
}(t)>0$ for some $\ell $ and $t$. We already proved in \Cref{app:dyn-incentives}
(in the main text) where $\bar{K}<\infty $ that
\begin{equation*}
\dot{r}_{\ell }^{K}(t)=r_{\ell }^{K}(t)\left( \mu _{\ell -1}-\mu _{\ell
}-\mu _{\ell -1}r_{\ell }^{K}(t)+\mu _{\ell }r_{\ell +1}^{K}(t)\right) \leq 0
\end{equation*}%
for all $K<\infty $, $\ell $ and $t.$ To show a contradiction, it is enough
to prove that
\begin{equation*}
r_{\ell }^{K}(t)\left( \mu _{\ell -1}-\mu _{\ell }-\mu _{\ell -1}r_{\ell
}^{K}(t)+\mu _{\ell }r_{\ell +1}^{K}(t)\right) \rightarrow r_{\ell }^{\infty
}(t)\left( \mu _{\ell -1}-\mu _{\ell }-\mu _{\ell -1}r_{\ell }^{\infty
}(t)+\mu _{\ell }r_{\ell +1}^{\infty }(t)\right)
\end{equation*}%
as $K\rightarrow \infty $. To this end, it suffices to show that $r_{\ell
}^{K}(t)$ and $r_{\ell +1}^{K}(t)$ converge respectively to $r_{\ell
}^{\infty }(t)$ and $r_{\ell +1}^{\infty }(t)$. It follows from \Cref{lemma:convergence_solutions}
that for each $k:$%
\begin{equation*}
\lim_{K\rightarrow \infty }\tilde{\gamma}_{k}^{K}(t)=\tilde{\gamma}%
_{k}^{\infty }(t).
\end{equation*}%
By assumption $\tilde{\gamma}_{k}^{\infty }(0)>0$ for all $k$, so
\begin{equation*}
\lim_{K\rightarrow \infty }r_{\ell }^{K}(t)=\lim_{K\rightarrow \infty }\frac{%
\tilde{\gamma}_{\ell }^{K}(t)}{\tilde{\gamma}_{\ell -1}^{K}(t)}=\frac{\tilde{%
\gamma}_{\ell }^{\infty }(t)}{\tilde{\gamma}_{\ell -1}^{\infty }(t)}=r_{\ell
}^{\infty }(t)
\end{equation*}%
Similarly,
\begin{equation*}
\lim_{K\rightarrow \infty }r_{\ell +1}^{K}(t)=r_{\ell +1}^{\infty }(t),
\end{equation*}%
which completes the argument. \end{proof}


\section{The necessity of regularity for \Cref{thm:dyn-fcfs}}\label{Example_non-regular}

In the sequel, we provide an example where our primitive process \textit{is
not} regular and where \Cref{thm:dyn-fcfs} does not hold, i.e., FCFS with no
information $(q^{\ast },I^{\ast }$) does not implement the optimal outcome $%
(x^{\ast },p^{\ast })$ where $p^{\ast }$ solves $[P^{\prime }]$. More
specifically, we provide an example where the \cref{IC} constraint is violated
at $t>0$ under FCFS with no information when trying to implement the optimal
outcome.

We assume that $\mu _{k}=\mu =1$ for all $k\geq 1$ so that the service
process is regular. The arrival process $(\lambda _{k})_{k}$ is yet to be
specified but we assume that $V$, $C$ and $\alpha $ will be chosen in such a
way that $K^{\ast }=4$, \Cref{IR} is binding and no rationing occurs at $k=3$.

We already know (see Equation \Cref{Eq: cond beliefs}) that for each $\ell =1,...,4:$%
\begin{equation*}
\tilde{\gamma}_{\ell }^{t+dt}=\frac{(1-\mu _{\ell }dt)\tilde{\gamma}_{\ell
}^{t}+\mu _{\ell }dt\tilde{\gamma}_{\ell +1}^{t}}{\sum_{i=1}^{K^{\ast }}%
\tilde{\gamma}_{i}^{t}(1-{q}_{i}^{\ast }dt)}+o(dt)\text{.}
\end{equation*}%
Thus, forgetting about the lower order terms, the evolution of beliefs is
given by%
\begin{eqnarray*}
\dot{\gamma}_{\ell }^{t} &=&\frac{\tilde{\gamma}_{\ell }^{t+dt}-\tilde{\gamma%
}_{\ell }^{t}}{dt} \\
&=&\frac{1}{dt}\frac{(1-\mu dt)\tilde{\gamma}_{\ell }^{t}+\mu dt\tilde{\gamma%
}_{\ell +1}^{t}-\tilde{\gamma}_{\ell }^{t}\left[ \tilde{\gamma}%
_{1}^{t}(1-\mu dt)+\sum_{i=2}^{K^{\ast }}\tilde{\gamma}_{i}^{t}\right] }{%
\tilde{\gamma}_{1}^{t}(1-\mu dt)+\sum_{i=2}^{K^{\ast }}\tilde{\gamma}_{i}^{t}%
} \\
&=&\frac{1}{dt}\frac{\tilde{\gamma}_{\ell }^{t}-(\mu dt)\tilde{\gamma}_{\ell
}^{t}+\mu dt\tilde{\gamma}_{\ell +1}^{t}-\tilde{\gamma}_{\ell }^{t}\left[ 1-%
\tilde{\gamma}_{1}^{t}(\mu dt)\right] }{\tilde{\gamma}_{1}^{t}(1-\mu
dt)+\sum_{i=2}^{K^{\ast }}\tilde{\gamma}_{i}^{t}} \\
&=&\frac{1}{dt}\frac{-(\mu dt)\tilde{\gamma}_{\ell }^{t}+\mu dt\tilde{\gamma}%
_{\ell +1}^{t}+(\tilde{\gamma}_{\ell }^{t}\tilde{\gamma}_{1}^{t})(\mu dt)}{%
1-(\mu dt)\tilde{\gamma}_{1}^{t}} \\
&=&\frac{-(\mu )\tilde{\gamma}_{\ell }^{t}+\mu \tilde{\gamma}_{\ell +1}^{t}+(%
\tilde{\gamma}_{\ell }^{t}\tilde{\gamma}_{1}^{t})(\mu )}{1-(\mu dt)\tilde{%
\gamma}_{1}^{t}} \\
&\rightarrow &-\tilde{\gamma}_{\ell }^{t}+\tilde{\gamma}_{\ell +1}^{t}+%
\tilde{\gamma}_{\ell }^{t}\tilde{\gamma}_{1}^{t}\text{ as }dt\rightarrow 0
\end{eqnarray*}%
where the last equality uses our assumption that $\mu =1$.

We now focus on how the expected waiting time
\begin{equation*}
\sum_{\ell =1}^{K^{\ast }}\tilde{\gamma}_{\ell }^{t}\tau _{\ell }^{\ast }
\end{equation*}%
evolves at $t=0$. Given that we are using FCFS, we know that $\tau _{\ell
}^{\ast }=\frac{\ell }{\mu }=\ell $ (recall \Cref{lem:waiting-time-FCFS}), the differential
w.r.t. time at $t=0$ of the expected waiting time must be
\begin{eqnarray*}
\dot{\gamma}_{1}^{0}+2\dot{\gamma}_{2}^{0}+3\dot{\gamma}_{3}^{0}+4\dot{\gamma%
}_{4}^{0} &=&-\tilde{\gamma}_{1}^{0}+\tilde{\gamma}_{2}^{0}+\tilde{\gamma}%
_{1}^{0}\tilde{\gamma}_{1}^{0} \\
&&-2\tilde{\gamma}_{2}^{0}+2\tilde{\gamma}_{3}^{0}+2\tilde{\gamma}_{2}^{0}%
\tilde{\gamma}_{1}^{0} \\
&&-3\tilde{\gamma}_{3}^{0}+3\tilde{\gamma}_{4}^{0}+3\tilde{\gamma}_{3}^{0}%
\tilde{\gamma}_{1}^{0} \\
&&-4\tilde{\gamma}_{4}^{0}+4\tilde{\gamma}_{4}^{0}\tilde{\gamma}_{1}^{0} \\
&=&-1+\tilde{\gamma}_{1}^{0}\left[ \tilde{\gamma}_{1}^{0}+2\tilde{\gamma}%
_{2}^{0}+3\tilde{\gamma}_{3}^{0}+4\tilde{\gamma}_{4}^{0}\right] \text{.}
\end{eqnarray*}%
Hence, we want to find $(\lambda _{k})_{k=0,...,3}$ such that
\begin{equation}
\tilde{\gamma}_{1}^{0}\left[ \tilde{\gamma}_{1}^{0}+2\tilde{\gamma}_{2}^{0}+3%
\tilde{\gamma}_{3}^{0}+4\tilde{\gamma}_{4}^{0}\right] >1.
\label{increasing WT}
\end{equation}%
Recall that the invariant distribution $p^{\ast }$ is given by
\begin{eqnarray*}
p_{0}^{\ast } &=&\frac{1}{1+\lambda _{0}+\lambda _{0}\lambda _{1}+\lambda
_{0}\lambda _{1}\lambda _{2}+\lambda _{0}\lambda _{1}\lambda _{2}\lambda _{3}%
} \\
p_{1}^{\ast } &=&\lambda _{0}p_{0}^{\ast } \\
p_{2}^{\ast } &=&\lambda _{1}p_{1}^{\ast } \\
p_{3}^{\ast } &=&\lambda _{2}p_{2}^{\ast } \\
p_{4}^{\ast } &=&\lambda _{3}p_{3}^{\ast }.
\end{eqnarray*}%
Further, the beliefs at $t=0$ are
\begin{equation*}
\tilde{\gamma}_{\ell }^{0}=\frac{p_{\ell }^{\ast }}{1-p_{0}^{\ast }}
\end{equation*}%
for each $\ell =1,...,4$.

Now, consider $\lambda _{0}=1$ and a sequence of $(\lambda _{1},\lambda
_{2},\lambda _{3})$ such that $\lambda _{1}=\lambda _{2}\rightarrow 0$ and $%
\lambda _{1}\lambda _{2}\lambda _{3}\rightarrow \frac{1}{2}$. Note that this
implies that $\lambda _{3}\rightarrow \infty $ and so $\lambda _{3}-\lambda
_{2}$ becomes strictly greater than $0=\mu _{3}-\mu _{2}$, and so the
primitive process becomes non-regular.\ It is clear that the invariant
distribution converges to a distribution putting $\frac{2}{5}$ on $0$, $%
\frac{2}{5}$ on $1$ and the remaining probability of $\frac{1}{5}$ on $4$.
Hence, the beliefs on the queue length at date $0$ converge to%
\begin{equation*}
\tilde{\gamma}_{1}^{0}=\frac{2}{3},\tilde{\gamma}_{2}^{0}=\tilde{\gamma}%
_{3}^{0}=0\text{ and }\tilde{\gamma}_{4}^{0}=\frac{1}{3}\text{.}
\end{equation*}%
Hence, Equation \Cref{increasing WT} must be satisfied since%
\begin{equation*}
\tilde{\gamma}_{1}^{0}\left[ \tilde{\gamma}_{1}^{0}+2\tilde{\gamma}_{2}^{0}+3%
\tilde{\gamma}_{3}^{0}+4\tilde{\gamma}_{4}^{0}\right] =\frac{2}{3}\left[
\frac{2}{3}+\frac{4}{3}\right] =\frac{4}{3}>1.
\end{equation*}%
We obtain that $\dot{\gamma}_{1}^{0}+2\dot{\gamma}_{2}^{0}+3\dot{\gamma}%
_{3}^{0}+4\dot{\gamma}_{4}^{0}$ must be strictly positive in the limit.
Thus, we found $(\lambda _{k})_{k=0,...,3}$\ under which, at date $t=0$, the
expected residual waiting time increases. Since it is binding at $t=0$,
\Cref{IR} becomes violated at $t>0$ small.

\section{Generalization of \protect\cite{naor1969regulation}}

\label{sec:Naor generalization}

In this section, we generalize \cite{naor1969regulation}'s classic result
(obtained for the $M/M/1$ queue) to our more general Markov process:\textit{%
agents would have excess incentives to queue under FCFS with full
information.} Since the designer can simply stop excessive queueing, this
means that FCFS with a full information rule, denoted $I^{FI}$, can be used
to achieve the optimal cutoff policy.

\begin{proposition}
\label{prop: FI} Suppose $\alpha =1$ and $\mu $ is regular. Then, FCFS with
full information, $I^{FI}$, implements the optimal cutoff outcome $(x^{\ast
},y^{\ast },p^{\ast })$. 
\end{proposition}

\begin{proof}
Consider FCFS with full information. We need to show that $(IC_{t})$ holds
for all $t\geq 0$. With the full information rule, we only need to show that
$(IC_{0})$ holds. By \cref{lem:waiting-time-FCFS}, condition $(IC_{0})$ can
be written as:
\begin{equation}
V-C\frac{k}{\mu _{k}}\geq 0\Longleftrightarrow \mu _{k}V-Ck\geq 0
\label{FI IC}
\end{equation}%
for all $k\leq K^{\ast }$. In the sequel, we let $K^{FI}$ be the largest
integer satisfying \cref{FI IC}. We know that, by regularity of $\mu $, $%
k\mapsto \mu _{k}V-Ck$ is single-peaked (by \cref{lem: single-peakedness}
for $\alpha =1$ and $\xi =0$). Hence, $K^{FI}$ is well-defined (i.e.,
finite) given our assumption that $\mu _{k}$ is uniformly bounded. In
addition, Equation \cref{FI IC} holds at state $k$ if and only if $k\leq
K^{FI}$. Hence, it is enough for our purpose to show that $K^{\ast }\leq
K^{FI}$.

Proceed by contradiction and assume that the optimal cutoff policy $p^{\ast
} $, which we recall solves $[P^{\prime }]$, puts strictly positive weight
on $k>K^{FI}$. Note that, using again the fact that $k\mapsto \mu _{k}V-Ck$
is single-peaked, for any such $k$, $\mu _{k}V-Ck<0$. Now, build $p^{\prime
} $ such that $p_{k}^{\prime }=0$ for all $k>K^{FI}$ and $p_{k}^{\prime
}=Zp_{k}^{\ast }$ for all $k\leq K^{FI}$ where $Z>1$ is set so that the sum
of $p_{k}^{\prime }$ is equal to $1$. Given that $p^{\ast }$ satisfies $%
\cref{B'}$ and given that, by construction, $p_{k}^{\prime }/p_{k-1}^{\prime
}=p_{k}^{\ast }/p_{k-1}^{\ast }$ for all $k\leq K^{FI}$, we must have that $%
p^{\prime }$ also satisfies \cref{B'}. Compared to $p^{\ast }$, distribution
$p^{\prime }$ removes all weight on negative values and, for each positive
value, increases its weight. This must strictly increase the value of the
objective. It remains to show that $p^{\prime }$ satisfies \cref{IR}. The
value of the objective must be positive under $p^{\ast }$ (recall that the
dirac mass on $0$ brings a value of the objective of $0$), and so the value
of the objective must be positive under $p^{\prime }$ as well. Given that $%
\alpha =1$, this implies that \cref{IR} is satisfied.
\end{proof}

\section{General analysis: beyond stationarity assumption} \label{app:beyond-stationary}


In the text, we analyzed queueing mechanisms under the assumption that the
Markov chain has attained the stationary distribution. While this provides a
clean analysis, one may wonder if the results remain valid and robust in the
long run when the Markov chain begins with length $k=0$ in the beginning.

For instance, it is not obvious that such a Markov chain induced by some
queueing policy will admit a unique invariant distribution such that the
time average of states will converge to the distribution. More importantly,
one may wonder if the optimal mechanism we identify in the main text will be optimal
in the long-run limit average sense. The analysis here will answer these
questions in the affirmative.



\subsection{Long-run average formulation}

Fix any entry/exit rule $(x,y,z)$. Recall that such a rule induces a birth
and death process. Let $A(t)$ and $D(t)$ be the number of agents who arrived
at the queue and departed the queue by time $t$ and $p_{(x,y,z)}^{T}\in
\Delta (\mathbb{Z}_{+})$---or simply, $p^{T}$---be the the expected average
frequency of states until date $T$. More formally, for each $k\in \mathbb{Z}%
_{+},$ we define
\begin{equation*}
p_{k}^{T}:=\frac{\mathbb{E}\left[ \int_{0}^{T}1_{\{A(t)-D(t)=k\}}dt\right] }{%
T}\text{.}
\end{equation*}%
Let $S(t)$ be the number of agents who have been served by time $t$.

When we
do not assume the existence of a stationary distribution, the problem can be
written as:
\begin{equation*}
\max_{(x,y,z,q,I)}\lim \inf_{T\rightarrow \infty }(1-\alpha )\frac{\mathbb{E}%
\left[ S(T)\right] R}{T}+\alpha \frac{\mathbb{E}\left[ S(T)\right] V-C%
\mathbb{E}\left[ \int_{0}^{T}(A(t)-D(t))dt\right] }{T}\leqno{[\bar{P}]}
\end{equation*}%
subject to
\begin{equation*}
\lim_{T\rightarrow \infty }\inf \sum_{k,\ell }\gamma _{k,\ell }^{t,T}(V\cdot
\sigma _{k,\ell }-C\cdot \tau _{k,\ell })\geq 0,\forall \gamma ^{t,T}\in %
\supp(I^{t}),\forall t\geq 0; \tag{$\overline{IC}$} \label{icstar}
\end{equation*}%
where $\gamma _{k,\ell }^{t,T}$ corresponds to the agents' beliefs on $%
(k,\ell )$ given that he joined the queue at date $T$ (hence, his
unconditional beliefs are given by $p^{T}$) and that he has spent an amount
of time of $t$ in the queue.
Our interpretation here is that $\gamma^{t,T}$ is the belief held by an entrant who only knows her arrival time is uniform within $[0,T]$.

The program $[\bar{P}]$ maximizes the
expectation of the long-run time average of the weighted sum of the agents'
and the service provider's payoffs as $T\rightarrow \infty $. Further, since
$p^{T}$ need not converge, we take the worst-case evaluation where the
objective is evaluated taking the $\lim \inf $. A similar requirement is
imposed for the incentive constraint. Finally, we assume that $(x,y,z,q,I)$
may depend on $t$ (the time spent in the queue---as in the main text) but do not
depend on the calendar time $T$.\footnote{Note that $I$ is a function of calendar time $T$ since it determines for each time $t$ spent in the queue, a collection of posterior beliefs which depend on prior belief $p^T$. We omit this dependence in our notations.}

Recall our relaxed program $[P^{\prime }]$ (assuming stationarity) and its optimal solution $p^{\ast }$.  Recall also the cutoff policy $x^{\ast }$ that implements $p^{\ast }$.\footnote{%
Recall our convention to ignore $y$ and $z$ in our notations since these are
all equal to zero.}  We now claim that our main result  \Cref{thm:dyn-fcfs} remains valid even when the Markov chain begins with $k=0$.  Specifically, our main result is as follows:

\begin{theorem}
\label{thm:main_no_inv} Assume that the primitive process is regular and $%
\alpha >0$.\footnote{We assume here that $\alpha >0$ which implies that the optimal queue length $K^*$ is finite. As will become clear \Cref{thm:main_no_inv}  holds under the (more general) assumption that $K^*$ is finite.} Denoting FCFS with no information $(q^{\ast },I^{\ast })$ we
have that $(x^{\ast },q^{\ast },I^{\ast })$ is an optimal solution of $[\bar{%
P}]$.
\end{theorem}

We prove
\Cref{thm:main_no_inv} in the remaining subsections.  Specifically, \Cref{app:convergence} proves that a solution to $[\bar{P}]$ admits a long-run time average $(p^T)$ that weakly converges to a unique invariant measure $\bar p$.   \Cref{app:relaxation} then shows that $[P^{\prime }]$ is a relaxation of $[\bar{P}]$.  Finally,  \Cref{app:optimality} proves that the FCFS with no information $(q^{\ast },I^{\ast })$ together with the cutoff policy $x^*$ satisfies the constraint of $[\bar{P}]$, thus concluding the proof of \Cref{thm:main_no_inv}.

\subsection{The convergence of $(p^T)$ satisfying \Cref{icstar}.} \label{app:convergence}

The fundamental proposition for us is the following.

\begin{proposition}
\label{prop:fund} Assume  $\sup_{k}\lambda _{k}<\infty $. Fix
any solution $(x,y,z,q,I)$ to $[\bar{P}]$. We must have that $p^{T}$
converges (in weak convergence of measures) to some $\bar{p}$. Further, $%
\bar{p}$ is the unique invariant distribution of the birth-death process
induced by $(x,y,z)$.
\end{proposition}

\begin{proof}
Fix any solution $(x,y,z,q,I)$ to $[\bar{P}]$. If $x_0=0$ then the result holds trivially. So in the sequel, we assume that $x_0>0$.

\vskip 1cm

\noindent  {\textbf{Step 1.} {\it The birth-death process induced by $%
(x,y,z)$ has an invariant distribution.}}

\bigskip

In order to show this, we first prove the following lemma.

\begin{lemma}
\label{lem:Exp(S)} We have
\begin{equation*}
\frac{\mathbb{E}\left[ S(T)\right] }{T}=\sum_{k=1}^{\infty }p_{k}^{T}\mu _{k}%
\text{.}
\end{equation*}
\end{lemma}

\begin{proof} Let $k^{t}$ be the length of queue at time $t$ and let $%
Z(t):=\mathbb{E}[S(t)]$. Then,
\begin{align*}
Z({t+dt})& =\mathbb{E}\left[ S^{t}+1_{\{\mbox{an agent is served during }%
\lbrack t,t+dt)\}}+o(dt)\right] \\
& =Z(t)+\sum_{k}\Pr \{k^{t}=k\}\mathbb{E}\left[ 1_{\{%
\mbox{an agent is
served during }\lbrack t,t+dt)\}}\bigg|k^{t}=k\right] +o(dt) \\
& =Z(t)+\sum_{k}\Pr \{k^{t}=k\}\mu _{k}dt+o(dt), \\
&
\end{align*}%
where the last equality follows from the independence of $\{k^{t}=k\}$ and $%
\{\mbox{an agent is served during }\lbrack t,t+dt)\}$. This means that
\begin{equation*}
Z^{\prime }(t)=\sum_{k}\Pr \{k^{t}=k\}\mu _{k}.
\end{equation*}%
Hence, by the fundamental theorem of calculus,
\begin{align*}
\frac{Z(T)}{T}& =\frac{Z(0)}{T}+\frac{\int_{0}^{T}\sum_{k}\Pr \{k^{t}=k\}\mu
_{k}dt}{T} \\
& =\sum_{k}\mu _{k}\frac{\int_{0}^{T}\Pr \{k^{t}=k\}dt}{T} \\
& =\sum_{k}\mu _{k}\frac{\mathbb{E}\left[ \int_{0}^{T}1_{\{k^{t}=k\}}dt%
\right] }{T} \\
& =\sum_{k}\mu _{k}p_{k}^{T}.
\end{align*}%
The second equality holds from $Z(0)=0$.
\end{proof}

Using this lemma, we prove the following additional result.

\begin{lemma}
There is $\xi >0$ and $\bar{T}<\infty $ such that
\begin{equation*}
\sum_{k=1}^{\infty }p_{k}^{T}(\mu _{k}V-kC)\geq -\xi
\end{equation*}%
for any $T\geq \bar{T}$.
\end{lemma}

\begin{proof}   Let
$$U^T:=\int\sum_{k}\gamma _{k,k}^{0,T}(V\cdot \sigma
_{k,k}-C\cdot \tau _{k,k}) I^0(d\gamma)$$ be the ex-ante expected average payoff for an agent who enters the queue between time $0$ and $T$, obtained by aggregating \Cref{icstar} for $t=0,$ across all $\gamma_{k,k}^{0,T}\in \supp(I^0)$.

Since each agent who enters the queue enjoys $U^T$ on average, the total ex-ante expected surplus accruing to all agents entering the queue by time $T$ is:
\begin{align*}
 & \,  \sum_{n=1}^{\infty} \Pr\{A(T)=n\} n U^T
 =
\E[ A(T)] U^T \\
\le & \,  \mathbb{E}\left[ S(T)\right] V-C\mathbb{E}\left[
 \int_{0}^{T}(A(t)-D(t))dt\right]+ \mathbb{E}[(A(T)-D(T))V \\
 =  & \, \E[ A(T)] \frac{T}{\E[ A(T)] }  \frac{\mathbb{E}\left[ S(T)\right] V-C\mathbb{E}\left[
 \int_{0}^{T}(A(t)-D(t))dt\right]+ \mathbb{E}[(A(T)-D(T))V].}{T}\\
  \le & \, \E[ A(T)]  \frac{1}{\sum_k p_k^T  \lambda_k x_k} \left(\sum_{k}p_{k}^{T}(V\mu _{k}-kC)+\sup_{k}\lambda _{k}V\right).
\end{align*}
The first inequality is explained as follows. The first two terms in the second line account for the expected total surplus realized up to time $T$, excluding the surplus that will eventually accrue to those who arrived by $T$ but are not fully served by $T$.  The last term adds the gross surplus for these agents assuming that they are all eventually served but excludes the waiting costs they will incur after $T$.  Hence, the second line gives the upper bound of the expected total surplus for agents who enter the queue by time $T$.    The second inequality comes from the \Cref{lem:Exp(S)} as well as the
observation that, by definition of a Poisson process $\mathbb{E}[A(T)]/T$
is equal to $\sum_{k}p_{k}^{T}\lambda _{k} x_k$.\footnote{%
To prove this, one can use a similar argument as in \Cref{lem:Exp(S)}.}
Invoking \Cref{icstar} for $t=0$, we have
\begin{equation*}
0\le    \lim \inf_{T\rightarrow \infty} U^T.
\end{equation*}
Hence, from the above inequalities we have
\begin{equation*}
0\le \lim \inf_{T\rightarrow \infty }\sum_{k}p_{k}^{T}(V\mu _{k}-kC)+\sup_{k}\lambda _{k}V,
\end{equation*}
so
\begin{equation*}
\lim \inf_{T\rightarrow \infty }\sum_{k}p_{k}^{T}(V\mu _{k}-kC)\geq
-\sup_{k}\lambda _{k}V.
\end{equation*}%
Since we assumed that $\sup_{k}\lambda _{k}<\infty $, this yields the
desired result. \end{proof}

Let us now fix such $\xi $ and $\bar{T}$ throughout. We can now show that
the collection $\{p^{T}\}_{T\geq \bar{T}}$ is tight which implies that $%
\{p^{T}\}_{T}$ is tight given that $\bar{T}<\infty $. We need to show that
for any $\varepsilon >0$, there is $K$ large enough so that any probability
measure $p\in \{p^{T}\}_{T\geq \bar{T}}\ $has $\sum_{k=K+1}^{\infty
}p_{k}<\varepsilon $. Proceed by contradiction and assume that there is $%
\varepsilon >0$ and a sequence $\{p^{K}\}$ in $\{p^{T}\}_{T\geq \bar{T}}$
(which, by definition, satisfies $\sum_{k=0}^{\infty }p_{k}^{K}\left[ \mu
_{k}V-Ck\right] \geq -\xi $) such that $\sum_{k=K+1}^{\infty
}p_{k}^{K}>\varepsilon $ for all $K$. This implies
\begin{eqnarray*}
\sum_{k=0}^{\infty }p_{k}^{K}\left[ \mu _{k}V-Ck\right] &=&V\sum_{k=0}^{%
\infty }p_{k}^{K}\mu _{k}-C\sum_{k=0}^{\infty }p_{k}^{K}k \\
&\leq &V\sup_{k}\mu _{k}-C\sum_{k=K+1}^{\infty }p_{k}^{K}k \\
&\leq &V\sup_{k}\mu _{k}-C(K+1)\sum_{k=K+1}^{\infty }p_{k}^{K} \\
&\leq &V\sup_{k}\mu _{k}-C(K+1)\varepsilon \text{.}
\end{eqnarray*}%
Note that for $K$ large enough, the above term must be strictly below $-\xi $
(recall that, by assumption, $\sup_{k}\mu _{k}<\infty $).\ This contradicts
the fact that $\sum_{k=0}^{\infty }p_{k}^{K}\left[ \mu _{k}V-Ck\right] \geq
-\xi $ for all $K$. We conclude that collection $\{p^{T}\}_{T\geq \bar{T}}$
is tight.

Finally, let $\tilde{p}^{T}$ be defined by
\begin{equation*}
\tilde{p}_{k}^{T}:=\Pr \left[ k^{T}=k\left\vert k^{0}=0\right. \right]
\end{equation*}%
where $k^{T}$ is queue length at time $T$. We show that tightness of $%
\{p^{T}\}_{T\geq \bar{T}}$ implies the tightess of $\{\tilde{p}^{T}\}_{T\geq
\bar{T}}$.

\begin{lemma}
If $\{p^{T}\}_{T}$ is tight then $\{\tilde{p}^{T}\}_{T}$ is tight.
\end{lemma}

\begin{proof}
Let us start with the simple following relation between $p^{T}$ and $\{%
\tilde{p}^{t}\}_{t}:$%
\begin{eqnarray*}
p_{k}^{T} &=&\frac{1}{T}\mathbb{E}\left[ \int_{0}^{T}1\{k_{t}=k\}dt\right] \\
&=&\frac{1}{T}\int_{0}^{T}\Pr \{k_{t}=k\}dt \\
&=&\frac{1}{T}\int_{0}^{T}\tilde{p}_{k}^{t}dt\text{.}
\end{eqnarray*}%
Now, by way of contradiction, assume that $\{p^{T}\}_{T}$ is tight but $\{%
\tilde{p}^{t}\}_{t}$ is not tight. The latter means that there is $%
\varepsilon >0$ such that for any integer $K$, there is $\tilde{p}%
^{t_{K}}\in \{\tilde{p}^{t}\}_{t}$ satisfying
\begin{equation*}
\sum_{k=K+1}^{\infty }\tilde{p}_{k}^{t_{K}}>\varepsilon \text{.}
\end{equation*}%
Since $\tilde{p}^{t}$ increases in $t$ in the stochastic dominance order\footnote{\label{fn:SD} Indeed, $\tilde{p}^{t}$ is increasing in the stochastic dominance order given that
the process is birth-death and starts at state $0$ (see \cite%
{KeilsonKester77} and \cite{VanDoorne80}).}, we obtain that for any $t\geq t_{K}:$
\begin{equation*}
\sum_{k=K+1}^{\infty }\tilde{p}_{k}^{t}>\varepsilon \text{.}
\end{equation*}%
Now, we show that $\{p^{T}\}_{T}$ cannot be tight, a contradiction. To see
this, let $\varepsilon ^{\prime }:=\varepsilon /2$. Consider any positive
integer $K$. We will find $T$ large enough under which $\sum_{k=K+1}^{\infty
}p_{k}^{T}>\varepsilon ^{\prime }$. Take $T\geq t_{K}$, we have%
\begin{eqnarray*}
\sum_{k=K+1}^{\infty }p_{k}^{T} &=&\sum_{k=K+1}^{\infty }\frac{1}{T}%
\int_{0}^{T}\tilde{p}_{k}^{t}dt \\
&=&\frac{1}{T}\int_{0}^{T}\sum_{k=K+1}^{\infty }\tilde{p}_{k}^{t}dt \\
&\geq &\frac{1}{T}\int_{t_{K}}^{T}\sum_{k=K+1}^{\infty }\tilde{p}_{k}^{t}dt
\\
&>&\frac{1}{T}\int_{t_{K}}^{T}\varepsilon dt=\varepsilon \frac{T-t_{K}}{T}%
\text{.}
\end{eqnarray*}%
Note that the last term is greater than $\varepsilon ^{\prime }=\varepsilon
/2$ when $T$ is large enough. This concludes the proof. \end{proof}

We conclude from the above lemma that $\{\tilde{p}^{T}\}_{T}$ is tight. Now,
to complete the proof of Step 1, we make use of the following result, due to Krylov-Bogolioubov's Theorem for Markov chains with countable
state spaces.

\begin{theorem}[Krylov-Bogolioubov]
Consider a time-homogenous Markov chain $\{X_{t}\}_{t}$ on $\mathbb{Z}_{+}$
and let $\tilde{p}^{T}$ be defined by
\begin{equation*}
\tilde{p}^{T}(k,A):=\Pr \left\{ X_{T}\in A\left\vert X_{0}=k\right. \right\}
\end{equation*}%
for all sets $A\subset \mathbb{Z}_{+}$. If for some $k$, $\{\tilde{p}%
^{T}(k,\cdot )\}_{T>0}$ is tight then the Markov chain has at least one
invariant distribution.\footnote{%
The general version of the theorem requires that $(p^{T})$ satisfies the
Feller property, i.e., for any $T\geq 0$ and any bounded and continuous
function $g:\mathbb{Z}_{+}\rightarrow \mathbb{R}$, the function $%
\sum_{k^{\prime }=1}^{\infty }\tilde{p}^{T}(k,k^{\prime })g(k^{\prime })$
must be continuous in $k$. Recall that $\mathbb{Z}%
_{+}$ is endowed with the discrete topology and so this requirement is
trivially satisfied. Further, the theorem requires that $\mathbb{Z}_{+}$ is
a Polish space which holds true again under the discrete topology.}
\end{theorem}

Since we just showed $\{\tilde{p}^{T}(0,\cdot )\}_{T>0}$ is tight, we can
apply the above theorem.  We thus conclude that  the birth-death process induced by $(x,y,z)$ has an invariant distribution.

\vskip 0.3cm

\noindent {\textbf{Step 2.} {\it $p^{T}$ converges to $\bar{p}$, the
unique invariant distribution of the birth-death process.}}

\vskip 0.2cm
Let $\bar{p}$ be
the invariant distribution of the birth-death process induced by $(x,y,z)$
which exists by the result of Step 1.
First, notice that for our birth-death process, this invariant distribution must be unique (there is at most one solution to the balance condition (B)). In addition, our birth-death process is irreducible and positive recurrent.\footnote{For birth-death processes, positive recurrence is implied by the existence of an invariant distribution.}  Hence, the Ergodic Theorem for continuous-time Markov processes applies and our result follows.
\end{proof}

This implies that, without loss of generality, we can add a constraint to
problem $[\bar{P}]$ which guarantees that $p^{T}$ converges to a $\bar{p}$
satisfying the balanced condition $(B)$. In the sequel, we will assume that
this constraint is added to our problem.

\subsection{$[P^{\prime }]$ is a relaxation of $[\bar{P}]$} \label{app:relaxation}

With \Cref{prop:fund} in hand, assuming $\sup \lambda _{k}<\infty $, we will
show that the $[P^{\prime }]$ (as defined in the paper) is a relaxation of $[%
\bar{P}]$. Recall that, by \cref{prop:fund}, $[\bar{P}]$ writes as
\begin{equation*}
\max_{(x,y,z,q,I)}\liminf_{T\rightarrow \infty }(1-\alpha )\frac{\mathbb{E}%
\left[ S(T)\right] R}{T}+\alpha \frac{\mathbb{E}\left[ S(T)\right] V-C%
\mathbb{E}\left[ \int_{0}^{T}(A(t)-D(t))dt\right] }{T}\leqno{[\bar{P}]}
\end{equation*}%
subject to
\begin{equation*}
\liminf_{T\rightarrow \infty }  \sum_{k,\ell }\gamma _{k,\ell }^{t,T}(V\cdot
\sigma _{k,\ell }-C\cdot \tau _{k,\ell })\geq 0,\forall \gamma ^{t,T}\in %
\supp(I^{t}),\forall t\geq 0;
\end{equation*}%
and
\begin{equation*}
\lim p^{T}=\bar{p}
\end{equation*}%
and%
\begin{equation*}
\lambda _{k}x_{k}(1-\sum_{\ell }z_{k,\ell })\bar{p}_{k}=(\mu
_{k+1}+\sum_{\ell }y_{k+1,\ell })\bar{p}_{k+1},\,\forall k\in \supp(p)\text{
}(B)\text{.}
\end{equation*}%

Now, fix any solution $(x,y,z,q,I)$ of program $[\bar{P}]$%
. We show that the invariant distribution $\bar{p}$ of the stochastic
process induced by $(x,y,z)$ satisfies \cref{IR}. Indeed, the \cref{IC}
constraint at $t=0$ under the no information policy must be satisfied by $%
(x,y,z,q,I)$, i.e.,

\begin{equation*}
\liminf_{T\rightarrow \infty }\sum_{k}\gamma _{k,k}^{0,T}(V\cdot \sigma
_{k,k}-C\cdot \tau _{k,k})\geq 0
\end{equation*}%
where
\begin{equation*}
\gamma _{k,k}^{0,T}=\frac{p_{k-1}^{T}\lambda _{k-1}x_{k-1}}{\sum_{j\geq
1}p_{j-1}^{T}\lambda _{j-1}x_{j-1}}\text{.}
\end{equation*}%

Since $p^{T}$ converges to $\bar{p}$, the unique invariant distribution of
our birth-death process, we have that, for all $k$,
\begin{equation*}
\lim_{T\to \infty} \gamma _{k}^{0,T}=\gamma _{k}^{0},
\end{equation*}%
which implies that $\gamma ^{0,T}$ converges to $\gamma ^{0}$ (as defined in
the paper) in weak convergence of measures.\footnote{%
Note that the denominator of $\gamma _{k}^{0,T}$ converges by definition of
weak convergence of measures  together with our assumption that $%
\inf_{k}\lambda _{k}<\infty $ which implies that the mapping $j\mapsto
\lambda _{j}x_{j-1}$  is bounded (recall $\mathbb{Z}$ is endowed with the
discrete topology so continuity holds trivially.} By
Portmanteau's Theorem, together with the fact that function $k\mapsto V\cdot \sigma
_{k,k}-C\cdot \tau _{k,k}$ is upper bounded, we have that
\begin{eqnarray}
0&\le& \liminf_{T\rightarrow \infty }  \sum_{k}\gamma _{k}^{0,T}(V\cdot \sigma
_{k,k}-C\cdot \tau _{k,k})\\
&\leq &\limsup_{T\rightarrow \infty }
\sum_{k}\gamma _{k}^{0,T}(V\cdot \sigma _{k,k}-C\cdot \tau _{k,k})
\label{IR_weakening} \\
&\leq &\sum_{k}\gamma _{k}^{0}(V\cdot \sigma _{k,k}-C\cdot \tau _{k,k})
\notag \\
&=&\frac{1}{\sum_{i}\bar p_{i} \mu _{i}}\sum_{k=1}^{\infty }\bar{p}%
_{k}(\mu _{k}V-kC)  \notag
\end{eqnarray}%
where the equality is proved in \Cref{lem:general} of \Cref{app-sec: general}. Thus, \cref{IR} is satisfied by
$\bar{p}$, as claimed. In addition, since $(x,y,z,q,I)$ satisfies $(B)$, $%
\bar{p}$\ must satisfy $(B^{\prime })$.

Finally, the value of the objective of $[\bar{P}]$ at $(x,y,z,q,I)$ must be
equal to that of $[P^{\prime }]$ under $\bar{p}$. To see this, observe first
that, by \Cref{lem:Exp(S)},
\begin{eqnarray*}
&&\liminf_{T\rightarrow \infty }(1-\alpha )\frac{\mathbb{E}\left[ S(T)\right]
R}{T}+\alpha \frac{\mathbb{E}\left[ S(T)\right] V-C\mathbb{E}\left[
\int_{0}^{T}(A(t)-D(t))dt\right] }{T} \\
&=&\lim_{T\rightarrow \infty }\inf (1-\alpha )R\sum_{k=1}^{\infty
}p_{k}^{T}\mu _{k}+\alpha \sum_{k=1}^{\infty }p_{k}^{T}(\mu _{k}V-kC).
\end{eqnarray*}

Now, let us prove that
\begin{equation}
\lim_{T\rightarrow \infty }\sum_{k=1}^{\infty }p_{k}^{T}k=\sum_{k=1}^{\infty
}\bar{p}_{k}k.  \label{equality}
\end{equation}
Since function $k\mapsto k$ is lower bounded, by Portmanteau's Theorem we
must have that
\begin{equation}
\lim_{T\rightarrow \infty }\inf \sum_{k=1}^{\infty }p_{k}^{T}k\geq
\sum_{k=1}^{\infty }\bar{p}_{k}k.  \label{lower}
\end{equation}%
Now, since $p^{T}$ increases in the stochastic dominance order when $T$
increases, we must have that for each $T$, $p^{T}$ is stochastically
dominated by $\bar{p}$.\footnote{\label{fn:SD2} Recall (see \Cref{fn:SD}) that $\tilde{p}^{t}$ is increasing in the stochastic dominance order as $t$ increases. This implies that $p^{T}$  increases in the stochastic dominance order as well as $T$ increases.} Thus, we obtain that
\begin{equation*}
\sum_{k=1}^{\infty }p_{k}^{T}k\leq \sum_{k=1}^{\infty }\bar{p}_{k}k
\end{equation*}%
for each $T$. Thus,
\begin{equation}
\lim_{T\rightarrow \infty }\sup \sum_{k=1}^{\infty }p_{k}^{T}k\leq
\lim_{T\rightarrow \infty }\sum_{k=1}^{\infty }\bar{p}_{k}k.  \label{upper}
\end{equation}%
Now, Equations \Cref{lower} and \Cref{upper} imply Equation \Cref{equality}. With this in our hand, and again using Portmanteau's Theorem, we have the
following
\begin{equation*}
\lim_{T\rightarrow \infty }\inf (1-\alpha )R\sum_{k=1}^{\infty }p_{k}^{T}\mu
_{k}+\alpha \sum_{k=1}^{\infty }p_{k}^{T}(\mu _{k}V-kC)=(1-\alpha
)R\sum_{k=1}^{\infty }\bar{p}_{k}\mu _{k}+\alpha \sum_{k=1}^{\infty }\bar{p}%
_{k}(\mu _{k}V-kC)\text{.}
\end{equation*}%
\bigskip We conclude that $[P^{\prime }]$ is indeed a weakening of $[\bar{P}]$.

\subsection{Proof of Theorem \protect\ref{thm:main_no_inv}} \label{app:optimality}

Let us fix the cutoff policy $(x^{\ast },y^{\ast },z^{\ast })$---or simply $%
x^{\ast }$---which induces $p^{\ast }$ the optimal solution of $[P^{\prime
}] $ (which is well-defined by \Cref{thm:cutoff} in the paper). Now, consider $x^{\ast }$ together with FCFS and the no information
policy. We show
that this is feasible in $[\bar{P}]$. First, under $(x^{\ast },y^{\ast
},z^{\ast })$, $p^{T}$ must converge to $p^{\ast }$ again by the Ergodic Theorem for continuous-time Markov processes and $(B)$ must be satisfied by construction of $(x^{\ast },y^{\ast },z^{\ast
})$. We prove the incentive constraint in the remaining part of this
section. In order to do so, we will make use of the previous results. Recall we assume $\sup_{k}\lambda _{k}<\infty $. This condition turns out to be
satisfied under regularity of the primitive process as well as the
assumption we make in the paper that $\sup_{k}\mu _{k}<\infty $.\footnote{%
Indeed, Let $c:=\lambda _{0}-\mu _{0}=\lambda _{0}$. Define $\hat{\lambda}%
_{k}:=\lambda _{k}-c$ for all $k\geq 0$. Since $\hat{\lambda}_{k}-\hat{%
\lambda}_{k-1}=\lambda _{k}-\lambda _{k-1}\leq \mu _{k}-\mu _{k-1}$ and $%
\hat{\lambda}_{0}=\mu _{0}$, one can easily show by induction that $\hat{%
\lambda}_{k}\leq \mu _{k}$. Thus, $\hat{\lambda}_{k}\leq \sup_{k^{\prime
}}\mu _{k^{\prime }}$ for all $k$. Thus, $\lambda _{k}\leq \sup_{k^{\prime
}}\mu _{k^{\prime }}+c$ for all $k$. We obtain $\sup_{k}\lambda _{k}<\infty $%
.}

\subsubsection{ \Cref{icstar} at $t=0$.}

We want to show that $(x^{\ast },y^{\ast },z^{\ast })$ which induces $%
p^{\ast }$, the optimal solution of $[P^{\prime }]$ together with FCFS and
the no-information policy satisfies \Cref{icstar} at $t=0$. If $x^{\ast }_0=0$ then this holds trivially. So in the sequel, we assume that $x^{\ast }_0>0$. Let $(\gamma_{k,\ell}^{0,T})$ be the belief that the policy induces for an agent who enters the queue during $[0,T]$. Letting
$$U^T:= \sum_{k}\gamma_{k,k}^{0,T}(V\cdot \sigma
_{k,k}-C\cdot \tau _{k,k}),$$
we need to prove that
$$U^{\infty}:=\liminf_{T\to \infty} U^T\ge 0.$$

Consider the Markov chain induced by $(x^{\ast },y^{\ast },z^{\ast })$. Again, the associated time average ${p}^T=({p}^T_k)_{k\in \mathbb{Z}_+}$, where ${p}^T_k=\int_0^T1_{\{A(t)-D(t)=k\}}dt/T$, converges weakly to $p^*$ by the Ergodic Theorem for continuous-time Markov chains.

Using an argument as before, we have:
\begin{align*}
  U^T
= & \,  \frac{\mathbb{E}\left[ S(T)\right] V-C\mathbb{E}\left[
 \int_{0}^{T}(A(t)-D(t))dt\right] +\mathbb{E}[(A(T)-D(T))]U^T}{\E[ A(T)]}\\
  =  & \,  \frac{T}{\E[ A(T)] }  \frac{\mathbb{E}\left[ S(T)\right] V-C\mathbb{E}\left[
 \int_{0}^{T}(A(t)-D(t))dt\right]+ \mathbb{E}[(A(T)-D(T))] U^T}{T}\\
   = & \,   \frac{\sum_{k=1}^{\infty}p_{k}^{T}(V\mu _{k}-kC)}{\sum_{k=0}^{\infty} p_k^T  \lambda_k x^*_k}   +\frac{ U^T \sum_{k=1}^{\infty}k p_k^T}{T \sum_{k=0}^{\infty} p_k^T  \lambda_k x^*_k},
\end{align*}
which yields
\begin{align*}
  U^T
= &  \, \frac{\sum_{k=1}^{\infty}p_{k}^{T}(V\mu _{k}-kC)}{\sum_{k=0}^{\infty} p_k^T  \lambda_k x^*_k}  \left /\left [1- \frac{ \sum_{k=1}^{\infty}k p_k^T}{T \sum_{k=0}^{\infty} p_k^T  \lambda_k x^*_k}\right]. \right.
\end{align*}
We first argue that the denominator converges to one as $T\to\infty$.  This is because
\begin{align*}
\frac{ \sum_{k=1}^{\infty}k p_k^T}{T \sum_{k=0}^{\infty} p_k^T  \lambda_k x^*_k}\le & \, \frac{ \sum_{k=1}^{\infty}k p_k^*}{T \sum_{k=0}^{\infty} p_k^T  \lambda_k x^*_k}
\le  \frac{ \sum_{k=1}^{\infty}p_k^* \mu_k V}{C T \sum_{k=0}^{\infty} p_k^T  \lambda_k x^*_k} \to 0 \, \mbox{ as } T\to \infty,
\end{align*}
where the first inequality is from the fact that $p^*$ stochastically dominates $p^T$ (see Footnotes \ref{fn:SD} and \ref{fn:SD2}) and the second inequality is from \Cref{IR}, and the convergence follows from $\lim_{T\to \infty}\sum_{k=0}^{\infty} p_k^T  \lambda_k x^*_k =\sum_{k=0}^{\infty} p_k^*  \lambda_k x^*_0>0 $ (recall that $x^*_0>0$).

We next argue that
$$\liminf_{T\to \infty} \frac{\sum_{k=1}^{\infty}p_{k}^{T}(V\mu _{k}-kC)}{\sum_{k=0}^{\infty} p_k^T  \lambda_k x^*_k} = \frac{\sum_{k=1}^{\infty}p_{k}^{*}(V\mu_{k}-kC)}{\sum_{k=0}^{\infty} p_k^*  \lambda_k x^*_k}.$$
This follows from the fact that
$\sum_{k=1}^{\infty}p_{k}^{T}(V\mu _{k}-kC)$ converges to $\sum_{k=1}^{\infty}p_{k}^{*}(V\mu_{k}-kC)$ as $T\to\infty$.\footnote{Recall from the previous section that the term $C\sum_{k}p_{k}^{T} k$ converges to $C\sum_{k}p_{k}^{*}k$.}

Combining the arguments,
$$U^{\infty}=\liminf_{T\to \infty} U^T=\frac{\sum_{k=1}^{\infty}p_{k}^{*}(V\mu_{k}-kC)}{\sum_{k=0}^{\infty} p_k^*  \lambda_k x^*_k}\ge 0,$$
where the inequality follows from \Cref{IR}.  Hence, we have proven that \Cref{icstar} holds at $t=0$.

\subsubsection{\cref{icstar} at $t\geq 0$.}

We know that \cref{icstar} at $t=0$ holds.\ So we need to show that \cref{icstar}
holds at $t>0$. One can obtain the system of ODEs governing the evolution of
agents' beliefs following the very same argument as in \Cref{online_app:infty} (where it
is assumed that $K^{\ast }=\infty )$:%
\begin{equation*}
\dot{\gamma}_{k}^{t}=-\mu _{k}\gamma _{k}^{t}+\mu _{k}\gamma
_{k+1}^{t}+\gamma _{k}^{t}\left[ \sum_{i=1}^{K^{\ast }}\gamma _{i}^{t}(\mu
_{i}-\mu _{i-1})\right]
\end{equation*}%
for all $k=0,...,K^{\ast }$. Using the very same arguments as in \Cref{online_app:infty} (when $K^{\ast }=\infty $), one can show that the system satisfies the
conditions to apply Gr\"{o}nwall's inequality.
In addition, we can show that $\gamma^{0,T}$ converges to $\gamma ^{0}$ in $\ell^1-$norm. To see this, we need to show that $\gamma ^{0,T}\rightarrow \gamma ^{0}$ implies that $\left\Vert
\gamma ^{0,T}-\gamma ^{0}\right\Vert =\dsum\limits_{k=0}^{\infty }\left\vert
\gamma _{k}^{0,T}-\gamma _{k}^{0}\right\vert \rightarrow 0$. Fix $%
\varepsilon >0$, we show that
\begin{equation*}
\left\Vert \gamma ^{0,T}-\gamma ^{0}\right\Vert <\varepsilon
\end{equation*}%
for $T$ large enough. Note first that there is a finite $K$ large
enough so that $\dsum\limits_{k=K}^{\infty }\gamma _{k}^{0}<\frac{%
\varepsilon }{4}$. By Portmanteau Theorem, $\dsum\limits_{k=K}^{\infty
}\gamma _{k}^{0,T}\rightarrow \dsum\limits_{k=K}^{\infty }\gamma _{k}^{0}$. Hence, for
any $T$ large enough, $\dsum\limits_{k=K}^{\infty }\gamma _{k}^{0,T}<\frac{%
\varepsilon }{4}$. Given this, we have
\begin{eqnarray*}
\left\Vert \gamma ^{0,T}-\gamma ^{0}\right\Vert
&=&\dsum\limits_{k=0}^{\infty }\left\vert \gamma _{k}^{0,T}-\gamma
_{k}^{0}\right\vert  \\
&=&\dsum\limits_{k=0}^{K-1}\left\vert \gamma _{k}^{0,T}-\gamma
_{k}^{0}\right\vert +\dsum\limits_{k=K}^{\infty }\left\vert \gamma
_{k}^{0,T}-\gamma _{k}^{0}\right\vert  \\
&\leq &\dsum\limits_{k=0}^{K-1}\left\vert \gamma _{k}^{0,T}-\gamma
_{k}^{0}\right\vert +\dsum\limits_{k=K}^{\infty }\left\vert \gamma
_{k}^{0,T}\right\vert +\dsum\limits_{k=K}^{\infty }\left\vert \gamma
_{k}^{0}\right\vert  \\
&<&\dsum\limits_{k=0}^{K-1}\left\vert \gamma _{k}^{0,T}-\gamma
_{k}^{0}\right\vert +\frac{\varepsilon }{2}<\varepsilon
\end{eqnarray*}%
for $T$ large enough. The first inequality is by the triangular inequality
while the last inequality holds because for $T$ large enough, $%
\dsum\limits_{k=0}^{K-1}\left\vert \gamma _{k}^{0,T}-\gamma
_{k}^{0}\right\vert <\frac{\varepsilon }{2}$ since $\gamma ^{0,T}\rightarrow
\gamma ^{0}$.

Thus, we can apply Gr\"{o}nwall's inequality to  obtain $\gamma ^{t,T}\rightarrow \gamma
^{t}$. Now, since for $\alpha >0$, $K^{\ast }$ is finite (see next subsection). In
addition, the support of $\gamma ^{t,T}$ must be included in the support of $%
p^{T}$, hence in $\{0,...,.K^{\ast }\}$. Thus, the function $k\longmapsto
V\cdot \sigma _{k,\ell }-C\cdot \tau _{k,\ell }$ is (lower and upper)
bounded over $\{0,...,K^{\ast }\}$ and so by definition of weak convergence
of measures, we obtain
\begin{equation*}
\lim_{T\rightarrow \infty }\sum_{\ell }\gamma _{k,\ell }^{t,T}(V\sigma
_{k,\ell }-C\cdot \tau _{k,\ell })=\sum_{k,\ell }\gamma _{k,\ell
}^{t}(V\cdot \sigma _{k,\ell }-C\cdot \tau _{k,\ell })\geq 0
\end{equation*}%
where the inequality comes from \Cref{thm:dyn-fcfs} in the paper which
states that \cref{IC} holds for all $t\geq 0$ in the problem $[P]$.

\subsection{Proof of finite $K^{\ast }$ when $\protect\alpha >0$}

In this section, we show that the maximum queue length at the optimal policy
is finite whenever $\alpha >0$. In the sequel, given $K\in \mathbb{Z}%
_{+}\cup \{\infty \}$ and $x\in (0,1]$, we consider the probability
distribution over $\mathbb{Z}_{+}$ defined by%
\begin{equation}
p_{k}=p_{0}\dprod\limits_{\ell =1}^{k}\frac{\lambda _{\ell -1}}{\mu _{\ell }}
\label{eq: p_k}
\end{equation}%
for any $k=1,...,K-1\ $and if $K<\infty $%
\begin{equation*}
p_{K}=p_{0}\dprod\limits_{\ell =1}^{K-1}\frac{\lambda _{\ell -1}}{\mu _{\ell
}}\frac{\lambda _{K-1}x}{\mu _{K}}
\end{equation*}%
while $p_{0}$ is defined to ensure that the total mass of $p_{k}$'s is equal
to $1$. Given $K\in \mathbb{Z}_{+}\cup \{\infty \}$ and $x$, whenever
well-defined, such a distribution will be denoted by $p(K,x)$.\footnote{%
We note that for $K=\infty $, $p_{0}$ may not always be well-defined.} By %
\cref{thm:cutoff} we know that, when $\mu $ is regular, there is an optimal
solution $p^{\ast }$ of $[P^{\prime }]$ which can be implemented by a cutoff
policy. Recall that this implies that \cref{B'} binds for all $k=0,...,K-1$
and holds with weak inequality for $k=K-1$ where $K$ is the largest state in
the support of $p^{\ast }$. A simple inductive argument yields that $p^{\ast
}$ is equal to $p(K^{\ast },x)$ for some $K\in \mathbb{Z}_{+}\cup \{\infty
\} $ and $x^{\ast }\in (0,1]$. We say $p(K^{\ast },x^{\ast })$ is optimal in
that case.

Central to our analysis is the following function $\psi :\mathbb{Z}_{+}\cup
\{\infty \}\rightarrow \mathbb{R}$ defined as
\begin{equation*}
\psi (K)\triangleq \sum_{k=1}^{K}\dprod\limits_{\ell =1}^{k}\left( \frac{%
\lambda _{\ell -1}}{\mu _{\ell }}\right) \left[ \mu _{k}V-Ck\right] \text{.}
\end{equation*}
The following lemma states that this function is single-peaked under the regularity of the service process.

\begin{lemma}
\label{lem: w single peaked} Assume $\mu $ is regular. Function $\psi $ is
single-peaked.
\end{lemma}

\begin{proof}
Single-peakedness of $\psi$ is equivalent to: $\psi (K-1)\le(<)\psi (K)$ implies $\psi
(K^{\prime }-1)\le(<)\psi (K^{\prime })$ for all $K^{\prime }\leq K$.  We just prove the strict inequality version; the argument for the weak inequality version is identical.  Assume that $\psi (K-1)<\psi (K)
$. This is equivalent to
\begin{equation*}
\dprod\limits_{\ell =1}^{K}\left( \frac{\lambda _{\ell -1}}{\mu _{\ell }}%
\right) \left[ \mu _{K}V-CK\right] >0
\end{equation*}%
which in turn is equivalent to $\mu _{K}V-CK>0$. We claim that this implies $%
\mu _{K^{\prime }}V-CK^{\prime }>0$ for any $K^{\prime }\leq K$---which by
the above reasoning will imply $\psi (K^{\prime }-1)<\psi (K^{\prime })$.

First, by the regularity of $\mu $,
recall \cref{lem: single-peakedness} (with $\alpha =1$ and $\xi =0$) which proves that $h(K)\triangleq \mu _{K}V-CK$ is single-peaked.
Now, observe that function $h$ is equal to $0$ at $K=0$. Hence, if $h(K)=\mu _{K}V-CK>0$, given single-peakedness of $h$, we must have
$h(K^{\prime })=\mu _{K^{\prime }}V-CK^{\prime }>0$ for any $K^{\prime
}=1,...,K$, as claimed.
\end{proof}

We can now state our main proposition in this section.

\begin{proposition}
\label{prop: main 2} Assume $\mu $ is regular and $\alpha \in (0,1]$. Let $%
p(K^{\ast },x^{\ast })$ be the optimal policy. We have $K^{\ast }<\infty $.
\end{proposition}

Define function $g:\mathbb{Z}_{+}\rightarrow \mathbb{R}$ as
\begin{equation*}
g(K)\triangleq \left( 1-\alpha \right) \mu _{K}R+\alpha \left[ \mu _{K}V-CK%
\right] \text{.}
\end{equation*}

\cref{lem: single-peakedness} (with $\xi =0$) implies that function $g$ is
single-peaked when $\mu $ is regular. Define $\bar{K}_{2}\triangleq \sup
\{K^{\prime }:g(K^{\prime })\geq 0\}$, where $\bar{K}_{2} \triangleq \infty $
whenever $g(K^{\prime })\geq 0$ for all $K^{\prime }\in \mathbb{Z}_{+}$.
Observe that, when $\mu $ is regular, because $g$ is single-peaked (and
since $g(0)=0$), $g(K)\geq 0$ if and only if $K\leq \bar{K}_{2}$. We state
the following result.

\begin{lemma}
\label{lem: upper bound on queue} Assume $\mu $ is regular. Let $p(K^*,x^*)$
be the optimal policy. We have $K^*\leq \bar{K}_{2}$.
\end{lemma}

\begin{proof}
The value of the objective of $[P^{\prime }]$ is given by
\begin{equation*}
\sum_{k=1}^{K^*}p_{k}(K^*,x^*)g(k)=p_{0}(K^*,x^*)\sum_{k=1}^{K^*-1}\dprod\limits_{\ell
=1}^{k}\frac{\lambda _{\ell -1}}{\mu _{\ell }}g(k)+p_{0}(K^*,x^*)\dprod\limits_{%
\ell =1}^{K^*-1}\frac{\lambda _{\ell -1}}{\mu _{\ell }}\frac{\lambda _{K^*-1}x^*}{%
\mu _{K^*}}g(K^*)\text{.}
\end{equation*}%
By way of contradiction, assume that $K^*>\bar{K}_{2}$. As we already
stated, by definition of $\bar{K}_{2}$ and single-peakedness of $g$, it
must be that $g(K^{\prime })<0$ for $K^{\prime }=\bar{K}_{2}+1,...,K^*$.
Consider the distribution $p(\bar{K}_{2},1)$. Compared to $p(K^*,x^*)$, this
distribution removes all weight on negative values and, for each positive
value, increases its weight. This must strictly increase the value of the
objective.

Now, it remains to show that \cref{IR} is satisfied under $p(\bar{K}_{2},1)$.
Since $p(K^*,x^*)$ is optimal, \cref{IR} holds at $p(K^*,x^*)$:
\begin{equation*}
p_{0}(K^*,x^*)\sum_{k=1}^{K^*-1}\dprod\limits_{\ell =1}^{k}\left( \frac{\lambda
_{\ell -1}}{\mu _{\ell }}\right) h(k)+p_{0}(K^*,x^*)\dprod\limits_{\ell
=1}^{K^*-1}\left( \frac{\lambda _{\ell -1}}{\mu _{\ell }}\right) \left( \frac{%
\lambda _{K^*-1}x^*}{\mu _{K^*}}\right) h (K^*)\geq 0
\end{equation*}%
where $h(k) \triangleq \mu_k V-Ck$.
From this, it follows that $\psi (K^*)\geq 0$ if $h (K^*)\geq 0$ and that $%
\psi (K^*-1)\geq 0$ if $h (K^*)\leq 0$. Since, by \cref{lem: w single
peaked}, $\psi $ is single-peaked and $\psi (0)=0$, this implies that $\psi
(K^{\prime })\geq 0$ for any $K^{\prime }<K^*$. In particular, $\psi
(\bar{K}_{2})\geq 0$. This implies that \cref{IR} holds at $p(\bar{K}_{2},1)$. This contradicts our assumption that $p(K^*,x^*)$ is optimal.%
\end{proof}

\begin{proof}[Completion of the proof of \cref{prop: main 2}]
Given \cref{lem: upper bound on queue}, it is enough to prove that $%
\bar{K}_{2}<\infty $. This follows since $\alpha \in (0,1]$ and $\mu _{k}$ is uniformly
bounded, which imply $g(K)\to -\infty $ a $K \to \infty$. \end{proof}

\section{Formal Arguments for Dynamic Matching with Overloaded Lists}

\label{app: formal arg discussion sec}

In this section, we explain how our results can be obtained in the setting
with overloaded waiting-lists as proposed by \cite{leshno2019dynamic}.
Consider an infinite discrete time horizon model where at each period a
number of agents are waiting on a wait-list. Each period $t$ begins with
arrival of an item which can be of two types, either $A$ with probability $%
\mu _{A}$ or $B$ with probability $\mu _{B}=1-\mu _{A}$ independently across
periods. Period $t$ ends when the item is assigned to an agent. Agents can
be either type $\alpha $ or $\beta $ each with probability $\mu _{\alpha }$
and $\mu _{\beta }=1-\mu _{\alpha }$. Type $\alpha $ agents prefer $A$ items
to $B$ items while type $\beta $ agents prefer $B$ over $A$. Agents'
non-preferred item is referred to as a mismatched item. As in our main
setting, all agents are infinitely lived, risk neutral, and incur a common
linear waiting cost $C>0$ per period until they are assigned. Opting out of
the waiting-list is assumed to be equivalent to never getting assigned and
entails a utility of $-\infty $. An agent's value of being assigned an item
is $V$ if the item is his most preferred and value $0$ if assigned a
mismatched item. As we already mentioned in the main text, since agents
prefer to receive a mismatched item over never being assigned, taking a
mismatched item for an agent could simply be intrepreted as choosing an
outside option.

A mechanism decides at each date, to which agent waiting on the list is
assigned the arriving item. As in \cite{leshno2019dynamic}, we restrict our
attention to \textit{buffer-queue mechanisms} where a separate buffer-queue
is held for each item. $A$ ($B$) items arriving are assigned to the agents
waiting on the $A$ ($B$) buffer-queue if it is non-empty and to agents on
the waiting-list otherwise. Those, agents from the wait-list can either
accept the $A$ ($B$) item or refuse. If they refuse a $A$ ($B$) item, these
agents are identified as $\beta $ ($\alpha $) type agents and enter the $B$ (%
$A$) buffer-queue. The buffer-queue mechanism specifies the queueing
discipline (i.e., how to prioritize agents within the buffer-queue) and so
uses positions within the buffer-queue to decide who gets assigned the
arriving item. In addition, a buffer-queue mechanism specifies the maximum
number of agents in each buffer queue, say $K_{A}$ ($K_{B}$) for the $A$ ($B$%
) buffer-queue. Any buffer-queue mechanism induces a stochastic process over
the number of agents in each buffer-queue. Note that at each date, one of
the two buffer-queues must be empty. Hence, we can think of the state space
as $\{-K_{B},...,-1,0,1,...,K_{A}\}$ where $k\geq 0$ means that the $B$
buffer-queue is empty and that there are $k$ agents in the A buffer-queue.
The invariant distribution of this process is denoted by $%
p=(p_{-K_{B}},...,p_{-1},p_{0},p_{1},...,p_{K_{A}})$\ and characterized in
\cite{leshno2019dynamic}. The stochastic processes induced by buffer-queue
mechanisms are not birth-death processes. Hence, our results do not directly
apply to the environment under study. However, in the sequel, we explain how
these can be adapted.

The social planner's goal is to allocate items to maximize total utility. We
follow \cite{leshno2019dynamic} and assume that the waiting-list is
\textquotedblleft overloaded\textquotedblright , i.e., no mechanism will
ever exhaust the waiting-list. In this context, any allocation reduces total
waiting costs by the same amount. Hence, the social planner can ignore
waiting costs when comparing different allocations and his goal boils down
to minimizing misallocations under incentive constraints.

So far the setup is the same as \cite{leshno2019dynamic}. However, \cite%
{leshno2019dynamic} assumes that upon entering a buffer-queue, agents are
informed of the length of the buffer queue and, hence, perfectly know their
positions in that queue at all subsequent periods. We, however, depart from
the full information rule. We allow similar information policy as in
previous sections and, hence, impose similar obedience constraint, i.e., we
require that conditional on the information released to agents, any agent
recommended to join a buffer-queue or to stay in that queue must have an
incentive to follow that recommendation. Consistently with what we proved in
previous sections, we will show that FCFS with no information is optimal.
Under FCFS with no information, the obedience constraints for the two
buffer-queues can be written as follows. For agents recommended to enter the
$A$ buffer-queue:%
\begin{equation*}
V-C\sum_{\ell =1}^{K_{A}}\tilde{\gamma}_{\ell }^{t}\tau _{\ell }^{\ast }\geq
0,\forall t\geq 0
\end{equation*}%
where $\tilde{\gamma}_{\ell }^{t}$ stands for an agent's belief on having
position $\ell $ in the $A$ buffer-queue after spending $t$ periods on this
buffer-queue while $\tau _{\ell }^{\ast }$ is the expected waiting time
induced by the policy of an agent having position $\ell $ in the $A$
buffer-queue. A similar condition applies to agents recommended to join the $%
B$ buffer-queue.

Intuitively, if one wants to minimize misallocations, the problem of
deriving the optimal buffer-queue mechanism reduces to finding the maximal
size of an incentive-compatible buffer-queue mechanism. These maximal sizes $%
K_{A}^{\ast }$ and $K_{B}^{\ast }$, for buffer-queues $A$ and $B$
respectively, are identified in \cite{leshno2019dynamic}.\ In the sequel, we
show that under FCFS with no information, when the size of the buffer-queues
are set to these maximal sizes, obedience constraints are satisfied at all $%
t\geq 0$.

\begin{theorem}
\label{thm:leshno} Assume that the maximal sizes of buffer-queues are given
by $K_{A}^{\ast }$ and $K_{B}^{\ast }$. FCFS with no information satisfies
the obedience constraints.
\end{theorem}

In the full information context, \cite{leshno2019dynamic} proves that FCFS\
is not optimal among incentive compatible buffer-queue mechanisms and,
further, that it can be dominated by SIRO. SIRO is not incentive compatible
when the maximal sizes of buffer-queues are set to $K_{A}^{\ast }$ and $%
K_{B}^{\ast }$.\footnote{%
SIRO is not optimal in general but it is optimal within belief-free
incentive compatible mechanisms. We observe that, trivially, FCFS is
belief-free incentive compatible in \cite{leshno2019dynamic}'s terminology.}
Hence, \Cref{thm:leshno} not only shows that FCFS becomes optimal with a no
information policy but it also shows that FCFS under no information
outperforms SIRO under full information. Further, one can show that, under
the no information policy, SIRO may violate obedience constraints for $t>0$.

We prove that FCFS with these sizes of buffer-queues is incentive compatible
under the no information policy. Our argument parallels that of %
\Cref{thm:dyn-fcfs}. Indeed, we first prove that the obedience constraints
hold at $t=0$. In a second step, we show that the likelihood ratio of
beliefs about being in queue position $\ell $ cersus being in queue position
$\ell -1 $ after spending time $t$ on the queue.declines as $t$ increases,
meaning one's belief about getting served improves over time under FCFS with
no information proving that the obedience constraints hold at all $t\geq 0$.

In order to prove both of these results, the lemma below stating that the
likelihood ratio $r_{\ell }^{0}$ does not depend on $\ell $ is helpful. \cite%
{leshno2019dynamic} defines an extended Markov chain to describe the
evolution of the number of agents in the buffer-queues both across periods
and within a period and characterizes the invariant distribution of the
process. Upon entering the $A$ buffer-queue at a date $t$, an agent knows
that a $B$ item arrived and holds some beliefs over the number of agents who
are ahead of him (including those who entered before him at the current date
$t$). Conditional on a $B$ item arriving, \cite{leshno2019dynamic}'s
characterization states that the ratio of the likelihood of having $\ell $
agents over the likelihood of having $\ell -1$ agents in the $A$-queue is a
constant equal to $\frac{\mu _{\alpha }}{\mu _{A}}$. This yields the lemma
below whose proof is provided for completeness.\footnote{%
We simply focus on the likelihood ratio of beliefs of agents entering the $A$%
-buffer queue. A symmetric argument holds for agents entering the $B$
buffer-queue.}

\begin{lemma}
\label{lem:constant_lr} Under FCFS with no information, we have $r_{\ell
}^{0}=\frac{\tilde{\gamma}_{\ell }^{0}}{\tilde{\gamma}_{\ell -1}^{0}}=\frac{%
\mu _{\alpha }}{\mu _{A}}$ for all $\ell =2,...,K_{A}$.
\end{lemma}

\begin{proof}
Recall that our goal here is to show that the likelihood ratio $r^{\ell}$ of
beliefs of agents entering the $A$-buffer queue does not depend on $\ell$. A
symmetric argument clearly holds for agents entering into the $B$-buffer
queue.

Let $\mathbf{M}$ be the size of the main queue at $t=0$ (i.e., the date at
which the agent enters into the $A$ buffer-queue) and recall that there are
caps $K^{A}$ and $K^{B}$ on the $A$-buffer and $B$-buffer queues. For a
given $\alpha $-type agent, we compute the probability $\tilde{\gamma}_{\ell
}^{0}$ that he gets position $\ell $ in the $A$-buffer queue conditional on
the agent entering into the $A$-buffer queue,
\begin{equation*}
\tilde{\gamma}_{\ell }^{0}=\frac{\dsum\limits_{\ell ^{\prime }=0}^{\ell
}p_{\ell ^{\prime }}\mu _{B}\frac{1}{\mathbf{M}}\mu _{\alpha }^{\ell -\ell
^{\prime }}}{\dsum\limits_{\ell =1}^{K}\dsum\limits_{\ell ^{\prime
}=0}^{\ell }p_{\ell ^{\prime }}\mu _{B}\frac{1}{\mathbf{M}}\mu _{\alpha
}^{\ell -\ell ^{\prime }}}
\end{equation*}%
for all $\ell =1,...,K_{A}$. The probability that there are $\ell ^{\prime }$
agents in the $A$-queue at the begining of the period is $p_{\ell ^{\prime }}
$. Given that there are $\ell ^{\prime }$ agents in the $A$-queue, the event
that the agent enters in the $A$-queue and has position $\ell $ in this
buffer-queue corresponds to the joint event that (1) a $B$-item arrived
(which occurs with probability $\mu _{B}$), (2) the agents has queue
position exactly $\ell -\ell ^{\prime }$ in the main queue and all agents
ahead of him in the $A$-queue and himself are of $\alpha $-type (which
occurs with probability $(1/\mathbf{M})\times \mu _{\alpha }^{\ell -\ell
^{\prime }}$).\footnote{%
Note that these agents have never been offered any item. Otherwise, they
would have been matched or they would be in a buffer queue. So, wlog, we can
consider that we are drawing their types and positions only at the current
period. Further, we are assuming that each agent believes that the position
he holds in the main queue is a uniform draw over all possible positions.}
This explains the numerator. The denominator is simply the sum over $\ell $
of probabilities in the numerator.

Hence,%
\begin{equation*}
\tilde{\gamma}_{\ell }^{0}=\frac{\mu _{B}\frac{1}{\mathbf{M}}\left( p_{0}\mu
_{\alpha }^{\ell }+\dsum\limits_{\ell ^{\prime }=1}^{\ell }p_{\ell ^{\prime
}}\mu _{\alpha }^{\ell -\ell ^{\prime }}\right) }{\dsum\limits_{\ell
=1}^{K}\dsum\limits_{\ell ^{\prime }=0}^{\ell }p_{\ell ^{\prime }}\mu _{B}%
\frac{1}{\mathbf{M}}\mu _{\alpha }^{\ell -\ell ^{\prime }}}=\frac{\mu _{B}%
\frac{1}{\mathbf{M}}p_{0}\left( \mu _{\alpha }^{\ell }+\dsum\limits_{\ell
^{\prime }=1}^{\ell }\left( \frac{\mu _{\alpha }}{\mu _{A}}\right) ^{\ell
^{\prime }}\mu _{B}\mu _{\alpha }^{\ell -\ell ^{\prime }}\right) }{%
\dsum\limits_{\ell =1}^{K}\dsum\limits_{\ell ^{\prime }=0}^{\ell }p_{\ell
^{\prime }}\mu _{B}\frac{1}{\mathbf{M}}\mu _{\alpha }^{\ell -\ell ^{\prime }}%
}
\end{equation*}%
where the second equality is obtained by Lemma 2 in \cite{leshno2019dynamic}
where it is proved that $p_{\ell }=\mu _{B}\left( \frac{\mu _{\alpha }}{\mu
_{A}}\right) ^{\ell }p_{0}$ for all $\ell =1,...,K_{A}$.

To prove the lemma, we note that for all $\ell =2,...,K_{A}$,
\begin{eqnarray*}
r_{\ell }^{0} &=&\frac{\tilde{\gamma}_{\ell }^{0}}{\tilde{\gamma}_{\ell
-1}^{0}} \\
&=&\frac{\mu _{\alpha }^{\ell }+\dsum\limits_{\ell ^{\prime }=1}^{\ell
}\left( \frac{\mu _{\alpha }}{\mu _{A}}\right) ^{\ell ^{\prime }}\mu _{B}\mu
_{\alpha }^{\ell -\ell ^{\prime }}}{\mu _{\alpha }^{\ell
-1}+\dsum\limits_{\ell ^{\prime }=1}^{\ell -1}\left( \frac{\mu _{\alpha }}{%
\mu _{A}}\right) ^{\ell ^{\prime }}\mu _{B}\mu _{\alpha }^{\ell -\ell
^{\prime }-1}} \\
&=&\frac{\mu _{\alpha }^{\ell }+\left( \frac{\mu _{\alpha }}{\mu _{A}}%
\right) \mu _{B}\mu _{\alpha }^{\ell -1}+\dsum\limits_{\ell ^{\prime
}=2}^{\ell }\left( \frac{\mu _{\alpha }}{\mu _{A}}\right) ^{\ell ^{\prime
}}\mu _{B}\mu _{\alpha }^{\ell -\ell ^{\prime }}}{\mu _{\alpha }^{\ell
-1}+\dsum\limits_{\ell ^{\prime }=1}^{\ell -1}\left( \frac{\mu _{\alpha }}{%
\mu _{A}}\right) ^{\ell ^{\prime }}\mu _{B}\mu _{\alpha }^{\ell -\ell
^{\prime }-1}} \\
&=&\frac{\mu _{\alpha }^{\ell -1}(\mu _{\alpha }+\left( \frac{\mu _{\alpha }%
}{\mu _{A}}\right) \mu _{B})+\dsum\limits_{\ell ^{\prime }=2}^{\ell }\left(
\frac{\mu _{\alpha }}{\mu _{A}}\right) ^{\ell ^{\prime }}\mu _{B}\mu
_{\alpha }^{\ell -\ell ^{\prime }}}{\mu _{\alpha }^{\ell
-1}+\dsum\limits_{\ell ^{\prime }=1}^{\ell -1}\left( \frac{\mu _{\alpha }}{%
\mu _{A}}\right) ^{\ell ^{\prime }}\mu _{B}\mu _{\alpha }^{\ell -\ell
^{\prime }-1}} \\
&=&\frac{\mu _{\alpha }^{\ell -1}\frac{\mu _{\alpha }}{\mu _{A}}(\mu
_{A}+\mu _{B})+\dsum\limits_{\ell ^{\prime }=2}^{\ell }\left( \frac{\mu
_{\alpha }}{\mu _{A}}\right) ^{\ell ^{\prime }}\mu _{B}\mu _{\alpha }^{\ell
-\ell ^{\prime }}}{\mu _{\alpha }^{\ell -1}+\dsum\limits_{\ell ^{\prime
}=1}^{\ell -1}\left( \frac{\mu _{\alpha }}{\mu _{A}}\right) ^{\ell ^{\prime
}}\mu _{B}\mu _{\alpha }^{\ell -\ell ^{\prime }-1}} \\
&=&\frac{\mu _{\alpha }}{\mu _{A}}\frac{\mu _{\alpha }^{\ell
-1}+\dsum\limits_{\ell ^{\prime }=2}^{\ell }\left( \frac{\mu _{\alpha }}{\mu
_{A}}\right) ^{\ell ^{\prime }-1}\mu _{B}\mu _{\alpha }^{\ell -\ell ^{\prime
}}}{\mu _{\alpha }^{\ell -1}+\dsum\limits_{\ell ^{\prime }=1}^{\ell
-1}\left( \frac{\mu _{\alpha }}{\mu _{A}}\right) ^{\ell ^{\prime }}\mu
_{B}\mu _{\alpha }^{\ell -\ell ^{\prime }-1}} \\
&=&\frac{\mu _{\alpha }}{\mu _{A}}\frac{\mu _{\alpha }^{\ell
-1}+\dsum\limits_{\ell ^{\prime }=1}^{\ell -1}\left( \frac{\mu _{\alpha }}{%
\mu _{A}}\right) ^{\ell ^{\prime }}\mu _{B}\mu _{\alpha }^{\ell -\ell
^{\prime }-1}}{\mu _{\alpha }^{\ell -1}+\dsum\limits_{\ell ^{\prime
}=1}^{\ell -1}\left( \frac{\mu _{\alpha }}{\mu _{A}}\right) ^{\ell ^{\prime
}}\mu _{B}\mu _{\alpha }^{\ell -\ell ^{\prime }-1}}=\frac{\mu _{\alpha }}{%
\mu _{A}}\text{.}
\end{eqnarray*}
\end{proof}

Now, we are in a position to prove the first step of our argument, i.e.,
that obedience constraints hold at $t=0$. (Recall that the maximal sizes of
the buffer-queues compatible with the obedience constraints are denoted $%
K_{A}^{\ast }$ and $K_{B}^{\ast }$, for buffer-queues $A$ and $B$
respectively).

\begin{proposition}
Assume that the maximal sizes of buffer-queues are given by $K_{A}^{\ast }$
and $K_{B}^{\ast }$. FCFS with no information satisfies the obedience
constraints at $t=0$.
\end{proposition}

\begin{proof}
In the sequel, we simply focus on the obedience constraint at $t=0$ for
agents recommended to join the $A$ buffer-queue. A symmetric argument holds
for the other obedience constraint. The maximal sizes of the buffer-queues
satisfying the obedience constraints are identified in \cite%
{leshno2019dynamic}. In case $\mu _{\alpha }=\mu _{A}$, it is equal to $%
\left\lfloor 2\mu _{A}\frac{V}{C}\right\rfloor -1\equiv K_{A}^{\ast }$ for
the $A$ buffer-queue. In case $\mu _{\alpha }\neq \mu _{A}$, $K_{A}^{\ast }$
is equal to $\sup \{K\in \mathbb{Z}_+\mid K+\frac{\mu _{A}}{\mu _{A}-\mu
_{\alpha }}+\frac{K}{\left( \frac{\mu _{\alpha }}{\mu _{A}}\right) ^{K}-1}%
\leq \frac{V}{C}\mu _{A}\}$.

First, it is clear that $\tau _{\ell }^{\ast }=\ell /\mu _{A}$. Indeed, an
agent in position $\ell $ in the $A$-buffer-queue will have to wait the
arrival of $\ell $ items $A$ to get matched. Since, at each date, the
likelihood that an item $A$ arrives is $\mu _{A}$, the expected waiting time
for this agent must be $\ell /\mu _{A}$.

\underline{Case 1: $\mu _{\alpha }=\mu _{A}$.} \Cref{lem:constant_lr}
implies that $\tilde{\gamma}_{\ell }^{0}=\frac{1}{K_{A}^{\ast }}$ for all $%
\ell =1,...,K_{A}^{\ast }$. Hence, we can rewrite the obedience constraint
at $t=0$ as follows%
\begin{eqnarray*}
V-C\sum_{\ell =1}^{K_{A}^{\ast }}\tilde{\gamma}_{\ell }^{0}\tau _{\ell
}^{\ast } &=&V-C\frac{1}{K_{A}^{\ast }}\sum_{\ell =1}^{K_{A}^{\ast }}\frac{%
\ell }{\mu _{A}} \\
&=&V-\frac{1}{\mu _{A}}C\frac{K_{A}^{\ast }+1}{2}\geq 0\text{.}
\end{eqnarray*}%
Note that this inequality holds true if and only if $K_{A}^{\ast }\leq 2%
\frac{V}{C}\mu _{A}-1$. Since $K_{A}^{\ast }$ is an integer, this is
equivalent to $K_{A}^{\ast }\leq \left\lfloor 2\mu _{A}\frac{V}{C}%
\right\rfloor -1$ which holds by definition of $K_{A}^{\ast }$.

\underline{Case 2: $\mu _{\alpha }\neq \mu _{A}$.} \Cref%
{lem:constant_lr} implies that $\tilde{\gamma}_{\ell }^{0}=\tilde{\gamma}%
_{1}^{0}\left( \frac{\mu _{\alpha }}{\mu _{A}}\right) ^{\ell -1}$\ for all $%
\ell =1,...,K_{A}^{\ast }$ and $\tilde{\gamma}_{1}^{0}=\frac{1}{\sum_{\ell
=1}^{K_{A}^{\ast }}\left( \frac{\mu _{\alpha }}{\mu _{A}}\right) ^{\ell -1}}$%
. Hence, we can rewrite the obedience constraint at $t=0$ as follows%
\begin{eqnarray*}
V-C\sum_{\ell =1}^{K_{A}^{\ast }}\tilde{\gamma}_{\ell }^{0}\tau _{\ell
}^{\ast } &=&V-C\frac{1}{\mu _{A}}\frac{\sum_{\ell =1}^{K_{A}^{\ast }}\left(
\frac{\mu _{\alpha }}{\mu _{A}}\right) ^{\ell -1}\ell }{\sum_{\ell
=1}^{K_{A}^{\ast }}\left( \frac{\mu _{\alpha }}{\mu _{A}}\right) ^{\ell -1}}
\\
&=&V-C\frac{1}{\mu _{A}}\frac{\sum_{\ell =1}^{K_{A}^{\ast }}\left( \frac{\mu
_{\alpha }}{\mu _{A}}\right) ^{\ell }\ell }{\sum_{\ell =1}^{K_{A}^{\ast
}}\left( \frac{\mu _{\alpha }}{\mu _{A}}\right) ^{\ell }} \\
&=&1-C\frac{1}{\mu _{A}}\frac{\mu _{A}}{\mu _{A}-\mu _{\alpha }}\frac{\mu
_{\alpha }}{\mu _{A}}\frac{1-(K_{A}^{\ast }+1)\left( \frac{\mu _{\alpha }}{%
\mu _{A}}\right) ^{K_{A}^{\ast }}+K_{A}^{\ast }\left( \frac{\mu _{\alpha }}{%
\mu _{A}}\right) ^{K_{A}^{\ast }+1}}{\left( \frac{\mu _{\alpha }}{\mu _{A}}%
\right) -\left( \frac{\mu _{\alpha }}{\mu _{A}}\right) ^{K_{A}^{\ast }+1}} \\
&=&V-C\frac{1}{\mu _{A}}[K_{A}^{\ast }+\frac{\mu _{A}}{\mu _{A}-\mu _{\alpha
}}+\frac{K_{A}^{\ast }}{\left( \frac{\mu _{\alpha }}{\mu _{A}}\right)
^{K_{A}^{\ast }}-1}]\geq 0
\end{eqnarray*}%
where the third equality uses basic properties of power series. Note that
this inequality holds true since $K_{A}^{\ast }=\sup \{K\in \mathbb{Z}_+\mid K+%
\frac{\mu _{A}}{\mu _{A}-\mu _{\alpha }}+\frac{K}{\left( \frac{\mu _{\alpha }%
}{\mu _{A}}\right) ^{K}-1}\leq \frac{V}{C}\mu _{A}\}$ and since, as proved
in \cite{leshno2019dynamic}, the function $K\mapsto K+\frac{\mu _{A}}{\mu
_{A}-\mu _{\alpha }}+\frac{K}{\left( \frac{\mu _{\alpha }}{\mu _{A}}\right)
^{K}-1}$ is monotonically increasing in $K$ and goes to infinity when $K$
grows large.

\end{proof}

Finally, we need to show that the obedience constraints hold for all $t>0$.
In order to do so, we simply show that the likelihood ratio decreases over
time.

\begin{proposition}
$r_{\ell }^{t}\leq r_{\ell }^{0}$ for all $t\geq 0$ for all $\ell \in
\{2,...,K_{A}\}$.
\end{proposition}

\begin{proof}
Note that one can write $\{\tilde{\gamma}_{\ell }^{t+1}\}_{\ell }$ as a
function of $\{\tilde{\gamma}_{\ell }^{t}\}_{\ell }$ as follows%
\begin{equation*}
\tilde{\gamma}_{\ell }^{t+1}=\frac{\tilde{\gamma}_{\ell }^{t}\mu _{B}+\tilde{%
\gamma}_{\ell +1}^{t}\mu _{A}}{\tilde{\gamma}_{1}^{t}\mu _{B}+\sum_{i=2}^{K}%
\tilde{\gamma}_{i}^{t}}
\end{equation*}%
for $\ell =1,...,K_{A}$ where we recall that $\tilde{\gamma}_{K_{A}+1}^{t}=0$%
. Indeed, the numerator is the probability that the agent's queue position
is $\ell $ after staying in the queue for $t+1$ periods. This event occurs
either if (1) the agent has already $\ell -1$ agents ahead of him in the
queue at time $t$ and none of them as well as himself are served at time $t$%
; or (2) if there are $\ell $ agents ahead of him at $t$ and at least one
agent ahead of him is served at $t$. The denominator in turn gives the
probability that the agent has not been served by time $t+1$. Hence, given
that an agent has not been served when period $t+1$ starts, the above
expression gives the conditional belief that his position in the queue is $%
\ell $ when period $t+1$ starts.

Thus, we obtain%
\begin{eqnarray*}
r_{\ell }^{t+1} &=&\frac{\tilde{\gamma}_{\ell }^{t+1}}{\tilde{\gamma}_{\ell
-1}^{t+1}} \\
&=&\frac{\tilde{\gamma}_{\ell }^{t}\mu _{B}+\tilde{\gamma}_{\ell +1}^{t}\mu
_{A}}{\tilde{\gamma}_{\ell -1}^{t}\mu _{B}+\tilde{\gamma}_{\ell }^{t}\mu _{A}%
} \\
&=&\frac{\mu _{B}+r_{\ell +1}^{t}\mu _{A}}{\frac{1}{r_{\ell }^{t}}\mu
_{B}+\mu _{A}}
\end{eqnarray*}%
for $\ell =2,...,K_{A}$ where we recall that $r_{K_{A}+1}^{t}=0$. Hence, we
have a mapping from $r^{t}\equiv \{r_{\ell }^{t}\}_{\ell }$ to $%
r^{t+1}\equiv \{r_{\ell }^{t+1}\}_{\ell }$. Clearly, this mapping is
increasing (in the product order). Thus, if we can show that $r_{\ell
}^{1}\leq r_{\ell }^{0}$ for $\ell =2,...,K_{A}$, then the sequence $%
\{r_{\ell }^{t}\}_{t}$ will be decreasing for each $\ell =2,...,K_{A}$, and
so the proof will be complete.

Fix any $\ell =2,...,K_{A}$, we want to show that $r_{\ell }^{1}\leq r_{\ell
}^{0}$, i.e.,
\begin{equation*}
\frac{\mu _{B}+r_{\ell +1}^{0}\mu _{A}}{\frac{1}{r_{\ell }^{0}}\mu _{B}+\mu
_{A}}\leq r_{\ell }^{0}
\end{equation*}%
which in turn is equivalent to
\begin{equation*}
r_{\ell +1}^{0}\leq r_{\ell }^{0}\text{.}
\end{equation*}%
The above holds true by \Cref{lem:constant_lr} (we recall that $%
r_{K_{A}+1}^{0}=0$). \end{proof}


\end{document}